\documentclass[a4paper,11pt]{article}
\usepackage[bookmarks,colorlinks,breaklinks]{hyperref}  
\hypersetup{linkcolor=blue,citecolor=blue,filecolor=dullmagenta,urlcolor=blue} 
\usepackage{amsmath, amssymb, amsbsy, amsthm,slashed, simplewick}
\usepackage{epsfig}
\usepackage{axodraw}
\usepackage{tabularx}

\numberwithin{equation}{section}

\newtheorem{lemma}[equation]{Lemma}

\newtheorem{cor}[equation]{Corollary}
\newtheorem{proposition}[equation]{Proposition}
\theoremstyle{definition}

\theoremstyle{remark}

\usepackage{graphicx}
\usepackage{slashed} 
\usepackage{comment}
\usepackage{bbm}

\usepackage[utf8]{inputenc}
\usepackage{dsfont}

\def\beq{\begin{equation}}
\def\eeq{\end{equation}}
\def\bea{\begin{eqnarray}}
\def\eea{\end{eqnarray}}

\newcommand{\Tr}{\,{\rm Tr}\,}
\def\={\ =\ }

\newcommand{\dd}{\ensuremath{\mathrm{d}}}
\newcommand{\e}{\ensuremath{\,\mathrm{e}\,}}
\renewcommand{\i}{\ensuremath{\,\mathrm{i}\,}}
\newcommand{\N}{\ensuremath{\mathbb{N}}}
\newcommand{\R}{\ensuremath{\mathbb{R}}}
\renewcommand{\C}{\ensuremath{\mathbb{C}}}
\newcommand{\Z}{\ensuremath{\mathbb{Z}}}
\newcommand{\K}{\ensuremath{\mathbb{K}}}

\newcommand{\grad}{\ensuremath{\mathrm{\grad}\,}}

\newcommand{\p}[1]{\partial_{#1}}

\newcommand{\pdm}[2]{\frac{\partial^{#1}}{\partial #2^{#1}}}

\newcommand{\ket}[1]{|#1\rangle}
\newcommand{\bra}[1]{\langle #1|}
\newcommand{\bk}[2]{\langle #1|#2\rangle}
\newcommand{\kb}[2]{|#1\rangle\langle#2|}

\newcommand{\s}{\text{s}}

\newcommand{\f}{\mathcal{F}}

\renewcommand{\Im}{\ensuremath{\mathfrak{Im}}}
\renewcommand{\Re}{\ensuremath{\mathfrak{Re}}}

\renewcommand{\d}{\ensuremath{{\sf D}}}

\renewcommand{\v}{\ensuremath{\vech{V}}}
\newcommand{\h}{\ensuremath{\vech{\sf H}}}

\newcommand{\half}{\ensuremath{\frac{1}{2}}}

\newcommand{\eq}{\begin{eqnarray}}
\newcommand{\eqend}{\end{eqnarray}}

\def\vec#1{\mathchoice{\mbox{\boldmath$\displaystyle#1$}}{\mbox{\boldmath$\textstyle#1$}}{\mbox{\boldmath$\scriptstyle#1$}}{\mbox{\boldmath$\scriptscriptstyle#1$}}}
\newcommand{\vech}[1]{\ensuremath{\vec{\hat{\sf{#1}}}}}

\newcommand{\we}[1]{\ensuremath{{\sf{W}}^{-1}\left[#1\right]}}
\newcommand{\wee}[1]{\ensuremath{{\sf{W}}^{-1}[#1]}}
\newcommand{\wii}[1]{\ensuremath{{\sf W}}[#1]}
\newcommand{\wi}[1]{\ensuremath{{\sf W}}\left[#1\right]}

\newcommand{\vt}{\ensuremath{\vartheta}}
\newcommand{\so}{\ensuremath{\mathcal{S}_0}}

\renewcommand{\s}{\ensuremath{\mathcal{S}}}
\renewcommand{\P}{\ensuremath{\textsf P}}
\newcommand{\Pt}{\ensuremath{\tilde{\textsf P}{}^2}}
\newcommand{\Pii}{\ensuremath{\textsf P}{}_i^2}
\newcommand{\Pmm}{\ensuremath{\textsf P}{}_{\hspace{-0.045 cm}\mu}^2}
\newcommand{\Pti}{\ensuremath{\tilde{\textsf P}{}_i^2}}
\newcommand{\Ptm}{\ensuremath{\tilde{\textsf P}{}_{\hspace{-0.045 cm}\mu}^2}}
\renewcommand{\K}{\ensuremath{\textsf K}}
\newcommand{\Kt}{\ensuremath{\tilde{\textsf K}{}^2}}
\newcommand{\Kii}{\ensuremath{\textsf K}{}_i^2}
\newcommand{\Kmm}{\ensuremath{\textsf K}{}_{\hspace{-0.045 cm}\mu}^2}
\newcommand{\Kti}{\ensuremath{\tilde{\textsf K}{}_i^2}}
\newcommand{\Ktm}{\ensuremath{\tilde{\textsf K}{}_{\hspace{-0.045 cm}\mu}^2}}

\title{Propagators and Matrix Basis \\ on Noncommutative Minkowski
  Space\footnote{Report numbers: \ ITP--UH--04/11 \ , \ HWM--11--20 \ , \ EMPG--11--19}}

\author{Andr\'e Fischer\thanks{Email: {\tt 
   afischer3@kpmg.com}}\\[0.2cm]
 {\normalsize\slshape Institut f\"ur Theoretische Physik}\\[-0.1cm]
 {\normalsize\slshape Leibniz Universit\"at Hannover}\\[-0.1cm]
 {\normalsize\slshape Appelstra\ss e 2, D-30167 Hannover,
   Germany}\\[+0.5cm] 
Richard J. Szabo\footnote{Email: {\tt 
   R.J.Szabo@ma.hw.ac.uk}}\\[0.2cm]
 {\normalsize\slshape Department of Mathematics}\\[-0.1cm]
 {\normalsize\slshape Heriot-Watt University}\\[-0.1cm]
 {\normalsize\slshape Colin Maclaurin Building, Riccarton, Edinburgh
   EH14 4AS, U.K.}\\[-0.1cm]
 {\normalsize and}\\[-0.1cm]
 {\normalsize\slshape Maxwell Institute
   for Mathematical Sciences, Edinburgh, U.K.}
}

\begin{document}

\maketitle

\begin{abstract}
We describe an analytic continuation
of the Euclidean Grosse-Wulkenhaar and LSZ models which  
defines a one-parameter family of duality covariant noncommutative
field theories interpolating between Euclidean and Minkowski space
versions of these models, and provides an
alternative regularization to the usual Feynman prescription. This
regularization allows for a matrix model representation of the field
theories in
terms of a complex generalization of the usual basis of Landau
wavefunctions. The corresponding propagators are calculated and
identified with the Feynman propagators of the field theories. The
regulated quantum field theories are shown to be UV/IR-duality
covariant. We study the asymptotics of the regularized propagators in
position and matrix space representations, and confirm that they
generically possess a comparably good decay behaviour as in the Euclidean case.
\end{abstract}


\section{Introduction and Summary\label{Intro}}

This paper is devoted to an in-depth study of the perturbative
properties and renormalizability of noncommutative $\phi^{\star 4}$-type scalar field theories on real vector
spaces subjected to a Moyal deformation. The vast majority of the
literature on this subject has been devoted to Euclidean quantum field
theory. The most prominent feature of these models is the notorious mixing of
ultraviolet and infrared modes, which renders the noncommutative
$\phi^{\star n}$ field theories nonrenormalizable~\cite{mvrs00} (see
e.g.~\cite{sza01} for a review). Grosse and Wulkenhaar demonstrated
how to obtain field theories which are renormalizable to
all orders in perturbation theory by extending the kinetic
term of the $\phi^{\star 4}$ Lagrangian by an additional harmonic
oscillator potential~\cite{gw03,gw05}. The Grosse-Wulkenhaar
model also has vanishing beta-functions and
its perturbation series is likely to be Borel summable~\cite{riv07};
in two dimensions this has been recently established
in~\cite{wang11}. In four dimensions, the Euclidean Grosse-Wulkenhaar model is the first rigorous four-dimensional quantum field theory without unnatural cutoff which is expected to exist non-perturbatively and is not asymptotically free. 

The continuation of the Euclidean models to noncommutative
Minkowski space is presently an open problem which is plagued by both
conceptual and technical difficulties. The original problems were
unveiled in~\cite{gm00}, where it was found that
the standard perturbative expansion in terms of Feynman diagrams leads
to a violation of unitarity if space and {time} do not commute. As
subsequently pointed out in~\cite{BDFP02}, this due to the failure of
Wick's theorem, which does not apply to
non-local interactions in general. By using canonical quantization in the Hamiltonian framework involving the Dyson series and
time-ordered products of the interaction Hamiltonian, the resulting
field theory is still unitary but no longer equivalent to the Lagrangian formulation of the
quantum field theory in the path integral framework. For models built on the Hamiltonian framework see
e.g.~\cite{dfr94a,bahns04,pia10}. Yet another inequivalent
perturbative approach is based on the
Yang-Feldman formalism~\cite{BDFP02,bahns04}, which also gives a
unitary noncommutative quantum field theory on Minkowski space with
time-like noncommutativity.\footnote{Still another approach is
  provided by the twist-deformation
formalism for noncommutative quantum field theory on Minkowski space which
is considered in~\cite{bal06,gl07,gl08,sol08}; here one first
quantizes the classical field theory before deforming spacetime, and the
free part of the quantum field theory also differs from its commutative counterpart.}

The UV/IR-mixing problem of the ordinary $\phi^{\star n}$ field theory
is absent in the Hamiltonian framework to lowest orders, and it has
long been an open question as to whether it exists at all. Only
recently has UV/IR-mixing been shown to still occur, albeit through a mechanism which is different from that of
the Euclidean setting with modified Feynman rules~\cite{bahns10}. It
has also been shown that UV/IR-mixing arises in the Yang-Feldman
formalism~\cite{zahn11}.

Since the perturbative setups in the
Hamiltonian and Yang-Feldman formalisms are quite
complicated, it would be desirable to have an equivalent Euclidean
path integral formalism which simplifies the combinatorial aspects of
perturbation theory. However, the relationship between the Euclidean and
Minkowski space theories when time and space do not commute is unclear. In~\cite{bahns09} it has been shown that the Euclidean
counterparts of the $n$-point functions for the Klein-Gordon theory on
noncommutative Minkowski space in the Hamiltonian formalism are not those which follow from the
standard Euclidean framework, but appear with {on-shell} twisting
factors involving only on-shell momenta.

In this paper we will pursue the other
direction of this correspondence, starting with specific field
theories in Euclidean space. In order to find noncommutative field
theories in hyperbolic signature which are free from UV/IR-mixing, we
construct Minkowski space counterparts of the Grosse-Wulkenhaar model
and its generalizations known as the LSZ
models~\cite{lsz03,lsz04}. These models differ from the standard
noncommutative $\phi^{\star 4}$ field theory by the introduction of an
external background field into the kinetic term of the Lagrangian,
making them covariant under the duality comprising Fourier
transformation plus a rescaling of the fields~\cite{ls02}. This
duality is believed to be related to the improved asymptotic
behaviours of the propagators, which suppresses the UV/IR-mixing. This
UV/IR duality may thus be responsible for the renormalizability of
these field theories. The Euclidean Grosse-Wulkenhaar and LSZ models
are defined via path integral quantization, which leads to a violation
of unitarity for the usual noncommutative field theories in Minkowski
space. Here we will be interested in the
renormalization properties of their hyperbolic counterparts; unitarity
of these quantum field theories will be addressed elsewhere.

In Euclidean space the introduction of an external field has the useful
additional effect that the corresponding wave operators have discrete spectra
and the models can be analysed with the help of a matrix
basis for expansion of fields; the countably infinite set of
eigenfunctions are the Landau wavefunctions which diagonalize the free
parts of the action. This basis defines a mapping of the duality
covariant field theories onto matrix models, which permits a simple
and natural regularization of the field theories while maintaining duality manifestly at quantum level. In this way Grosse and Wulkenhaar were able to prove the renormalizability of their model to all orders of perturbation theory. In addition, it has been used to solve the LSZ model exactly and demonstrate the vanishing of the beta-function. 

However, in passing to hyperbolic
signature, the background field, which is a magnetic field in
the Euclidean metric, now plays the role of an electric field. 
This yields a qualitative
change due to the work
done on the particles by the field. The electric
field accelerates and splits virtual dipole pairs leading to pair
production. This is reflected in the spectra of
the wave operators, which now have a continuous part and are unbounded from
below. 

In~\cite{zahn} the perturbative
expansion in terms of modified Feynman diagrams in the continuous eigenvalue
representation has been investigated. At one-loop order unusual
divergences arise which are very likely to be
non-renormalizable. Moreover, the retarded propagator in position
space is no longer a tempered distribution in general.

In~\cite{fs08} a
different approach has been investigated, where a set of {resonance}
states has been used to expand the field theory in a discrete set
of functions. In the following we will take yet another path, which is related the resonance expansion found
in~\cite{fs08}, but which potentially avoids the associated technical
problems. We will show that the Grosse-Wulkenhaar and LSZ models allow for well-defined analytic continuations to Minkowski
space with the help of a special regularization, that we call the
``$\vt$-regularization'', which is a suitable replacement for
Feynman's $\i\epsilon$-prescription. 

This approach may also avoid the strange divergences found in~\cite{zahn}. These divergences come from squares of Dirac
delta-functions which arise from undetermined loop integrations. 
They are not ultraviolet divergences in the usual sense, as they occur
before performing loop integrals, and they show up in every
$\phi^{\star n}$-theory with $n\geq3$ for graphs with an unbroken
internal line. Using the $\vt$-regularization instead of the usual
$\i\epsilon$ regularization, one gets a
discrete spectrum instead of a continuous spectrum, leading to Kronecker delta-functions and sums
rather than Dirac delta-distributions and integrals. This procedure
renders these diagrams finite, and at
the same time keeps the model duality covariant. 

The outline of this paper is as follows. We will define duality covariant quantum field theories
on Minkowski space based on the work~\cite{fs08,fs10}, and describe some of their renormalization
properties. In order to employ an expansion of the action functionals
in terms of the resonance states found in~\cite{fs08}, we will
regularize the models such that the resonances turn into genuine
eigenfunctions of the regularized wave operators; this new matrix
basis and the corresponding matrix model representations of the LSZ
and Grosse-Wulkenhaar models on Minkowski space are described in
detail in \S\ref{Mink}--\S\ref{matrix3}. These wave operators are
related by Weyl-Wigner correspondence to the complex harmonic
oscillator Hamiltonian, which interpolates between the ordinary and
inverted harmonic oscillator Hamiltonians, and thus between the
Euclidean and Minkowski space theories; this unifies both theories
into a one-parameter family of duality covariant noncommutative quantum field
theories. The Feynman graphs are analytic continuations of the
Euclidean diagrams. We show that this regularized matrix basis is a
bi-orthogonal system whose linear span is the space of
square-integrable functions. At the quantum level and in the limit of
vanishing electric background, this regularization turns into the
usual $\i\epsilon$-prescription. For the special case of a
Klein-Gordon theory in a constant external electric field, where the various
propagators are known, we recalculate in \S\ref{cont} the propagator using the complex
matrix basis and verify that the regularization leads to Feynman
propagators. This confirms the equivalence to the
$\i\epsilon$-prescription, and demonstrates that the
$\vt$-regularization is also connected to causality of the quantum field theory. Using the
$\vt$-regularization, we show in \S\ref{quantum1} that a cutoff can be
introduced that renders the duality covariant field theories finite at
every order of perturbation theory and at the same time imposes duality
covariance manifestly. In \S\ref{ren} we derive the propagators for
the regularized models which include the Euclidean space propagators
and the Minkowski space causal propagators as special cases; away from
the hyperbolic point our propagators have a good decay behaviour in
all directions and are singular at coincident positions. The
$\vt$-regularization turns out to improve their asymptotic behaviour
and may thus be crucial for the renormalization programme. However,
due to the oscillatory behaviours of the occuring integrands in
Minkowski space, the corresponding asymptotics are much more
difficult to derive than in the Euclidean case. For a special case of
the LSZ model we find that the exponential decay in the short
Euclidean space variables ceases if one goes over to Minkowski space,
but persists in a neighbourhood of the hyperbolic point in the
one-parameter family of field theories. The regularization thus gives
a means to control the decay behaviour of the propagators. The
applicability of the matrix basis in this context, however, is still
an open question; the detailed analysis of the removal of the matrix
regularization will not be addressed in this paper. As we discuss in
the following, the Minkowski limit of our one-parameter family is very
singular; see~\cite{thesis} for a detailed analysis of some of the
uncontrollable divergences which arise. In the following we will
simply regard the $\vt$-regularized field theories as the appropriate
well-defined analytic continuations of the Euclidean space models.

The derivations of propagators with the help of the matrix basis may
be compared to calculations using other methods, such as Schwinger's
proper time formalism~\cite{Schw51}, the ``sum over
solutions method''~\cite{fgs91}, or the eigenvalue method using the
continuous eigenbasis~\cite{ritus78}. Compared to the latter technique
the matrix basis involves only polynomials and sums instead of
complicated integral expressions, and thus brings along a huge
simplification. In \S\ref{regprop} the causal propagator for a massive
complex scalar field in four dimensions in the background of a
constant electric field in the $\vt$-regularization is computed. As a further demonstration of how the matrix basis can be
applied, we calculate the one-loop effective action of the Klein-Gordon theory in a
background electric field (see Appendix~\ref{effaction}). Finally, in
Lemma~\ref{massshelllem} we demonstrate that the $\vt$-regularization
can likewise be used regulate the standard mass-shell singularities in the
Feynman propagator for the free Klein-Gordon theory.
We propose that going beyond the case
of a constant background field might be possible using our alternative
regularization and the matrix basis, by perturbing varying field
configurations around a uniform background. This might help in probing
quantum electrodynamics in the non-perturbative regime (see
e.g.~\cite{ringwald01,heinzl08,dunne09,ilderton10}). We conclude that
the matrix basis may serve as a powerful computational tool in
simplifying some otherwise cumbersome calculations. 

\section{Covariant Relativistic Noncommutative Field Theory}\label{Mink}

In this section we will introduce the duality covariant models in Minkowski space. We show that it is possible to construct a well-defined matrix model representation of the corresponding quantum field theories through a suitable regularization, that we call \emph{$\vt$-regularization}, which is an alternative to the usual $\i\epsilon$-prescription. For this, we will use the Weyl-Wigner transformation to map the eigenvalue problem of the $\vt$-regularized wave operators to that of the \emph{complex} harmonic oscillator. 

\subsection{Formulation of the Duality Covariant Models\label{Formulation}}

We work in $D=2n$ spacetime dimensions with metric of signature $(1,-1,\ldots,-1)$. Tensors will be labelled by Greek indices $\mu,\nu,\ldots$ ranging from $0$ to $d=D-1$. Throughout we use the Einstein summation convention. For simplicity we will denote the hyperbolic norm square of vectors $\vec a=(a^\mu)$ as
\beq
\|\vec a\|_{\rm M}^2=a_0^2-a_1^2-\dots-a_d^2=a_\mu\, a^\mu=:a_\mu^2 \ . 
\label{Mnorm}\eeq
Euclidean space dimensions are labelled by Latin indices $i,j,\ldots$ ranging from $1$ to $D$, and norm squares of vectors $\vec a=(a^i)$ with respect to the $D$-dimensional Euclidean metric are denoted 
\beq
\|\vec a\|_{\rm E}^2= a_1^2+\dots+a_D^2=a_i\, a^i=:a_i^2 \ . 
\label{Enorm}\eeq
Position vectors are denoted $\vec x=(x^\mu)$, with derivatives
$\partial_\mu:=\frac\partial{\partial x^\mu}$; in two dimensions we
often write $\vec x=(t,x)$. The dual pairing between a covariant vector $\vec x=(x^\mu)\in\R^D$ and a contravariant vector $\vec k=(k_\mu)\in(\R^D)^*$ is written $\vec k\cdot \vec x=k_\mu\, x^\mu$.

The LSZ model is a complex $\phi_D^{\star4}$-theory defined by the action
\beq
\s_{\sf LSZ}=\so+\s_{\sf{int}}
\eeq
with
\begin{eqnarray}
\so&=&\int\,\dd^{D}\vec x \ \phi^*(\vec x)\left(\sigma\, \Kmm+(1-\sigma)\, \Ktm-\mu^2\right)\phi(\vec x) \ , \label{lszfree} \\[4pt]
\s_{\sf{int}}&=&-g\, \int\,\dd^{D}\vec x\ \Big(\alpha\,(\phi^*\star_\Theta\phi\star_\Theta \phi^*\star_\Theta \phi)(\vec x)+\beta\,(\phi^*\star_\Theta \phi^*\star_\Theta \phi\star_\Theta \phi)(\vec x)\Big) \ ,\label{lszmink1}
\end{eqnarray}
where $\sigma\in[0,1]$, $\alpha,\beta\in\R_+:=[0,\infty)$, and $\mu^2,g>0$ are the mass and coupling parameters. The generalized momentum operators $\K_\mu$ and generalized dual momentum operators $\tilde\K_\mu$ are given by
\begin{equation}
\K_\mu=\i\p{\mu}-F_{\mu\nu}\, x^\nu \qquad \mbox{and}\qquad\tilde\K_\mu=\i\p{\mu}+F_{\mu\nu}\, x^\nu\ ,
\end{equation}
and they obey the commutation relations
\begin{eqnarray}
[\K_\mu,\K_\nu]=2\i F_{\mu\nu}\ ,\qquad[\tilde\K_\mu,\tilde\K_\nu]=-2\i F_{\mu\nu} \qquad \mbox{and} \qquad [\K_\mu,\tilde\K_\nu]=0 \label{com3}\ .
\end{eqnarray}
The star-product of arbitrary Schwartz functions $f(\vec{x}),g(\vec{x})\in\s(\R^D)$ is given by
\eq
(f\star_\Theta g)(\vec x):=\frac{1}{\pi^D\,|\det\Theta|}\, \int\, \dd^D\vec y\ \int\, \dd^D\vec z\ f(\vec x+\vec y)\,g(\vec x+\vec z)\,\e^{-2\i\vec y\cdot\Theta^{-1}\vec z} =(g^*\star_\Theta f^*)(\vec x) \label{moyal1} 
\eqend
with respect to a constant, real-valued, antisymmetric and non-degenerate $D\times D$ deformation matrix $\Theta$.

The coordinate system is chosen such that $\Theta$ takes the canonical skew-diagonal form 
\begin{eqnarray}
\Theta=(\Theta^{\mu\nu})=
\begin{pmatrix}
0 & \theta_0 & & & & & 0 \\
-\theta_0 &0 & & & & & \\
 & & 0 & \theta_1 & & & \\
& & -\theta_1 & 0 & & & \\ 
& & & & \ddots & & \\
& & & & & 
0 & \theta_{n-1} \\
0 & & & & & -\theta_{n-1} &0
\end{pmatrix}
\label{multitheta0}
\end{eqnarray}
with $\theta_k>0$ for $k=0,1,\dots,\frac D2-1$. The constant electromagnetic field strength tensor $F_{\mu\nu}$ is likewise given by
\begin{equation}
F=(F_{\mu\nu})=\begin{pmatrix}0 & E & & & & & 0 \\
-E & 0 & & & & & \\
 & & 0 & B_1 & & & \\
& & -B_1 & 0 & & & \\
& & & & \ddots& & \\
& & & & & 0 & B_{n-1}\\
0 & & & & & -B_{n-1} & 0
\end{pmatrix} \ ,
\end{equation}
with $E,B_k>0$ and
\eq
E\,\theta_0=B_k\,\theta_k=2\Omega
\eqend
for $k=1,\dots,\frac D2-1$ and $0<\Omega\leq1$. We will sometimes regard
$\Theta$ and $F$ as invertible linear maps $\Theta:(\R^D)^*\to \R^D$
and $F:\R^D\to(\R^D)^*$. The LSZ models with $\sigma=1$, where only
the generalized momentum operator $\K_\mu$ appears, are called
``critical'' models, while those with $\Omega=1$, where the field
theory is invariant under the UV/IR duality between position and
momentum space representations, are called ``self-dual'' models.

We will solve the eigenvalue equation for the wave operator and relate it to that of the Euclidean case. For this, we note that the $2n$-dimensional wave operators break up into $n$ blocks with
\begin{equation}
\Kmm=\sum_{k=0}^{n-1}\, (\Pmm)_k\qquad \mbox{and} \qquad\Ktm=\sum_{k=0}^{n-1}\,(\Ptm)_k\ .
\label{minkparts}\end{equation}
The operators
\begin{equation}
\begin{aligned}
(\Pmm)_k&=(\p{2k}^2+\p{2k+1}^2)+2\i B_k\,(x^{2k+1}\,\p{2k}-\,x^{2k}\,\p{2k+1})-\,B_k^2\, (x_{2k}^2+x_{2k+1}^2)\ , \\[4pt]
(\Ptm)_k&=(\p{2k}^2+\p{2k+1}^2)-2\i B_k\,(x^{2k+1}\,\p{2k}-\,x^{2k}\,\p{2k+1})-\,B_k^2\, (x_{2k}^2+x_{2k+1}^2)
\end{aligned}\label{minkpart2}
\end{equation}
for $k=1,\ldots,n-1$ act on two-dimensional Euclidean Klein-Gordon fields in a constant external magnetic background of field strengths $\pm\, 2B_k$, respectively, while the operators
\begin{equation}
\begin{aligned}
(\Pmm)_0&=-(\p{0}^2-\p{1}^2)-2\i E\, (x^1\,\p{0}+\,x^0\,\p{1})-\,E^2\, (x_0^2-x_1^2) \ , \\[4pt]
(\Ptm)_0&=-(\p{0}^2-\p{1}^2)+2\i E\,(x^1\,\p{0}+\,x^0\,\p{1})-\,E^2\,(x_0^2-x_1^2)
\end{aligned}\label{minkpart1}
\end{equation}
act on 1+1-dimensional Klein-Gordon fields in a constant {electric} background with field strengths $\pm\, 2E$, respectively. Since all component operators $(\Pmm)_k$ and $(\Ptm)_k$ for $k=0,1,\ldots,n-1$ mutually commute, the diagonalization of the full wave operators amounts to diagonalizing each of their two-dimensional blocks independently. 

As is well-known the spectra of the operators \eqref{minkpart2} are discrete with corresponding eigenfunctions the {Landau wavefunctions} $f_{mn}^{(B)}$ satisfying
\begin{eqnarray}
(\Pmm)_k\,f_{mn}^{(B_k)}(\vec
x_k)&=& -4B_k\,\big(m+\mbox{$\frac12$}\big)\, f_{mn}^{(B_k)}(\vec x_k) \
, \notag \\[4pt] (\Ptm)_k\, f_{mn}^{(B_k)}(\vec x_k)&=& -4B_k\, \big(n+\mbox{$\frac12$}\big)\, f_{mn}^{(B_k)}(\vec x_k)
\end{eqnarray}
for $m,n\in\N_0$, where we write $\vec x_k=(x^{2k},x^{2k+1})\in\R^2$ for $k=0,1,\dots,n-1$. These functions are Wigner transformations of tensor products of harmonic oscillator number basis states
\begin{eqnarray}
f_{mn}^{(B)}(x,y)=\wi{\,\kb{m}{n}\,}(x,y) 
\end{eqnarray}
of frequency $\Omega$, where, in two dimensions with deformation parameter $\Theta^{01}=\theta$, the Wigner distribution function of a compact operator $\vech \rho$ on Fock space is the Schwartz function on $\R^2$ given by~\cite{sza01}
\begin{eqnarray}
\wi{\vech\rho}(x,y)=\int\, \dd k\ \e^{\i k\,y/\theta}\, \bra{x+k/2}\,\vech\rho\,\ket{x-k/2}\ .\label{wignertrafo}
\end{eqnarray}
On the other hand, the spectra of the operators \eqref{minkpart1} are continuous, with corresponding eigenfunctions given by Wigner transformations of tensor products of parabolic cylinder functions~\cite{fs08,zahn}, denoted $\chi_{pq}$ with $p,q\in\R$, which solve the eigenvalue equations
\begin{eqnarray}
(\Pmm)_0\,\chi_{pq}(\vec x_0)=4E\,p\,\chi_{pq}(\vec x_0)\qquad
\mbox{and} \qquad(\Ptm)_0 \,\chi_{pq}(\vec x_0)=4E\,q\,\chi_{pq}(\vec x_0)\ .
\end{eqnarray}

For generic $\sigma$ the free part of the LSZ action (\ref{lszfree}) can be rewritten as
\begin{eqnarray}
\so=\int\,\dd^D\vec x\ \phi^*(\vec x)\,\left(\Kmm\big|_{F\rightarrow\tilde F}+\Omega^2\,\tilde x_\mu^2-\mu^2\right)\,\phi(\vec x)\label{lsz2}
\end{eqnarray}
with $\tilde F=(2\sigma-1)\,F=(2\sigma-1)\, (F_{\mu\nu})$ and $\tilde x_\mu=2\Theta^{-1}_{\mu\nu}\, x^\mu$. The free action thus describes a massive complex scalar field coupled to a constant electromagnetic background and in an oscillator potential proportional to $\Omega^2\,\tilde x_\mu^2$. The Grosse-Wulkenhaar model in $D=2n$ spacetime dimensions is the LSZ model for $\sigma=\frac12$ and $\alpha=\beta=\frac12$ with real scalar fields. The action is thus
\eq
\s_{\sf{GW}} =\int\,\dd^{D}\vec x\ \half\, \phi(\vec x)\, \left(-\p{\mu}^2+\Omega^2\, \tilde x_\mu^2-\mu^2\right)\,\phi(\vec x) - g\, \int\,\dd^{D}\vec x\ (\phi\star_\Theta \phi\star_\Theta \phi\star_\Theta \phi)(\vec x)\ .
\eqend
The $D$-dimensional wave operator again reduces to a sum of $n-1$ Euclidean wave operators plus a two-dimensional wave operator in Minkowski signature
\begin{equation}
\mbox{$\half$}\, (\Pmm)_0+\mbox{$\half$}\, (\Ptm)_0-\mu^2=-(\p{0}^2-\p{1}^2)- \Omega^2\,(x_0^2-x_1^2)-\mu^2
\end{equation}
with frequency $\Omega=E\,\theta_0/2$. The main difference, besides the hyperbolic signature, is an extra minus sign in front of the $\Omega$-term. The corresponding wave operator is given by the Hamiltonian of a harmonic oscillator with imaginary frequency, known as the {inverted harmonic oscillator}.

The corresponding models on Euclidean space are defined in terms of
the wave operator $\sigma\,\Kii+(1-\sigma)\,\Kti+\mu^2$, which also
split up into $n$ blocks made up of the operators
\begin{eqnarray}
\Kii=\sum_{k=1}^n\, (\Pii)_k\qquad \mbox{and} \qquad \Kti=\sum_{k=1}^n\, (\Pti)_k \ .
\label{euclparts}\end{eqnarray}
After relabelling of coordinates, one can relate
$(\Pii)_k=-(\Pmm)_{k-1}$ and $(\Pti)_k=-(\Ptm)_{k-1}$ for
$k=2,\ldots,n$, whereas $(\Pii)_1$ and $(\Pti)_1$ are of the same form
as (\ref{minkpart2}). Thus in contrast to the mixed discrete and
continuous spectrum of the hyperbolic space wave operator, the
Euclidean case deals with purely discrete spectrum. This situation is
responsible for the powerful application of the matrix model
representation of Grosse and Wulkenhaar~\cite{gw03,gw03a,gw05}.

The duality covariant field theories involve two parameters $\Theta$
and $F$. In the commutative limit $\Theta=0$, one recovers the field
theory for an interacting scalar field in a constant electromagnetic
background; in \S\ref{regprop} and Appendix~\ref{effaction} we
demonstrate how to reproduce the known standard results in the
literature using the novel regularization we propose below. In the
vanishing background limit $F=0$, we recover the usual
$\phi^4$-theories on noncommutative Minkowski space together with
their UV/IR mixing problems as discussed in \S\ref{Intro}; in
Lemma~\ref{massshelllem} we illustrate how our regularization is
applicable in this case as well. Neither of these two limits possess
duality covariance. In the self-dual limit $F=(2\Theta)^{-1}$ the
field theory is duality invariant; the matrix representation
we obtain below at the self-dual point makes no sense in the limit $F=0$.

\subsection{Spectral Decomposition and $\vt$-Regularization\label{Spectral}}

The external electromagnetic background will be treated by considering all terms quadratic in the fields as being part of the free action. Then the path integral quantization gives the usual (modified) Feynman diagrams but with the dressed propagator for the scalar field moving in this background. It is a feature of most field theories defined on hyperbolic space that there is more than one propagator, i.e. a distribution whose kernel $\Delta(\vec x,\vec y)$ solves the partial differential equation $\d_{\vec x}\Delta(\vec x,\vec y)=\delta(\vec x-\vec y)$ with $\d_{\vec x}$ the wave operator of the field theory. This is due to the occurence of zero eigenvalues of $\d_{\vec x}$, which prevents the naive inversion of the operator to give a propagator. It is therefore necessary to impose further conditions so as to make the solution of this problem unique. This may be done either by imposing boundary conditions, by postulating a spectral representation, or by extending the wave operator so as to make the solution of the partial differential equation unique. 

For the ordinary scalar field theory, the $\i\epsilon$-prescription is a method to single out a specific propagator, namely the Feynman propagator. In the commutative field theories, this prescription enhances the action by an additional term $\i\epsilon\,\int\, \phi^2$ which for $\epsilon>0$ ensures the required asymptotic damping of the integrand in the partition function at $|\phi|\rightarrow\infty$ (rather than an oscillatory behaviour), and at the same time regularizes the singularity of the free propagator and furthermore imposes causality. In this particular case it is also the infinitesimal version of the Wick rotation to Euclidean space $t\mapsto\e^{\i\epsilon}\,t$.

However, in our case the field theories defined on the two different spacetimes are not related by this \emph{ordinary} Wick rotation -- it has to be accompanied by an additional transformation $E\mapsto\pm\i B$. This is not surprising, since the model can be viewed as a field theory on a curved non-stationary spacetime, for which this is a generic feature~\cite{dewitt75}. Another characteristic of those field theories is that the multitude of different equivalent definitions of the Feynman propagator is resolved~\cite{candelas77}; we return to this point in \S\ref{causprop}. Since we are interested in an analytic continuation of the Euclidean space models in a path integral framework, it is the propagator we obtain by this transformation that we are concerned with.
The extra transformation of the magnetic field strength is also in harmony with the fact that in order to ensure the commutation relation
$[x^0,x^i]=\i\Theta^{0i}$
for both Euclidean and Minkowski space, the deformation parameter $\Theta^{0i}$ has to transform accordingly to compensate the phase coming from the Wick rotation. For the duality invariant field theories, i.e. at the self-dual point $\Omega=1$, the deformation matrix is proportional to the field strength tensor, which in turn implies a rotation of the field strength.

In \cite{fs08} it was shown that the actions of the Minkowski and Euclidean space wave operators, $(\Pmm)_0$ and $(\Pii)_1$, can be represented as a star-product with a classical Hamiltonian at the self-dual point $\Omega=1$. To compare to the Euclidean version, we have to identify $B= E$ and the ordered coordinate pairs $\vec x=(x^1,x^2)=(t,x)$ to find\footnote{These identities are taken in the multiplier algebra corresponding to the Schwartz space $\s(\R^2)$.}
\begin{equation}
(\Pii)_1\,f(\vec x)=E^2\,(x^2+ t^2)\star_{2/E} f(\vec x)\qquad \mbox{and} \qquad
(\Pmm)_0\,f(\vec x)=E^2\,(x^2- t^2)\star_{2/E} f(\vec x)
\end{equation}
and likewise
\begin{equation}
(\Pti)_1\,f(\vec x)= f(\vec x)\star_{2/E}E^2\,(x^2+ t^2)\qquad \mbox{and} \qquad
(\Ptm)_0\,f(\vec x)
= f(\vec x)\star_{2/E}E^2\,(x^2- t^2)\ ,
\end{equation}
which can be verified by explicitly writing out the individual terms
\begin{equation}
\begin{aligned}
x^2\star_\theta f(\vec x)&=\left(x^2-\i \theta\, x\,\p{t}-\mbox{$\frac{1}4$}\,\theta^2\, \p{t}^2\right)f(\vec x)\ ,\\[4pt]
t^2\star_\theta f(\vec x)&=\left(t^2+\i \theta\, t\,\p{x}-\mbox{$\frac{1}4$}\,\theta^2\, \p{x}^2\right)f(\vec x)\ .
\end{aligned}\label{txstar1}
\end{equation}
Consequently, there is a one-parameter family of operators which continuously interpolates between the Euclidean and the Minkowski space wave operators. They are denoted by $\P^2(\vt)$ and $\Pt(\vt)$, with $\vt\in[-\frac\pi2,\frac\pi2]$, and are defined by
\begin{eqnarray}
\begin{aligned}
\P^2(\vt)&=\e^{\i\vt}\,\Big(\cos(\vt)\,(\Pii)_1-\i\sin(\vt)\,(\Pmm)_0\Big) \ , \\[4pt]
\Pt(\vt)&=\e^{\i\vt}\,\left(\cos(\vt)\,(\Pti)_1-\i\sin(\vt)\,(\Ptm)_0\right)\ .
\end{aligned}\label{ptheta3}
\end{eqnarray}
Using (\ref{txstar1}) one easily checks
\begin{equation}
\P^2(\vt)\,f(\vec x)= H(\vt)\star_{2/E}f(\vec x)\qquad \mbox{and} \qquad
\Pt(\vt)\,f(\vec x)= f(\vec x)\star_{2/E}H(\vt)\ ,
\label{genp1}
\end{equation}
where
\begin{eqnarray}
H(\vt):=E^2\,\big(x^2+\e^{2\i\vt}\,t^2\big)\ .
\end{eqnarray}

The wave operators (\ref{ptheta3}) relate both signatures, with $\vt=0$ corresponding to Euclidean signature and $\vt=\pm\,\frac\pi2$ to hyperbolic signature. 
In the limit $E\rightarrow0$ one easily verifies that this regularization reduces to the $\i\epsilon$-prescription for the usual Klein-Gordon operator. Hence it can be regarded as a generalization of the $\i\epsilon$-prescription to the case with an external electromagnetic field. To distinguish both schemes we will call this alternative prescription the \emph{$\vt$-regularization}.

Using the Weyl-Wigner correspondence, the eigenvalue equations of our original operators can be represented on the space of Weyl symbols by
\begin{equation}
\begin{aligned}
\P^2(\vt)\,f_{mn}^{(E_\vt)}(\vec x)&=\wi{\vech H(\vt)\,\vech f_{mn}^{(E_\vt)}}(\vec x)=\lambda_{m}^{(E_\vt)}\,f_{mn}^{(E_\vt)}(\vec x) \ , \\[4pt]
\Pt(\vt)\,f_{mn}^{(E_\vt)}(\vec x)&=\wi{\vech f_{mn}^{(E_\vt)}\,\vech H(\vt)}(\vec x)=\lambda_{n}^{(E_\vt)}\, f_{mn}^{(E_\vt)}(\vec x) \ ,
\end{aligned}
\end{equation}
with $\vech f_{mn}^{(E_\vt)}=\wee{f_{mn}^{(E_\vt)}}$ and the Weyl symbol
\begin{eqnarray}
\vech H(\vt)=\half\,\Big(\wee{\sqrt{2}\,E\,x}^2+\e^{2\i\vt}\, \wee{\sqrt{2}\,E\,t}^2\Big) = \half\,\Big(\vech p^2+\e^{2\i\vt}\, \vech q^2\Big)\ .
\label{Hamho}\end{eqnarray}
The eigenvalues will turn out to depend on $E$ and $\vt$ only through the combination 
\begin{eqnarray}
E_\vt:=E\, \e^{\i\vt} \ ,
\label{Evartheta}\end{eqnarray}
which explains the notation. The Hermitian symbols $\we{\sqrt{2}\,E\,x}=\vech p$ and $\we{\sqrt{2}\,E\,t}=\vech q$ obey the commutation relation of the Heisenberg algebra
\begin{eqnarray}
[\vech q,\vech p]=2E^2\,\we{t\star_{2/E}x - x\star_{2/E}t} =4\i E\ ,
\label{Hamhocommrel}\end{eqnarray}
where we used the fundamental property
\beq
\wii{\vech f}\star_{2/E}\wii{\vech g}=\wii{\vech f\, \vech g}
\eeq
of the Weyl-Wigner correspondence.

The operators $\vech H(\vt)$ for $\vt\in(-\frac\pi2,\frac\pi2)$ are
known as \emph{complex harmonic oscillator Hamiltonians}. When defined on $\s(\R)$ they have a \emph{discrete} spectrum given by
\begin{eqnarray}
\sigma\big(\vech H(\vt)\big)=\big\{\lambda^{(E_\vt)}_m=4E_\vt\,\big(m+\mbox{$\frac12$} \big)\ \big|\ m\in\N_0 \big\}\ .
\label{eigenvalues1}\end{eqnarray}
The spectrum of $\h(\vt)$ and its eigenoperators $\vech f_{mn}^{(E_\vt)}$ will be investigated in \S\ref{matrix3}.
The simultaneous eigenfunctions of $\P{}^2(\vt)$ and $\Pt(\vt)$ are given by the Wigner transformation
\begin{eqnarray}
f_{mn}^{(E_\vt)}(\vec x)=\wii{\,\vech f_{mn}^{(E_\vt)}\,}(\vec x)\ .\label{fmn0}
\end{eqnarray}
These functions are calculated explicitly in Appendix~\ref{matrixfunctions}. They have an exponential decay for $x,t\rightarrow\infty$ and are Schwartz functions. This is in contrast to the functions one obtains in the limit $\vt\rightarrow\pm\, \frac\pi2$, which are tempered distributions and have been found in \cite{fs08}. As shown in~\cite[App.~D]{thesis}, these functions span a dense subspace of $L^2(\R^2)$, thus every square-integrable function on $\R^2$ can be expanded pointwise into functions lying in the span. In addition, they fulfill the important projector property
\begin{eqnarray}
f_{mn}^{(E_\vt)}(\vec x)\star_{2/E}f_{kl}^{(E_\vt)}(\vec x)=\sqrt{\frac E{4\pi}}\ \delta_{nk}\ f_{ml}^{(E_\vt)}(\vec x)\ . \label{projprop1}
\end{eqnarray}

The generalization of the $\vt$-regularized wave operators to $D=2n$ dimensions are given by
\begin{equation}
\K^2(\vt)=\e^{\i\vt}\, \big(\cos(\vt)\, \K_i^2-\i\sin(\vt)\, \K_\mu^2\big) \qquad \mbox{and} \qquad
\Kt(\vt)=\e^{\i\vt}\, \big(\cos(\vt)\,\tilde\K{}_i^2-\i\sin(\vt)\, \tilde\K{}_\mu^2\big)\ ,
\label{ptheta1}
\end{equation}
which again split up into two-dimensional wave operators defined by \eqref{minkparts}--\eqref{minkpart1} and \eqref{euclparts}. In $D=2n$ dimensions the components $(\Pii)_k$ and $(\Pmm)_{k-1}$, and likewise $(\Pti)_{k}$ and $(\Ptm)_{k-1}$, differ only by a sign for $k=2,\ldots,n$ up to a relabelling of the coordinates. We thus have
\begin{eqnarray}
\K^2(\vt)=\P{}^2(\vt)+\e^{2\i\vt}\, \sum_{k=2}^{n}\,(\Pii)_k\qquad \mbox{and} \qquad 
\Kt(\vt)=\Pt(\vt) +\e^{2\i\vt}\, \sum_{k=2}^{n}\, (\Pti)_k
\label{ptheta2}
\end{eqnarray}
according to \eqref{ptheta3}. The eigenfunctions of the operators $(\Pii)_k$ and $(\Pti)_k$ are just Landau wavefunctions. What remains is to find the eigenfunctions of the remaining parts of the wave operators. Since all of these two-dimensional differential operators commute, the total eigenfunctions will be a product of the individual two-dimensional eigenfunctions.

\subsection{Perturbative Quantum Field Theory\label{Perturbative}}

Without loss of generality we will choose in the following always $\vt>0$ and define $\vt=\frac\pi2-\kappa>0$ for small $\kappa>0$. Denoting
\begin{eqnarray}
\big(\Kmm-\mu^2\big)_{\kappa}:=\e^{\i\kappa}\,\K^2\big(\mbox{$\frac\pi2$}-\kappa\big)-\e^{-\i\kappa}\, \mu^2\qquad \mbox{and} \qquad
\big(\Ktm-\mu^2\big)_{\kappa}:=\e^{\i\kappa}\,\tilde\K{}^2\big(\mbox{$\frac\pi2$}-\kappa\big)-\e^{-\i\kappa}\, \mu^2 \nonumber\\
\end{eqnarray}
the regularized LSZ model is defined by the classical action
\begin{equation}
\begin{aligned}
\s_{\sf LSZ}^{(\kappa)}&=\int\, \dd^D\vec x\ \phi^*(\vec x)\, \left(\sigma\, (\Kmm-\mu^2)_\kappa+(1-\sigma)\, (\Ktm-\mu^2)_\kappa\right)\, \phi(\vec x)\\
&\quad - g\, \Big(\alpha\, \int\,\dd^D\vec x\ (\phi^*\star_\Theta\phi\star_\Theta \phi^*\star_\Theta\phi)(\vec x)+\beta\, \int\,\dd^D\vec x\ (\phi^*\star_\Theta \phi^*\star_\Theta \phi\star_\Theta \phi)(\vec x)\Big) \ ,
\end{aligned}\label{lszreg}
\end{equation}
and the regularized Grosse-Wulkenhaar model by
\begin{eqnarray}
\begin{aligned}
\s_{\sf GW}^{(\kappa)}&=\int\,\dd^D\vec x\ \half \, \phi(\vec x)\, \Big(\,\half \, (\Kmm-\mu^2)_\kappa+\half \, (\Ktm-\mu^2)_\kappa\,\Big)\, \phi(\vec x) \\ &\quad -g\, \int\,\dd^D\vec x\ \left(\phi\star_\Theta\phi\star_\Theta \phi\star_\Theta \phi\right)(\vec x)\ .
\end{aligned}
\label{gwreg}\end{eqnarray}

The Minkowski space duality covariant noncommutative quantum field theory of the regularized LSZ model is defined by the partition function, which is the generating functional obtained by adding external sources $J(\vec x)$ and $J^*(\vec x)$ to the action (\ref{lszreg}) with
\begin{eqnarray}
Z[J,J^*]=\lim_{\kappa\to0^+}\, \mathcal N\,\int\, \mathcal D\phi\ \mathcal D\phi^*\ \exp\Big(\i \s^{(\kappa)}_{\sf LSZ}+\int\, \dd^D\vec x\ J^*(\vec x)\,\phi(\vec x)+\int\, \dd^D\vec x\ \phi^*(\vec x)\,J(\vec x)\Big)
\end{eqnarray}
and analogously for the real Grosse-Wulkenhaar model, where $\mathcal N$ is a normalization constant. The precise definition of the path integral measure is not required to determine $Z[J,J^*]$ perturbatively, since only the vanishing of the integrand for $|\phi|\rightarrow\infty$ is needed to find a functional differential equation for the partition function via formal integration by parts in field space; this is ensured by the $\vartheta$-regularization. The ``free'' partition function $Z_0[J,J^*]:=Z[J,J^*]\big|_{g=0}$ is then the solution of
\begin{equation}
\begin{aligned}
\lim_{\kappa\to0^+}\, \big(\sigma\, (\Kmm-\mu^2)_\kappa+(1-\sigma)\, (\Ktm-\mu^2)_\kappa\big) \frac{\delta Z_0[J,J^*]}{\delta J^*(\vec x)}&=\i J(\vec x)\, Z_0[J,J^*]\ , \\[4pt]
\lim_{\kappa\to0^+}\, \big(\sigma\, (\Kmm-\mu^2)_\kappa+(1-\sigma)\, (\Ktm-\mu^2)_\kappa\big) \frac{\delta Z_0[J,J^*]}{\delta J(\vec x)}&=\i J^*(\vec x)\, Z_0[J,J^*]
\end{aligned}
\end{equation}
given by
\begin{eqnarray}
Z_0[J,J^*]=\lim_{\kappa\rightarrow0^+}\, \exp\Big(\i \int\,\dd^D\vec x\ \int\,\dd^D\vec y\ J^*(\vec x)\,\Delta^{(\kappa,\sigma)}(\vec x,\vec y)\,J(\vec y)\Big) \ ,
\label{z0real}
\end{eqnarray}
with $\Delta^{(\kappa,\sigma)}(\vec x,\vec y)$ the regularized dressed propagator defined through the equation
\begin{equation}
\big(\sigma\, (\Kmm-\mu^2)_\kappa+(1-\sigma)\, (\Ktm-\mu^2)_\kappa\big)\Delta^{(\kappa,\sigma)}(\vec x,\vec y)=\delta(\vec x-\vec y) \ .
\label{dresseddef}\end{equation}
An explicit expression for $\Delta^{(\kappa,\sigma)}(\vec x,\vec y)$ will be derived in \S\ref{proppos}. The full interacting quantum field theory is given by the partition function 
\begin{eqnarray}
Z[J,J^*]=\lim_{\kappa\rightarrow0^+}\, \mathcal N\, \exp\Big[\i\s_{\sf int}\Big(\frac{\delta}{\delta J^*},\frac{\delta}{\delta J}\Big)\Big]\exp\Big(\i\int\,\dd^D\vec x\ \int\,\dd^D\vec y\ J^*(\vec x)\,\Delta^{(\kappa,\sigma)}(\vec x,\vec y)\,J(\vec y)\Big) \label{cgenfunctional1}
\end{eqnarray}
leading to a perturbative expansion in Feynman diagrams corresponding to the interaction part $\s_{\sf int}[\phi,\phi^*]$ of the action (\ref{lszmink1}) and the dressed propagator $\Delta^{(\kappa,\sigma)}(\vec x,\vec y)$. The corresponding Green's functions contain products of distributions and have to be regularized; this is described in \S\ref{quantum1}. For real scalar fields we get
\begin{eqnarray}
Z[J]=\lim_{\kappa\rightarrow0^+}\, \mathcal{N}\exp\Big[\i\s_{\sf int}\Big(\frac{\delta}{\delta J}\Big)\Big]\exp\Big(\, \frac{\i}2\, \int\,\dd^D\vec x \ \int\, \dd^D\vec y\ J(\vec x)\,\Delta^{(\kappa)}(\vec x,\vec y)\,J(\vec y)\Big)\ , \label{rgenfunctional1}
\end{eqnarray}
with regularized dressed propagator $\Delta^{(\kappa)}(\vec x,\vec y)$ given by (\ref{dresseddef}) for $\sigma=\frac12$.

In the following we will construct dynamical matrix models
representing the regularized duality covariant quantum field
theories. In fact, our construction has killed two birds with one
stone. First of all, we have regularized the wave operator such that
no zero eigenvalues occur, so we can invert it to get a unique
propagator. On the other hand, we have also found a discrete spectrum
for the regulated wave operator which, together with the projector
relation \eqref{projprop1}, is needed to define a proper matrix model
formulation of the quantum field theory. This is in marked contrast to
the usual $\i\epsilon$-prescription that gives a regulated wave
operator $\d^{(\epsilon)}_{\vec x}=\sigma\, \Kmm+(1-\sigma)\,
\Ktm-\mu^2+\i\epsilon$, which simply amounts to adding the constant
$\i\epsilon$ to the continuous spectrum of the electric part of the wave operators, but otherwise leaves its continuous character unaltered. A perturbative quantum field theory amenable for the continuous basis approach with functions $\chi_{pq}$ is analysed in this way in~\cite{zahn}.

In the following we shall address the following questions:
\begin{itemize}
\item What is the interpretation of the sign of $\vt$?
\item Which propagator do we obtain in the limit $\vt\rightarrow\pm\,\frac\pi2$?
\item Is it possible to prove duality covariance for $\vt\neq\pm\,\frac\pi2$ at the quantum level?
\item Are the Feynman diagrams finite in the limit $\vt\rightarrow\pm\,\frac\pi2$?
\end{itemize}
In \S\ref{cont} we shall argue that, like the $\i\epsilon$ prescription at $E=0$, the $\vt$-regularization is related to causality, and flipping the sign of the regulator corresponds to interchanging the Feynman (causal) propagator with the Dyson (anti-causal) propagator, i.e. interchanging the particle and anti-particle descriptions in a background electric field. The proof of duality covariance at the classical level follows easily from the Euclidean and Minkowski space proofs given in \cite{ls02} and \cite{fs08}, respectively. The spacetime metric plays no role in the proof, which relies solely on Fourier expansion techniques and Gaussian integrations; in \S\ref{quantum1} we will calculate the partition function of duality covariant noncommutative quantum field theories, and hand in the proof of the duality invariance at quantum level.

\section{Dynamical Matrix Models}\label{matrix3}

In this section we will work out the matrix model representations of the perturbative quantum field theories defined above. In \S\ref{Mink} we used the Weyl-Wigner transformation to map the eigenvalue problem for the $\vt$-regularized wave operators to that of the {complex} harmonic oscillator. Below we investigate its spectrum and eigenfunctions, and construct the appropriate generalizations of the Landau wavefunctions. Using their Fock space representation, we will finally arrive at the matrix model representation for the two-dimensional classical models and their corresponding quantum field theories. The generalization to higher dimensions is also presented.

\subsection{Complex Harmonic Oscillator Wavefunctions}\label{genoscbasis0}

We will begin by investigating the spectrum of the complex harmonic oscillator Hamiltonian $\h(\vt)$ defined in (\ref{Hamho}), with commutation relation (\ref{Hamhocommrel}) and positive real frequency $E\in\R_+$, which turns out to have a discrete spectrum (\ref{eigenvalues1}) resembling the usual harmonic oscillator spectrum rotated into the complex plane by a phase factor $\e^{\i\vt}$. Since $\vech q=\we{\sqrt{2}\, E\, t}$, it is natural to work in the representation defined by the eigenbasis of $\vech q$ such that 
\begin{eqnarray}
\bra{q'\,}\, \vech q\, \ket{q}=\sqrt{2}\, E\, q\,\bk{q'\,}{q} \qquad \mbox{and} \qquad \bra{q'\,}\, \vech p\, \ket q=-\i \sqrt{8}\, \partial_q\,\bk{q'\,}{q}\label{rep0}
\end{eqnarray}
and thus
\begin{eqnarray}
\bra{q'\,}\,\h(\vt)\, \ket{q}=\left(-4\p{q}^2+E_\vt^2\, q^2\right)\bk{q'\,}{q}
\end{eqnarray}
with the condensed notation (\ref{Evartheta}). Firstly, note that the eigenvalue differential equation
\begin{eqnarray}
\left(-4\p{q}^2+E_\vt^2\, q^2\right)f^{(E_\vt)}_m(q)=4E_\vt\, \left(m+\mbox{$\half$}\right)\, f^{(E_\vt)}_m(q)
\end{eqnarray}
is fulfilled for complex values $E_\vt$ if $f^{(E_\vt)}_m(q)$ represent the usual Hermite oscillator wavefunctions $f_m^{(E)}(q)$ with the complex frequency $E_\vt$ substituted for $E$,
\begin{eqnarray}
f_m^{(E_\vt)}(q)=\Big(\,\frac{\sqrt{E_\vt}}{2^m\, m!\, \sqrt{2\pi}}\,\Big)^{1/2}\, \e^{-E_\vt\, q^2/4}\, H_m\big(\sqrt{E_\vt/2}\, q\big) \ , \label{genosc1}
\end{eqnarray}
where $H_m(z)=(-1)^m\,\e^{z^2}\,\partial_z^m\e^{-z^2}$ are the Hermite polynomials.
These functions will be called {complex harmonic oscillator wavefunctions}, as a generalization of the harmonic oscillator wavefunctions to complex frequencies $E_\vt$. They possess an exponential decay due to the Gaussian factor, and are thus Schwartz functions on $\R$ for $|\vt|<\frac\pi2$. We expect that by continuity, for $|\vt|$ small enough the eigenvalues of the complex harmonic oscillator Hamiltonian are given by the set (\ref{eigenvalues1}); the values \eqref{eigenvalues1} are indeed the eigenvalues of $\h(\vt)$ for $|\vt|<\frac\pi2$~\cite{davies98}. 

The complex harmonic oscillator wavefunctions \eqref{genosc1} are not orthogonal, and thus do not serve as a usual Hilbert space basis for $\s(\R)$. But together with their complex conjugated functions and for $\Re(E_\vt)>0$, they constitute a bi-orthogonal system with respect to the $L^2$-inner product $\bk{-}{-}$. This means that the two sets of functions $\big(f_m^{(E_\vt)}\big)_{m\in\N_0}$ and $\big(f_m^{(E_{-\vt})} \big)_{m\in\N_0}$ with nonzero $E_\vt$ and $\Re(E_\vt)>0$ fulfill
\begin{eqnarray}
\big\langle f_n^{(E_{-\vt})}\big| f_m^{(E_\vt)}\big\rangle = \int\, \dd q\ f_n^{(E_\vt)}(q)\,f_m^{(E_\vt)}(q)=\delta_{nm}\ ,
\end{eqnarray}
which follows immediately from the orthogonality of the Hermite functions on $\R$ by a deformation of the integration contour to a straight line from $-\infty\, \e^{\i\vt}$ to $+\infty\, \e^{\i\vt}$. This rotation is possible due to the Gaussian factor $\e^{-E_\vt\,q^2/2}$ in the integrand, ensuring an exponential decay for $\Re(E_\vt)>0$. In addition their linear span is dense in $L^2(\R)$, which means that every square-integrable function on $\R$ can be approximated pointwise by a linear combination of these functions; the proof can be found in~\cite[App.~D]{thesis}. But the series which occur are not convergent in the $L^2$-norm and thus do not build a Riesz basis~\cite{davies98,davies04}.

To ensure the applicability of this basis to arbitrary quantum field
theories, however, one has also to be able to deal with scalar
products and (tempered) distributions. The problem of uniform
convergence for $|\vt|\leq \frac\pi2$ might be circumvented by considering
a smaller space than the Schwartz space, like the space of smooth
functions with compact support, or by considering the Sturm-Liouville
problem for the complex harmonic oscillator Hamiltonian on a finite
interval $[-L,L]$ in $\R$; the expansion on $\s(\R)$ might then be
defined in some limiting procedure. The applicability of the Gel'fand-Shilov space of
type $\s_{\alpha}^{\alpha}(\R)\subset\s(\R)$ with $\alpha=\frac12$ as
an appropriate dense subspace of
fields on $\R$ is discussed in~\cite[App.~C]{thesis}, see also~\cite{Sol07a,Sol07b}; but the question of which precise spaces of functions this complex oscillator basis is applicable in a distributional sense is still an open problem.

Although it would be desirable to have a general rule which tells us
for which functions the matrix basis is applicable, for a given field
theory it suffices to derive the asymptotics of the matrix space
representation of the corresponding propagator in order to ensure the
convergence of the sums in Feynman diagrams; this is investigated in
\S\ref{ren}, but at present it is an open issue. In the following we
will use the matrix basis to derive the propagators of the various
field theories and find that they coincide with the position space
propagators in all cases for which results are already known in the
literature. In Appendix~\ref{effaction} the one-loop effective action
of the Klein-Gordon theory in a constant background electric field is
calculated with the help of the matrix basis, and shown to also
coincide with the known results. By picking up the regularization
scheme imposed on the position space propagator in the Euclidean case,
which effectively cuts off the matrix element summations at some
finite rank $N$, the occuring Feynman diagrams of the
$\vt$-regularized field theories are well-defined and duality
covariant. Whether or not new divergences arise in the limit
$N\to\infty$ remains to be investigated.

\subsection{Complex Landau Wavefunctions}\label{genlandau}

We will now construct the complex Landau wavefunctions $f_{mn}^{(E_\vt)}$, defined by \eqref{fmn0}, through Wigner distribution of the tensor product of two complex oscillator wavefunctions $f_m^{(E_\vt)}$. We will also derive a ``ladder operator'' type construction, which allows us to obtain the matrix model representation of the duality covariant field theories. For the moment we set $\theta=2/E$.

We will first relate the ordinary and complex harmonic oscillator wavefunctions using {complex scaling} methods. Introducing the Hermitian scaling operator
\begin{eqnarray}
\v (\vartheta)=\exp\Big(-\frac{\vartheta}{2E}\,\big(\vech p\, \vech q+\vech q\, \vech p\big)\Big)
\end{eqnarray}
we see that 
\begin{equation}
\v(\vt)\,\vech q\,\v(\vt)^{-1}=\e^{\i\vt/2}\, \vech q\qquad \mbox{and} \qquad
\v(\vt)\,\vech p\,\v(\vt)^{-1}=\e^{-\i\vt/2}\, \vech p\ . \label{pqscaled1}
\end{equation}
The complex harmonic oscillator Hamiltonian (\ref{Hamho}) is thus related to the ordinary oscillator Hamiltonian by
\begin{eqnarray}
\h(\vartheta)=\e^{\i\vartheta}\,\v(\vartheta)\,\mbox{$\frac12$}\,\big(\vech p^2+\vech q^2\big)\,\v(\vartheta)^{-1} \ ,
\end{eqnarray}
while the complex eigenfunctions can now easily be obtained from the orthonormal oscillator number basis states $\ket{m}$, $m\in\N_0$, where $\bk{m}{n}=\delta_{mn}$ and $\bk{q}{m}=f_m^{(E)}(q)$ are the ordinary harmonic oscillator wavefunctions
\begin{eqnarray}
f_m^{(E)}(q)=\Big(\,\frac{\sqrt{E}}{2^m\, m!\, \sqrt{2\pi}}\,\Big)^{1/2}\, \e^{-E\,q^2/4}\, H_m\big(\sqrt{E/2}\, q\big)\ .\label{hosci1}
\end{eqnarray}
By noting that
\begin{eqnarray}
\h(\vartheta)\,\v(\vartheta)\ket{m}=\e^{\i\vartheta}\, \v(\vartheta)\,\h(0)\,\ket{m}=\e^{\i\vartheta}\, 4E\,\big(m+\mbox{$\frac12$}\big)\,\v(\vartheta)\,\ket{m} \ , \label{eigeneq1}
\end{eqnarray}
the corresponding eigenfunctions are related to the oscillator wavefunctions by
\begin{eqnarray}
f_m^{(E_\vt)}(q)=\bk{q}{f_m^{(E_\vt)}}:=\bra{q}\,\v(\vartheta)\,\ket{m}=\e^{\i\vartheta/4}\, f_m^{(E)}\big(\e^{\i\vartheta/2}\,q\big)\ .
\end{eqnarray}

Clearly the left/right eigenoperators of $\h(\vt)$ are tensor products of the form
\begin{eqnarray}
\vech f_{mn}^{(E_\vt)}=\sqrt{\frac{E}{4\pi}}\,\v(\vartheta)\kb{m}{n}\v(-\vartheta)\ ,
\end{eqnarray}
and the complex Landau wavefunctions are thus given by
\begin{eqnarray}
f_{mn}^{(E_\vt)}(\vec x) =\sqrt{\frac{E}{4\pi}}\,\wi{\v(\vt)\kb{m}{n}\v(-\vt)}(\vec x) \ .
\end{eqnarray}
The normalization has been chosen such that
\begin{eqnarray}
\int\,\dd^2\vec x\ f_{mn}^{(E_\vt)}(\vec x)&=& \sqrt{\frac{E}{4\pi}}\, \int \, \dd t\ \int\,\dd x \ \int\, \dd k\ \e^{\i E\, k\, x/2}\, \langle t+k/2|\, \v(\vt)\, \kb{m}{n}\, \v(-\vt)\, |t-k/2\rangle \nonumber\\[4pt] &=&\sqrt{\frac{E}{4\pi}}\, \frac{4\pi}E\, \big\langle f_n^{(E_{-\vt})}\big| f_m^{(E_\vt)}\big\rangle \ = \ 
\sqrt{\frac{4\pi}E}\ \delta_{mn}\ ,
\label{normcond}\end{eqnarray}
where we have used the explicit representation for the Wigner transformation \eqref{wignertrafo}. From \eqref{wignertrafo} we also see that complex conjugation yields
\begin{eqnarray}
f_{mn}^{(E_\vt)}(\vec x)^*=\sqrt{\frac{E}{4\pi}}\, \int\, \dd k\ \e^{-\i E\, k\, x/2}\, \bra{t+k/2}\, \v(-\vt)\, \kb{n}{m}\, \v(\vt)\, \ket{t-k/2}=f_{nm}^{(E_{-\vt})}(\vec x)
\label{fmnconj}\end{eqnarray}
and the {projector property} takes the form
\begin{eqnarray}
\big(f_{mn}^{(E_\vt)}\star_{2/E} f_{kl}^{(E_\vt)}\big)(\vec x)= \frac{E}{4\pi}\, \wi{\v(\vt)\ket{m} \, \bk{n}{k} \, \bra{l}\v(-\vt)}(\vec x)=\sqrt{\frac{E}{4\pi}}\,\delta_{nk}\, f_{ml}^{(E_\vt)}(\vec x)\ .\label{genprojector}
\end{eqnarray}
Together with the normalization condition this implies the bi-orthogonality of the system of complex Landau wavefunctions with respect to the $L^2$-inner product
\begin{eqnarray}
\bk{f_{mn}^{(E_\vt)}}{f_{kl}^{(E_{-\vt})}}&=&\int\, \dd^2\vec x\ f_{nm}^{(E_{-\vt})}(\vec x)\,f_{kl}^{(E_{-\vt})}(\vec x)\notag\\[4pt]
&=&\int\,\dd^2\vec x\ \big(f_{nm}^{(E_{-\vt})}\star_{2/E} f_{kl}^{(E_{-\vt})}\big)(\vec x)\notag\\[4pt]
&=&\sqrt{\frac{E}{4\pi}}\, \int\, \dd^2\vec x\ \delta_{mk}\, f_{nl}^{(E_{-\vt})}(\vec x)\ = \ \delta_{mk}\, \delta_{nl}\ .
\end{eqnarray}

The explicit expressions for the matrix basis functions are given by
\begin{proposition}\label{genlandautheo}
The complex Landau wavefunctions $f_{mn}^{(E_\vt)}(\vec x)$ for $m,n\in\N_0$ are given by
\begin{eqnarray}
f_{mn}^{(E_\vt)}(t,x)&=&(-1)^{\min(m,n)}\, \sqrt{\frac{E}{\pi}}\, \sqrt{\frac{\min(m!,n!)}{\max(m!,n!)}}\, E_\vt^{|m-n|/2}\notag\\
&&\times\, \e^{-E_\vt\, x_+^{(\vt)}\, x_-^{(\vt)}/2}\,\big(x^{(\vt)}_{-{\rm sgn}(m-n)}\big)^{|m-n|}\,L_{\min(m,n)}^{|m-n|}\big(E_\vt\,x_+^{(\vt)}\, x_-^{(\vt)}\big) \ ,
\label{genlandau0}\end{eqnarray}
where
\beq
x_\pm^{(\vt)}=t\pm\i\e^{-\i\vartheta}\, x
\label{genLCcoords}\eeq
are {complex light-cone coordinates} and 
$L_n^k(z)$ are the associated Laguerre polynomials.
\end{proposition}

The proof of Proposition~\ref{genlandautheo} is found in
Appendix~\ref{matrixfunctions}. Setting $\vt=0$ this result coincides
with the well-known expression for the Landau wavefunctions in the
Euclidean case. The coordinates (\ref{genLCcoords}) continuously interpolate between the
complex coordinates $t\pm\i x$ in the Euclidean case $\vartheta=0$
and the light-cone coordinates $t\pm x$ in the hyperbolic case
$\vartheta=\pm\, \frac\pi2$. Since
\begin{eqnarray}
E_\vt \,x_+^{(\vt)}\,x_-^{(\vt)}=E\,\big(\e^{\i\vt}\, t^2+\e^{-\i\vt}\, x^2\big) = E\, \big(\cos(\vt)\,(t^2+x^2)+\i\sin(\vt)\, (t^2-x^2)\big) \ ,
\end{eqnarray}
we see that similarly to $f_{m}^{(E_\vt)}$ these functions are Schwartz functions only for $|\vt|<\frac\pi2$; in particular they belong to the Gel'fand-Shilov spaces $\s_\alpha^\alpha(\R^2)$ for all $\alpha\geq\frac12$. At $\vt=\pm\,\frac\pi2$ they have a polynomial increase and are thus tempered distributions.

The Fock space representation of the harmonic oscillator wavefunctions has a counterpart in the complex scaled version, which will prove very useful in the explicit determination of the dynamical matrix models. For this, we note that
\begin{eqnarray}
\kb{f_n^{(E_\vt)}}{f_m^{(E_{-\vt})}}&=&\v(\vt)\,\frac{(\vech a^\dagger)^m}{\sqrt{m!}}\, \kb{0}{0}\, \frac{(\vech a)^n}{\sqrt{n!}}\, \v(\vt)^{-1}\notag\\[4pt]&=&\frac{1}{\sqrt{m!\, n!}}\, \big(\v(\vt)\, \vech a^\dagger\, \v(\vt)^{-1}\big)^m\, \kb{f_0^{(E_\vt)}}{f_0^{(E_{-\vt})}}\, \big(\v(\vt)\, \vech a\, \v(\vt)^{-1}\big)^n
\end{eqnarray}
with $\vech a=\frac1{\sqrt{8E}}\, (\vech q+\i\vech p)$. We can use the relations \eqref{pqscaled1} to get
\begin{eqnarray}
\v(\vt)\,\vech a^\dagger\,\v(\vt)^{-1}&=&\frac1{\sqrt{8E}}\, \big(\e^{\i\vt/2}\, \vech q-\i\e^{-\i\vt/2}\, \vech p\big) \ , \nonumber\\[4pt]
\v(\vt)\,\vech a\,\v(\vt)^{-1}&=&\frac1{\sqrt{8E}}\, \big( \e^{\i\vt/2}\,\vech q+\i\e^{-\i\vt/2}\,\vech p \big)\ .
\end{eqnarray}
Since $\we{\sqrt{2}\, E\, t}=\vech q$ and $\we{\sqrt{2}\, E\, x}=\vech p$ we find
\begin{eqnarray}
\wi{\v(\vt)\,\vech a^\dagger\, \v(\vt)^{-1}}=\sqrt{\frac{E_\vt}4}\, x_-^{(\vt)}\qquad \mbox{and} \qquad\wi{\v(\vt)\,\vech a\,\v(\vt)^{-1}}=\sqrt{\frac{E_\vt}4}\, x_+^{(\vt)}\ ,
\end{eqnarray}
where we used the complex light-cone coordinates (\ref{genLCcoords}).
The corresponding derivatives are given by
\begin{eqnarray}
\p{\pm}^{(\vt)}=\p t\mp\i\e^{\i\vt}\, \p x\ ,\label{rel1}
\end{eqnarray}
with $\p{\pm}^{(\vt)}x_\pm^{(\vt)}=2$ and $\p{\pm}^{(\vt)}x_\mp^{(\vt)}=0$. The matrix basis functions on $\R^2$ can now be obtained via the Weyl-Wigner correspondence
\begin{eqnarray}
f_{mn}^{(E_\vt)}(\vec x)&=&\sqrt{\frac{E}{4\pi}}\, \wi{\kb{f_m^{(E_\vt)}}{f_n^{(E_{-\vt})}}}(\vec x) \\[4pt]
&=&\frac{1}{\sqrt{m!\,n!}}\,\Big(\, \sqrt{\frac{E_\vt}4}\, x_-^{(\vt)}\, \Big)^{\star_{2/E}\, m}\star_{2/E}\sqrt{\frac{E}{4\pi}}\,\wi{\kb{f_0^{(E_\vt)}}{f_0^{(E_{-\vt})}}}\star_{2/E}\Big(\sqrt{\frac{E_\vt}4}\, x_+^{(\vt)}\, \Big)^{\star_{2/E}\, n}\ . \nonumber
\end{eqnarray}

We define ``ladder operators'' through the equations\footnote{Since
  $(a^+_{(E_\vt)})^\dagger\neq a^-_{(E_\vt)}$ and
  $(b^+_{(E_\vt)})^\dagger\neq b^-_{(E_\vt)}$ they are strictly
  speaking not ladder operators in the conventional sense, but we will nevertheless refer to them as such.}
\begin{equation}
\begin{aligned}
\Big(\,\sqrt{\frac{E_\vt}4}\, x_-^{(\vt)}\,\Big)\star_{2/E} g(\vec x)&=a^+_{(E_\vt)}g(\vec x)\qquad \mbox{and} \qquad \Big(\,\sqrt{\frac{E_\vt}4}\, x_+^{(\vt)}\,\Big)\star_{2/E} g(\vec x)=a^-_{(E_\vt)}g(\vec x)\ ,\\[4pt]
g(\vec x)\star_{2/E} \Big(\,\sqrt{\frac{E_\vt}4}\, x_+^{(\vt)}\,\Big)&=b^+_{(E_\vt)}g(\vec x)\qquad \mbox{and} \qquad g(\vec x)\star_{2/E} \Big(\, \sqrt{\frac{E_\vt}4\, }x_-^{(\vt)}\,\Big)=b^-_{(E_\vt)}g(\vec x)\ .
\end{aligned}
\end{equation}
The differential operators on the right-hand sides of these equations can most easily be obtained by expressing the star-product in terms of the complex light-cone coordinates. Inverting the relations (\ref{genLCcoords}) and \eqref{rel1}, after a bit of algebra we find that the ladder operators are then given by
\begin{equation}
a^\pm_{(E_\vt)}=\half\, \Big(\, \sqrt{E_\vt}\, x_\mp^{(\vt)}\mp\sqrt{\frac{1}{E_\vt}}\, \p{\pm}^{(\vt)}\, \Big)\qquad \mbox{and} \qquad
b^\pm_{(E_\vt)}=\half\,\Big(\, \sqrt{E_\vt}\, x_\pm^{(\vt)}\mp\sqrt{\frac{1}{E_\vt}}\, \p{\mp}^{(\vt)}\, \Big)\ ,
\label{minkladder1}
\end{equation}
and that they fulfill the commutation relations
\begin{eqnarray}
\big[a^-_{(E_\vt)}\,,\,a^+_{(E_\vt)}\big]=1\qquad \mbox{and} \qquad
\big[b^-_{(E_\vt)}\,,\,b^+_{(E_\vt)}\big]=1\ ,
\end{eqnarray}
with all other commutators equal to zero. As expected, we arrive at the usual Euclidean case from \cite{lsz04} when substituting $E_\vt$ by $E$. 

The ground state wavefunction is determined by the differential equations
\begin{eqnarray}
a^-_{(E_\vt)}f^{(E_\vt)}_{00}(\vec x)=b^-_{(E_\vt)}f^{(E_\vt)}_{00}(\vec x)=0
\end{eqnarray}
plus the normalization condition (\ref{normcond}) with $m=n=0$,
which has the solution 
\begin{eqnarray}
f^{(E_\vt)}_{00}(\vec x)=\sqrt{\frac{E}\pi}\, \e^{-E_\vt\, x_+^{(\vt)}\,x_-^{(\vt)}/2}\ .
\end{eqnarray}
The wavefunctions $f_{mn}^{(E_{\vt})}$ have the ladder operator representation
\begin{eqnarray}
f_{mn}^{(E_\vt)}(\vec x)=\frac{\big(a^+_{(E_\vt)}\big)^m}{\sqrt{m!}}\, \frac{\big(b^+_{(E_\vt)}\big)^n}{\sqrt{n!}}\, f_{00}^{(E_\vt)}(\vec x)\ .\label{fmnminkdef}
\end{eqnarray}
It immediately follows that
\begin{equation}
\begin{aligned}
a^-_{(E_\vt)}f_{mn}^{(E_\vt)}(\vec x)=\sqrt{m}\, f_{m-1,n}^{(E_\vt)}(\vec x)\qquad& \mbox{and} \qquad a^+_{(E_\vt)}f_{mn}^{(E_\vt)}(\vec x)=\sqrt{m+1}\, f_{m+1,n}^{(E_\vt)}(\vec x)\ ,\\[4pt]
b^-_{(E_\vt)}f_{mn}^{(E_\vt)}(\vec x)=\sqrt{n}\, f_{m,n-1}^{(E_\vt)}(\vec x)\qquad& \mbox{and} \qquad b^+_{(E_\vt)}f_{mn}^{(E_\vt)}(\vec x)=\sqrt{n+1}\, f_{m,n+1}^{(E_\vt)}(\vec x)\ .
\end{aligned}\label{creation2}
\end{equation}
We will use these relations to obtain the desired matrix model representations. Note that, by~\cite[Lem.~5]{fs08}, the problem of the right test function space is the same as in the complex oscillator case of \S\ref{genoscbasis0}; we can relate the subspaces of Gel'fand-Shilov spaces $\s^{\alpha}_{\alpha}(\R)$ to subspaces of $\s^{\alpha}_{\alpha}(\R^2)$ via Wigner transformation.

\subsection{Matrix Space Representation}\label{matrix4}

Using the Fock space representation of \S\ref{genlandau}, we will now derive the matrix model representation of the classical regularized actions of the LSZ model. In the following we denote 
\begin{eqnarray}
f_{mn}^{\kappa}:=f_{mn}^{(2/\theta_{-\vt})}
\end{eqnarray}
with $\vt=\frac\pi2-\kappa$ and $\theta\neq2/E$ in general; in this case the complex Landau wavefunctions diagonalize the interaction part of the action, but not necessarily the free part of the action. 

We expand the scalar fields in terms of the complex Landau
basis\footnote{This expansion is defined for $L^2$-functions in a
  limiting procedure which can be found in~\cite[App.~D]{thesis}.}
\begin{equation}
\phi(\vec x)=\sum_{m,n=0}^\infty\,f^{\kappa}_{mn}(\vec x)\,\phi^{\kappa}_{mn}\qquad \mbox{and} \qquad
\phi(\vec x)^*=\sum_{m,n=0}^\infty\,f^{\kappa}_{mn}(\vec x)\,\overline\phi{}^{\kappa}_{mn} \ ,
\end{equation}
where the complex expansion coefficients are given by
\begin{equation}
\phi{}^{\kappa}_{mn}=\bk{f^{-\kappa}_{mn}}{\phi}=\int\,\dd^2\vec x\ f^{\kappa}_{nm}(\vec x)\,\phi(\vec x)\qquad \mbox{and} \qquad
\overline{\phi}{}^{\kappa}_{mn}=\bk{f^{-\kappa}_{mn}}{\phi^*\,}= \int\,\dd^2\vec x\ f^{\kappa}_{nm}(\vec x)\,\phi(\vec x)^*
\label{minkmatrixcoeff}
\end{equation}
with $\overline\phi{}^\kappa_{mn}=\big(\phi^{-\kappa}_{nm}\big)^*$. The free parts of the actions can be deduced from
\begin{lemma}\label{matrixps1}
The $\vt$-regularized wave operator of the 1+1-dimensional LSZ model in matrix space is given by 
\begin{eqnarray}
D_{mn;kl}^{(\kappa,\sigma)}&=&\Big(-\e^{-\i\kappa}\, \mu^2+2\i\, \frac{(1+\Omega^2)}{\theta}\, (m+n+1)+\frac{4\i \tilde\Omega}\theta\, (n-m)\Big)\, \delta_{ml}\, \delta_{nk}\notag\\
&&\quad+\, 2\i\,\frac{\Omega^2-1}\theta\,\left(\sqrt{n\, m}\,\delta_{m,l+1}\,\delta_{n,k+1}+\sqrt{(n+1)\,(m+1)}\,\delta_{m,l-1}\,\delta_{n,k-1}\right)\label{LSZop}
\end{eqnarray}
with frequencies $\Omega=E\,\theta/2$ and $\tilde\Omega=(2\sigma-1)\, \Omega$.
\end{lemma}
\begin{proof}
The wave operator is defined in the matrix basis by
\begin{eqnarray}
D_{mn;kl}^{(\kappa,\sigma)}=\int\, \dd^2\vec x\ f_{mn}^{\kappa}(\vec x)\,\left(\sigma\, \e^{\i\kappa}\,\P{}^2\big(\mbox{$\frac\pi2$}-\kappa\big)+(1-\sigma)\, \e^{\i\kappa}\,\Pt\big(\mbox{$\frac\pi2$}-\kappa\big)-\e^{-\i\kappa}\,\mu^2\right)\, f_{kl}^\kappa(\vec x)\ .
\end{eqnarray}
One has
\begin{eqnarray}
\P{}^2(\vt)&=& \frac{\e^{\i\vt}}{2\theta}\, \Big((2+E\, \theta)^2\, \big(a_{(2/\theta_\vt)}^+ \, a_{(2/\theta_\vt)}^-+\mbox{$\half$}\big)+(2-E\,\theta)^2\, \big(b_{(2/\theta_\vt)}^+\, b_{(2/\theta_\vt)}^-+\mbox{$\half$}\big) \nonumber\\ && \qquad \qquad +\, \left(\theta^2\, E^2-4\right)\, \big(a_{(2/\theta_\vt)}^+\, b_{(2/\theta_\vt)}^++a_{(2/\theta_\vt)}^-\, b_{(2/\theta_\vt)}^-\big)\Big)\ , \label{pminkladder1}
\end{eqnarray}
together with a similar expression for $\Pt(\vt)$ with $a_{(2/\theta_\vt)}^\pm$ and $b_{(2/\theta_\vt)}^\pm$ interchanged. These formulas are verified directly by inserting (\ref{minkladder1}). The matrix space representation of the partial differential operators $\P{}^2(\vt)$ and $\Pt(\vt)$ away from the self-dual point can be obtained from \eqref{pminkladder1} with the help of \eqref{creation2} to get
\begin{eqnarray}
\P{}^2_{mn;kl}(\vt)&=&\frac{\e^{\i\vt}}{2\theta}\, \Big((2+E\, \theta)^2\, \big(m+\mbox{$\half$}\big)\, \delta_{ml}\,\delta_{nk}+(2-E\, \theta)^2\, \big(n+\mbox{$\half$}\big)\, \delta_{ml}\,\delta_{nk}\label{pmm}\\
&& \qquad \quad+\, \left(\theta^2\, E^2-4\right)\,\big(\sqrt{n\,m}\,\delta_{m,l+1}\,\delta_{n,k+1}+\sqrt{(n+1)\, (m+1)}\,\delta_{m,l-1}\,\delta_{n,k-1}\big)\Big) \nonumber
\end{eqnarray}
and
\begin{eqnarray}
\Pt_{mn;kl}(\vt)&=&\frac{\e^{\i\vt}}{2\theta}\, \Big((2+E\, \theta)^2\, \big(n+\mbox{$\half$}\big)\, \delta_{ml}\,\delta_{nk}+(2-E\,\theta)^2\, \big(m+\mbox{$\half$}\big)\, \delta_{ml}\,\delta_{nk}\label{ptm}\\
&&\qquad \quad+\, \left(\theta^2\, E^2-4\right)\,\big(\sqrt{n\, m}\,\delta_{m,l+1}\,\delta_{n,k+1}+\sqrt{(n+1)\, (m+1)}\,\delta_{m,l-1}\,\delta_{n,k-1}\big)\Big) \ ,\nonumber
\end{eqnarray}
which can be combined to give \eqref{LSZop}.
\end{proof}

For the LSZ interaction terms, we use the projector property \eqref{genprojector} to get
\begin{eqnarray}
f^{\kappa}_{m_1n_1}\star_\theta f^{\kappa}_{m_2n_2}\star_\theta f^{\kappa}_{m_3n_3}\star_\theta f^{\kappa}_{m_4n_4}=\frac{1}{(2\pi\,\theta)^{3/2}}\ \delta_{n_1m_2}\, \delta_{n_2m_3}\, \delta_{n_3m_4}\ f^{\kappa}_{m_1n_4}\ .
\label{fepsilonproj}\end{eqnarray}
The regularized LSZ model in two-dimensional Minkowski space can then be represented in the matrix basis as
\begin{eqnarray}
\s_{\sf LSZ}^{(\kappa)}= \sum_{m,n,k,l}\, \overline{\phi}{}^{\kappa}_{mn}\, D_{mn;kl}^{(\kappa,\sigma)}\,\phi^{\kappa}_{l k}-\frac{g}{2\pi\, \theta}\, \sum_{m,n,k,l}\, \left(\alpha\,\overline{\phi}{}^{\kappa}_{mn}\, \phi^{\kappa}_{nk}\,\overline{\phi}{}^{\kappa}_{kl}\, \phi^{\kappa}_{l m}+\beta\,\overline{\phi}{}^{\kappa}_{mn}\, \overline{\phi}{}^{\kappa}_{nk}\, \phi^{\kappa}_{kl}\,\phi^{\kappa}_{l m}\right)\ .
\end{eqnarray}
As a one-matrix model with infinite complex matrices $\phi_\kappa=\big(\phi^\kappa_{mn}\big)_{m,n\in\N_0}$ this representation reads
\begin{equation}
\begin{aligned}
\s_{\sf LSZ}^{(\kappa)}&=\frac{1}{2\theta}\, \Tr\Big(\big((2-\theta\, E)^2+8\sigma\, \theta\, E\big)\, \phi_{-\kappa}^\dagger\, \mathcal E\, \phi_{\kappa}+\big((2+\theta E)-8\sigma\, \theta\, E\big)\, \phi_{\kappa}\, \mathcal E\, \phi_{-\kappa}^\dagger\\
&\quad \qquad \qquad +\i \big(\theta^2\, E^2-4\big)\, \big(\phi_{-\kappa}^\dagger\, \Gamma^\dagger\, \phi_{\kappa}\, \Gamma+\phi_{\kappa}\, \Gamma^\dagger\, \phi_{-\kappa}^\dagger\, \Gamma\big)-2\theta\, \e^{-\i\kappa}\, \mu^2\, \phi_{-\kappa}^\dagger\, \phi_{\kappa}\\
&\quad \qquad \qquad - \, \frac{g}{2\pi\, \theta}\, \big(\alpha\,\phi_{-\kappa}^\dagger\, \phi_{\kappa}\,\phi_{-\kappa}^\dagger\, \phi_{\kappa}+\beta\,\phi_{-\kappa}^\dagger\, \phi_{-\kappa}^\dagger\, \phi_{\kappa}\,\phi_{\kappa}\big)\Big) \ ,
\end{aligned}\label{matrixmodel2}
\end{equation}
with the diagonal matrix
\begin{eqnarray}
\mathcal E_{mn}=4\i\big(m+\mbox{$\half$}\big)\ \delta_{mn}
\end{eqnarray}
and the infinite shift matrix
\begin{eqnarray}
\Gamma_{mn}=\sqrt{n-1}\ \delta_{m,n-1}\ .
\end{eqnarray}

Using the perturbative solution \eqref{cgenfunctional1}, the duality covariant field theory can be defined perturbatively in the matrix basis by the partition function 
\begin{eqnarray}
Z[J,\overline{J}\,]&=&\lim_{\kappa\rightarrow0^+}\, \mathcal N\, \exp\Big(\, -\frac{\i \alpha\,g}{2\pi\,\theta}\, \sum_{m,n,k,l}\, \frac{\partial^4}{\partial J{}^{\kappa}_{ml}\,\partial\overline{J}{}^{\kappa}_{l k}\,\partial J{}^{\kappa}_{kn}\, \partial\overline{J}{}^{\kappa}_{nm}}\, \Big) \\
&&\qquad\times\, \exp\Big(\,-\frac{\i \beta\,g}{2\pi\,\theta} \, \sum_{m,n,k,l}\, \frac{\partial^4}{\partial J{}^{\kappa}_{ml}\, \partial J{}^{\kappa}_{l k}\, \partial\overline{J}{}^{\kappa}_{kn}\, \partial\overline{J}{}^{\kappa}_{nm}}\, \Big)\exp\Big(\i\, \sum_{m,n,k,l}\, \overline{J}{}^{\kappa}_{mn}\, \Delta{}^{(\kappa,\sigma)}_{mn;kl}\, J{}^{\kappa}_{kl}\Big)\ , \nonumber
\end{eqnarray}
with $J^{\kappa}_{mn}$, $\overline{J}{}^{\kappa}_{mn}$ external sources in the matrix basis and the propagator $\Delta{}^{(\kappa,\sigma)}_{mn;kl}$ defined as the inverse of $D_{mn;kl}^{(\kappa,\sigma)}$,
\begin{eqnarray}
\sum_{k,l}\, D^{(\kappa,\sigma)}_{mn;kl}\, \Delta^{(\kappa,\sigma)}_{l k;sr}=\sum_{k,l}\, \Delta^{(\kappa,\sigma)}_{nm;l k}\, D{}^{(\kappa,\sigma)}_{kl;rs}=\delta_{mr}\, \delta_{ns}\ .
\end{eqnarray}
An explicit expression for the propagator $\Delta{}^{(\kappa,\sigma)}_{mn;kl}$ will be derived in \S\ref{matrixprops}.
The modified Feynman rules are presented in the double line formalism and are exactly as in the Euclidean case~\cite{sza01,gw05}. The generically non-local propagators are represented by double lines with orientation pointing from $\phi^*$ to $\phi$ as
\vspace{-10pt}
\begin{center}
\begin{picture}(40,20)(0,10)
\ArrowLine(5,8)(35,8)
\ArrowLine(5,13)(35,13)
\Text(5,6)[t]{$m$}
\Text(5,14)[b]{$n$}
\Text(35,6)[t]{$l$}
\Text(35,14)[b]{$k$}
\end{picture}
$ \ = \ \Delta^{(\kappa,\sigma)}_{nm;l k}\ .$
\end{center}
\vspace{10pt}
The two interaction terms $\phi^*\star_\theta\phi\star_\theta\phi^*\star_\theta\phi$ and $\phi^*\star_\theta\phi^*\star_\theta\phi\star_\theta\phi$ are represented by different diagonal vertices given respectively by
\vspace{-25pt}
\begin{center}
\begin{picture}(70,70)(-10,25)
\DashArrowLine(0,25)(25,25){2}
\DashArrowLine(25,25)(25,0){2}
\DashArrowLine(30,25)(30,0){2}
\DashArrowLine(55,25)(30,25){2}
\DashArrowLine(55,30)(30,30){2}
\DashArrowLine(30,30)(30,55){2}
\DashArrowLine(25,30)(25,55){2}
\DashArrowLine(0,30)(25,30){2}
\end{picture}
$\ = \ -\frac{\i g\, \alpha}{2\pi\,\theta} \, \delta_{mp}\, \delta_{nq}\, \delta_{kr}\, \delta_{l s} \qquad \mbox{and} \qquad $
\begin{picture}(70,70)(-10,25)
\DashArrowLine(0,25)(25,25){2}
\DashArrowLine(25,0)(25,25){2}
\DashArrowLine(30,0)(30,25){2}
\DashArrowLine(55,25)(30,25){2}
\DashArrowLine(55,30)(30,30){2}
\DashArrowLine(30,55)(30,30){2}
\DashArrowLine(25,55)(25,30){2}
\DashArrowLine(0,30)(25,30){2}
\end{picture}
$ \ =\ -\frac{\i g\, \beta}{2\pi\,\theta} \, \delta_{mp}\, \delta_{nq}\, \delta_{kr}\, \delta_{l s} \ . $
\end{center}
\vspace{25pt}
Restricting to one of these interactions reduces the number of possible diagrams for the complex matrix model.

For real fields, one can apply Lemma~\ref{matrixps1} by setting $\sigma=\frac12$ to immediately get
\begin{lemma}
The $\vt$-regularized Grosse-Wulkenhaar wave operator in 1+1 dimensions has the matrix space representation given by
\begin{eqnarray}
D_{mn;kl}^{(\kappa)}&=&\Big(-\e^{-\i\kappa}\, \mu^2+2\i\, \frac{\Omega^2+1}{\theta}\, (m+n+1)\Big)\, \delta_{ml}\,\delta_{nk}\notag\\
&&\quad+\, 2\i\frac{\Omega^2-1}{\theta}\,\Big(\sqrt{n\,m}\,\delta_{m,l+1}\,\delta_{n,k+1}+\sqrt{(n+1)\,(m+1)}\, \delta_{m,l-1}\, \delta_{n,k-1}\Big)\label{GWop}
\end{eqnarray}
with frequency $\Omega=E\, \theta/2$.
\end{lemma}

The Minkowski space Grosse-Wulkenhaar action in the matrix basis then reads
\begin{eqnarray}
\s_{\sf GW}^{(\kappa)}=\sum_{m,n,k,l}\,\Big(\, \half\,
\phi^{\kappa}_{mn}\,D_{mn;kl}^{(\kappa)}\, \phi^{\kappa}_{kl}- \frac{g}{2\pi\, \theta}\, \phi^{\kappa}_{mn}\, \phi^{\kappa}_{nk}\, \phi^{\kappa}_{kl}\, \phi^{\kappa}_{l m}\, \Big) \label{GWmatrixactionmink}
\end{eqnarray}
with $\overline{\phi}{}^{\kappa}_{mn}=\phi{}^{\kappa}_{nm}$; it thus corresponds to a Hermitian one-matrix model. From the perturbative solution \eqref{rgenfunctional1} the partition function in matrix space is given by
\begin{eqnarray}
Z[J]=\lim_{\kappa\rightarrow0^+}\, \mathcal N\,\exp\Big(\,-\frac{\i g}{2\pi\,\theta} \, \sum_{m,n,k,l}\, \frac{\partial^4}{\partial J^{\kappa}_{ml}\, \partial J^{\kappa}_{l k}\, \partial J^{\kappa}_{kn}\, \partial J^{\kappa}_{nm}}\, \Big)\exp\Big(\, \frac{\i}2\, \sum_{m,n,k,l}\, J^{\kappa}_{mn}\, \Delta^{(\kappa)}_{mn;kl}\, J^{\kappa}_{kl}\, \Big)
\label{ZJGWmatrix}\end{eqnarray}
where the propagator $\Delta^{(\kappa)}_{mn;kl}$ is the inverse of $D^{(\kappa)}_{mn;kl}$ and is represented by the unoriented double line
\vspace{-10pt}
\begin{center}
\begin{picture}(40,20)(0,10)
\ArrowLine(5,8)(35,8)
\ArrowLine(35,13)(5,13)
\Text(5,6)[t]{$m$}
\Text(5,14)[b]{$n$}
\Text(35,6)[t]{$l$}
\Text(35,14)[b]{$k$}
\end{picture}
$ \ =\ \Delta{}^{(\kappa)}_{nm;l k}\ . $
\end{center}
\vspace{10pt}
The vertex of the $\phi^{\star4}$ interaction is given by the graph
\vspace{-25pt}
\begin{center}
\begin{picture}(70,70)(-10,25)
\DashArrowLine(0,25)(25,25){2}
\DashArrowLine(25,25)(25,0){2}
\DashArrowLine(30,0)(30,25){2}
\DashArrowLine(30,25)(55,25){2}
\DashArrowLine(55,30)(30,30){2}
\DashArrowLine(30,30)(30,55){2}
\DashArrowLine(25,55)(25,30){2}
\DashArrowLine(25,30)(0,30){2}
\end{picture}
$ \ =\ -\frac{\i g}{2\pi\,\theta} \, \delta_{mp}\, \delta_{nq}\, \delta_{kr}\, \delta_{l s}\ . $
\end{center}
\vspace{25pt}
Since the vertex is unoriented there are as many diagrams as in the LSZ model with both parameters $\alpha$ and $\beta$ turned on.

\subsection{Generalization to Higher Dimensions}\label{higher2}

The spectra of both partial differential operators in (\ref{ptheta2}) in generic $D=2n$ dimensions are given by the set
\begin{eqnarray}
\Big\{4E\,
\e^{\i\vt}\,\big(l_0+\mbox{$\frac12$}\big)+\mbox{$\sum\limits_{k=1}^{n-1}
  $}\, 4B_k\, \e^{2\i\vt}\, \big(l_k+\mbox{$\frac12$}\big)\ \Big| \
l_0, l_1,\ldots,l_{n-1}\in\N_0\Big\}\ ,
\end{eqnarray}
where the eigenfunctions are products of the complex Landau wavefunctions $f_{m_0n_0}^{(E_\vt)}$ from \S\ref{genlandau} and the ordinary Landau wavefunctions $f_{m_kn_k}^{(B_k)}$, so that
\begin{eqnarray}
f^{(\vec F_\vt)}_{\vec m\vec n}(\vec x):=f_{m_0n_0}^{(E_\vt)}(\vec
x_0)\ \prod_{k=1}^{n-1} \, f^{(B_k)}_{m_kn_k}(\vec x_k) \label{multif2}
\end{eqnarray}
with $\vec x_k=(x^{2k},x^{2k+1})\in\R^2$ for $k=0,1,\dots,n-1$, $\vec
x=(x^\mu)=(x^0,x^1,\dots,x^d) \in\R^{D}$, $\vec m=(m_k),\vec
n=(n_k)\in\N_0^n$, and $\vec F_\vt=(E_\vt,B_1,\ldots,B_{n-1} )\in\C_+\times\R_+^{n-1}$ where $\C_+:=\{z\in\C\ |\ \Re(z)\geq0\}$. The star-product of two multi-dimensional complex Landau wavefunctions (\ref{multif2}) with respect to the deformation matrix \eqref{multitheta0} decomposes into star-products of Landau wavefunctions depending on $\vec x_k$ for $k=0,1,\ldots,n-1$. If in addition $E=2/\theta_0$ and $B_k=2/\theta_k$ for $k=1,\dots,n-1$, then
\begin{eqnarray}
\big(f^{(\vec F_\vt)}_{\vec m\vec n}\star_\Theta f^{(\vec F_\vt)}_{\vec k\vec l}\big)(\vec x)=\frac1{\det(2\pi\,\Theta)^{1/4}}\ \delta_{\vec n\vec k}\ f^{(\vec F_\vt)}_{\vec m\vec l}(\vec x)
\end{eqnarray}
with $\delta_{\vec m\vec n}=\prod_{k}\, \delta_{m_kn_k}$.

The generalization of the matrix model representation to higher spacetime dimensions is now straightforward. To conform with our previous conventions we again set $\vt=\frac\pi2-\kappa>0$ and use the notation
\begin{eqnarray}
f_{\vec m\vec n}^{\kappa}(\vec
x)=f^{(2/(\theta_0)_{-\vt})}_{m_0n_0}(\vec x_0)\, \prod_{k=1}^{n-1} \, f^{(2/\theta_k)}_{m_kn_k}(\vec x_k)\ .
\end{eqnarray}
The functions $f_{\vec m\vec n}^{\kappa}$ are arranged so as to diagonalize the interaction part but not necessarily the free part of the action. The scalar fields on $\R^D$ are expanded in the complex Landau basis
\begin{equation}
\phi(\vec x)=\sum_{\vec m,\vec n\,\in\,\N_0^{n}}\,f^{\kappa}_{\vec m\vec n}(\vec x)\,\phi^\kappa_{\vec m\vec n}\qquad \mbox{and} \qquad
\phi(\vec x)^*=\sum_{\vec m,\vec n\,\in\,\N_0^{n}}\,f^{\kappa}_{\vec m\vec n}(\vec x)\,\overline\phi{}^\kappa_{\vec m\vec n} \ ,
\end{equation}
where the complex expansion coefficients are given by
\begin{equation}
\phi_{\vec m\vec n}^\kappa=\bk{f^{-\kappa}_{\vec m\vec n}}{\phi}= \int\, \dd^D\vec x\ f^{\kappa}_{\vec n\vec m}(\vec x)\,\phi(\vec x)\qquad \mbox{and} \qquad
\overline{\phi}{}^\kappa_{\vec m\vec n}=\bk{f^{-\kappa}_{\vec m\vec n}}{\phi^*\,}= \int\, \dd^D\vec x\ f^{\kappa}_{\vec n\vec m}(\vec x)\,\phi(\vec x)^*\ .
\end{equation}

The matrix space representation of the LSZ model in $D=2n$ dimensions
away from the self-dual point can be obtained by comparing the
operators \eqref{ptheta2} with their two-dimensional constituents, and
their matrix representations given by \eqref{pmm}--\eqref{ptm}
together with their Euclidean counterparts for
$\vt=0$~\cite{lsz04}. The matrix space LSZ wave operator is thus the sum of the two-dimensional Minkowski space wave operator given by \eqref{LSZop} plus $n-1$ copies of the {massless} Euclidean wave operator for $\vt=0$~\cite{lsz04} times the factor $-\e^{-\i\kappa}$. Noting that the massless LSZ wave operators in Euclidean and Minkowski space differ only by a factor of the imaginary unit $\i$, we can write
\begin{eqnarray}
D_{\vec m\vec n;\vec k\vec l}^{(\kappa,\sigma)}=\i\mathcal
D_{m_0n_0;k_0l_0}^{0\, (\sigma)}-\e^{-\i\kappa}\, \sum_{i=1}^{n-1} \, \mathcal D_{m_in_i;k_il_i}^{i\, (\sigma)}-\e^{-\i\kappa}\, \mu^2\ \delta_{\vec m\vec l}\, \delta_{\vec n\vec k} \ , \label{2nlsz}
\end{eqnarray}
where $\vec m=(m_0,m_1, \ldots,m_{n}),\vec n=(n_0,n_1,
\ldots,n_{n}),\vec k=(k_0,k_1, \ldots,k_{n}),\vec l=(l_0,l_1, \ldots,l_{n})\in\N_0^{n}$
and $\mathcal D_{mn;kl}^{j\, (\sigma)}$ are the two-dimensional massless
Euclidean LSZ matrix space wave operators
\begin{eqnarray}
\mathcal D^{j\, (\sigma)}_{mn;kl}&=&\Big(2\, \frac{\Omega^2+1}{\theta_j}\, (m+n+1)+\frac{4\tilde\Omega}{\theta_j}\, (n-m)\Big)\, \delta_{ml}\,\delta_{nk}\notag\\
&&\quad+\, 2\, \frac{\Omega^2-1}{\theta_j}\, \left(\sqrt{n\, m}\, \delta_{m,l+1}\,\delta_{n,k+1}+\sqrt{(n+1)\,(m+1)}\, \delta_{m,l-1}\, \delta_{n,k-1}\right) \ , \label{G0}
\end{eqnarray}
with frequencies $\Omega=E\,\theta_0/2=B_i\, \theta_i/2$ and $\tilde\Omega=(2\sigma-1)\, \Omega$. The $2n$-dimensional regularized LSZ action is then given in the usual matrix space form
\begin{eqnarray}
\s_{\sf LSZ}^{(\kappa)} &=& \sum_{\vec m,\vec n,\vec k,\vec l}\,
\overline{\phi}{}^\kappa_{\vec m\vec n}\,
D^{(\kappa,\sigma)}_{\vec m\vec n;\vec k\vec l}\,\phi^\kappa_{\vec
  l \vec k} \\ && -\, \frac{g}{\sqrt{\det(2\pi\, \Theta)}}\ \sum_{\vec m,\vec n,\vec k,\vec l}\, \left(\alpha\,\overline{\phi}{}^\kappa_{\vec m\vec n}\, \phi^\kappa_{\vec n\vec k}\,\overline{\phi}{}^\kappa_{\vec k\vec l}\, \phi^\kappa_{\vec l \vec m}+\beta\,\overline{\phi}{}^\kappa_{\vec m\vec n}\, \overline{\phi}{}^\kappa_{\vec n\vec k}\, \phi^\kappa_{\vec k\vec l}\,\phi^\kappa_{\vec l \vec m}\right)\ . \nonumber
\end{eqnarray}
Every other result of this section (and ensuing ones) can now formally be generalized to higher dimensions by substituting multi-indices $\vec m,\vec n,\ldots\in\N_0^{n}$ for the usual matrix indices $m,n,\ldots\in\N_0$.

\section{Causality}\label{cont}

In this section we will treat the problem of determining the causal
propagator of the duality covariant field theories in Minkowski
space. Problematic for this issue is the lack of time translation
invariance, which allows for transitions that violate energy
conservation; this manifests itself in an instability of the vacuum
with respect to production of particle-antiparticle pairs. We review how the standard techniques must be altered to take care of these features. We will then examine the corresponding propagators which one obtains by removing the $\vt$-regularization.

\subsection{Causal Propagators}\label{causprop}

The way we chose the propagator of the Minkowski space theory was to find the analytically continued propagator of the Euclidean case. This also brought about the possibility of finding a matrix space representation. In the following we will show that the propagator which is prescribed by the $\vt$-regularization is the causal propagator of the duality covariant quantum field theory. For this, we will first review the connection between regularization and propagators by describing the {eigenvalue representation}, and the related {operator extension method}. 

The free partition function $Z_0[J,J^*]$ of a complex scalar field theory is defined as the vacuum-to-vacuum amplitude
\begin{eqnarray}
Z_0[J,J^*]=\bk{\Omega,{\rm out}}{\Omega,{\rm in}}^{J,J^*} \ ,
\label{vactovacampl}\end{eqnarray}
where $\ket{\Omega,{\rm in}}$ and $\bra{\Omega,{\rm out}}$ are the
vacuum states at time instances $t_{\rm in}$ and $t_{\rm out}$ of the
quantum field theory defined by the free action $\s_0[\phi,\phi^*]$ in
the presence of the sources $J$ and $J^*$. Using Schwinger's action principle, one can show that {causality} implies the relation
\begin{eqnarray}
\left.\frac{\delta^2\log Z_0[J,J^*]}{\delta J^*(\vec x)\,\delta J(\vec
    y)}\, \right|_{J=J^*= 0}=\frac{\bra{0,{\rm out}}T\big(\, \hat\phi(\vec
  x)\, \hat\phi(\vec y)^\dag\, \big)\ket{0,{\rm in}}}{\bk{0,{\rm out}}{0,{\rm in}}}\ ,\label{propfeyn1}
\end{eqnarray}
where $\hat\phi(\vec x)$ is the second quantized field operator, $T$
denotes time-ordering with respect to the time variables $x^0$ and
$y^0$, and $\ket{0,{\rm in}}$ and $\bra{0,{\rm out}}$ are the in- and
out- vacuum states for $J=J^*=0$ which in the presence of further
interactions are taken in the interaction picture where the field
operators satisfy the equations of motion obtained from varying
$\so[\phi,\phi^*]$. For field theories which allow for spontaneous
particle-antiparticle pair production, like the covariant models we
are considering, the in- and out- vacuum vectors are in general not
dual to each other. Thus a non-trivial vacuum-to-vacuum probability $|\bk{0,{\rm out}}{0,{\rm in}}|^2<1$ may occur, since $|\bk{0,{\rm out}}{0,{\rm in}}|^2$ measures the vacuum persistence which is equal to $1$ only if no spontaneous pair creation occurs.

The tempered distribution defined by the right-hand side of (\ref{propfeyn1}) is known as the \emph{causal} propagator and will be denoted as $\i\Delta_c(\vec x, \vec y)$, where the imaginary unit has been factored out to conform with our previous conventions. Quite generally, for a Klein-Gordon field which may be free or moving in an external background which preserves vacuum stability, the expression \eqref{propfeyn1} may be evaluated through the expansion
\begin{eqnarray}
\i\Delta_c(\vec  x,\vec  y)=\tau(x^0-y^0)\, \int_{{\cal C}}\,\dd m(\nu)\ \phi{}_\nu^{(+)}(\vec  x)\, \phi{}_\nu^{\,(+)}(\vec y)^*+ \tau(y^{0}-x^0)\, \int_{{\cal C}} \,\dd m(\nu)\ \phi{}_\nu^{(-)}(\vec x)\, \phi{}_\nu^{\,(-)}(\vec y)^* \label{e1}
\end{eqnarray}
with $\tau$ the Heaviside distribution function, $\big(\phi_\nu^{(\pm)}\big)$ a complete set of distributional {solutions} of the classical field equation with positive and negative frequency, respectively, and $\dd m(\nu)$ a suitable measure on the set ${\cal C}$ of all quantum numbers $\nu$ parametrizing the space of solutions. One can check that the distribution \eqref{propfeyn1}--\eqref{e1} propagates particles (positive frequency solutions) forward in time and anti-particles (negative frequency solutions) backward in time. This is the imprint of causality and lends the causal propagator its name.

The situation is more complicated if the background field spoils
vacuum persistence. A typical example is ordinary quantum
electrodynamics in an external field which allows for pair creation. Crucial for the canonical
quantization scheme and for the expression \eqref{e1} is the existence
of a complete set of classical solutions which have definite positive
or negative frequency for all times. However, such a set of solutions
only exists if we are working on a stationary spacetime, i.e. a
spacetime which admits a global timelike Killing vector
field~\cite{dewitt75}. In our case, there is no such vector field due
to the absence of time translation symmetry; production of
particle-antiparticle pairs manifests itself in an inevitable mixing of positive and negative frequencies at the level of solutions to the field equations.
The requisite methods in this case have been developed in~\cite{git77,fradkin81}.  

Since the asymptotic Hilbert spaces in the remote past and future (if they exist) are different, there are two complete sets of solutions denoted $\big({}\phi_\nu{}^{(\pm)}\big)$ and $\big(\phi_{\nu(\pm)}\big)$, having definite positive/negative frequency at times $t_{\rm in}$ and $t_{\rm out}$, respectively, which are the equivalent of the positive/negative frequency solutions above in the infinite future and past, i.e. in the limits $t_{\rm in}\to-\infty$ and $t_{\rm out}\to+\infty$. The generalization of the expansion into classical solutions \eqref{e1} then reads
\begin{eqnarray}
\i\Delta_c(\vec  x,\vec y)&=&\tau(x^0-y^{0})\, \int_{{\cal C}}\,\dd m(\mu)\ \int_{{\cal C}}\, \dd m(\nu)\ \phi_\mu{}^{(+)}(\vec  x)\,\omega^+(\mu|\nu)\, \phi{}_{\nu(+)}(\vec y)^*\notag\\
&&\ +\, \tau(y^{0}-x^0)\, \int_{{\cal C}}\,\dd m(\mu)\ \int_{{\cal C}}\, \dd m(\nu)\ \phi{}_{\nu(-)}(\vec  x)\,\omega^-(\mu|\nu)\, \phi_\mu{}^{(-)}(\vec y)^* \ . \label{e2}
\end{eqnarray}
Here $\omega^\pm(\mu|\nu)$ is the relative probability for a particle/anti-particle to be scattered by the vacuum in the external field, given by a generalized Wick contraction of creation-annihilation operators on Fock space which appear in the mode expansions of the in- and out- field operators, and which create the in- and out- particle/anti-particle states. For a field theory with a stable vacuum state one has $\omega^\pm(\mu|\nu)=\delta(\mu,\nu)$ and $\phi_\nu^{\,\,(\pm)}=\phi_{\nu(\pm)}$, where $\int_{{\cal C}}\,\dd m(\mu)\ \delta(\mu,\nu)\, f(\mu)=f(\nu)$. This construction determines the propagator uniquely and is equivalent to the definition \eqref{propfeyn1}, but can be quite technically cumbersome to carry out explicitly; hence it is desirable to have another method at hand.

Such an equivalent method, which will prove profitable for us, is the {eigenvalue representation}. Let $\varphi_\lambda(\vec  x)$ be a complete orthonormal set of eigenfunctions of the wave operator $\d_{\vec x}$ of the field theory with eigenvalues $\lambda\in\sigma(\d_{\vec x})$, i.e.
\begin{eqnarray}
\d_{\vec x}\varphi_\lambda(\vec  x)=\lambda\, \varphi_\lambda(\vec  x)
\end{eqnarray}
with
\begin{eqnarray}
\int_{\sigma(\d_{\vec x})}\,\dd\ell(\lambda)\ \varphi_\lambda(\vec x)^*\,\varphi_\lambda(\vec y)=\delta(\vec  x- \vec y)\qquad\text{and}\qquad\int\ \dd^D\vec x\ \varphi_\lambda(\vec x)^* \,\varphi_{\lambda'}(\vec  x)=\delta(\lambda,\lambda'\,) \ ,
\end{eqnarray}
where the measure $\dd\ell(\lambda)$ is discrete measure on the point spectrum and Lebesgue measure on the absolutely continuous spectrum in $\sigma(\d_{\vec x})\subset\C$ with $\int_{\sigma(\d_{\vec x})}\,\dd\ell(\lambda)\ \delta(\lambda,\lambda'\,)\, f(\lambda) = f(\lambda'\,)$.
Contrary to the functions $\phi^{(\pm)}_\nu$ above, these eigenfunctions need not solve the field equations. Decomposing the propagator into these eigenfunctions gives the formal expansion
\begin{eqnarray}
\Delta(\vec  x,\vec y)=\int_{\sigma(\d_{\vec x})}\,\dd\ell(\lambda)\ \varphi_\lambda(\vec x)^*\, \lambda^{-1}\, \varphi_\lambda(\vec y) \ .
\end{eqnarray}

However, the potential poles at $\lambda=0$ make this definition
problematic, which reflects the existence of more than one propagator
for a given field theory. Usually one modifies the denominator by
adding an adiabatic cutoff $\lambda\rightarrow\lambda+\i\epsilon\,
F(\lambda)$ with small $\epsilon>0$ and a function $F:\sigma(\d_{\vec x})\to \R$
on the spectrum of $\d_{\vec x}$, so that
\begin{eqnarray}
\lambda+\i\epsilon\, F(\lambda)\neq0
\label{eigreg1}
\end{eqnarray}
for all $\lambda\in\sigma(\d_{\vec x})$. A Green's function for the partial differential operator $\d_{\vec x}$ is finally obtained by taking the adiabatic limit $\epsilon\rightarrow0^+$. Equivalently, one can regularize the operator $\d_{\vec x}\rightarrow\d^{(\epsilon)}_{\vec x}$ with $\lim_{\epsilon\rightarrow0^+}\, \d^{(\epsilon)}_{\vec x}=\d_{\vec x}$ and solve the equation $
\d^{(\epsilon)}_{\vec x}\Delta^{(\epsilon)}(\vec x,\vec y)=\delta(\vec x-\vec y)$,
where $\lim_{\epsilon\rightarrow0^+}\, \Delta^{(\epsilon)}(\vec x,\vec y)$ is a Green's function of the original wave operator $\d_{\vec x}$.

Hence any well-defined operator which is continuously connected to the original wave operator and has no zero eigenvalues gives rise to a propagator for $\d_{\vec x}$. However, apart from the requirement of absence of zero eigenvalues of $\d^{(\epsilon)}_{\vec x}$ (or equivalently the condition \eqref{eigreg1}), the regularization is arbitrary and different regularizations may lead to different propagators. For example, in the free Klein-Gordon theory the choice of $F(\vec k)$ as a positive constant leads to the Feynman propagator, while $F(\vec k)=2k^0$ yields the retarded propagator. In general, one cannot be sure whether one obtains the causal propagator unless one compares it by hand to the result obtained from \eqref{propfeyn1}. This is the obvious drawback of the eigenvalue method, and while the problem is easily solved in free field theory, it is still unsolved for the general case of an arbitrary propagator and arbitrary external field; in particular, the equivalence of the propagators in the different representations for generic electromagnetic backgrounds is still an open question. 

For the LSZ model we already encountered the two regularized wave operators $\d_{\vec x}^{(\epsilon)}$ given by the $\vt$-regularization and the $\i\epsilon$-prescription. 
The question of which propagator they lead to in the limit $\epsilon\rightarrow0^+$ has been answered for the $\i\epsilon$-prescription for several related models. For the Klein-Gordon field moving in crossed or parallel uniform electric and magnetic fields, or in an electric field with an additional plane wave background, this method gives the causal propagator~\cite{ritus70,ritus78,batalin84}. Since an additional uniform magnetic background should not change the pole structure of the propagator, the $\i\epsilon$-prescription should also give the causal propagator in the background of a pure electric field. In \S\ref{regprop} we will confirm that the $\vt$-regularization also gives the causal propagator in the background of a uniform electric field along one direction.

\subsection{Causal Propagator in Matrix Space}\label{regprop}

Using the ``sum over solutions method'' \eqref{e2}, the causal propagator for a massive complex scalar field of charge $e$ in four dimensions in the background of a constant electric field along one space direction has been calculated in \cite[eq.~(6.2.40)]{fgs91} with the result
\begin{equation}
\begin{aligned}
\Delta_c(\vec x,\vec y)&=\frac{e\, E}{16\pi^2}\, \e^{- \i e\, \vec x_\parallel\cdot\vec E\vec y_\parallel/2}\, \int_0^\infty\, \frac{\dd s}s\ \frac{1}{\sinh(s\,e\, E)}\\
&\qquad\qquad \times\, \exp\Big(\, - \i s\, \mu^2-\frac{\i}4\, e\,
E\,\|\vec x_\parallel-\vec y_\parallel\|_{\rm M}^2\,
\coth(s\,e\,E)+\i\frac{\|\vec x_\bot-\vec y_\bot\|_{\rm E}^2}{4s}\, \Big) \ .\label{propfgs1}
\end{aligned}
\end{equation}
Here we defined $\vec x=(\vec x_\parallel,\vec x_\bot)\in\R^4$ with $\vec x_\bot\in\R^2$ denoting the two space components perpendicular to the electric field and
\begin{eqnarray}
\vec x_\parallel\cdot\vec E\vec y_\parallel:=E\, x_\parallel^\mu\, \epsilon_{\mu\nu}\, y_\parallel^\nu\ ,
\end{eqnarray}
where $\epsilon_{\mu\nu}$ is the rank two antisymmetric tensor with $\epsilon_{01}=1$ and $E>0$ the electric field strength. Below we will start with this four-dimensional wave operator, with the electric part regularized as in \eqref{dresseddef}, and calculate its (unique) propagator. For $\kappa\rightarrow0^+$ we find exact agreement with \eqref{propfgs1} confirming that this is the causal propagator. The calculations performed here using the matrix basis are comparably simple, so that the matrix basis can be alternatively regarded as a powerful computational tool in ordinary field theory.

We begin with some notation and a preliminary result. We define the symmetric
bilinear form $(-,-)_\vt:\R^2\otimes\R^2\rightarrow\C$ for $\vt\in[-\frac\pi2,\frac\pi2]$ by
\begin{eqnarray}
(\vec x,\vec y)_\vt=\cos(\vt)\,(\vec x,\vec y)_{\rm
  E}+\i\sin(\vt)\,(\vec x,\vec y)_{\rm M} =
\mbox{$\half$}\, \e^{\i\vt}\, \big(x^{(\vt)}_+\,y^{(\vt)}_-+x^{(\vt)}_-\, y^{(\vt)}_+ \big) \ ,
\end{eqnarray}
where $(-,-)_{\rm E}$ is the two-dimensional Euclidean scalar product and $(-,-)_{\rm M}$ the two-dimensional hyperbolic inner product. We also define the map $\|-\|_\vt:\R^2\rightarrow\C$ by
\begin{eqnarray}
\|\vec x\|_\vt^2=(\vec x,\vec x)_\vt= \cos(\vt)\, \|\vec x\|^2_{\rm
  E}+\i\sin(\vt)\, \|\vec x\|^2_{\rm M} = \e^{\i\vt}\,
x^{(\vt)}_+\, x^{(\vt)}_-
\end{eqnarray}
with $\|-\|_{\rm E}$ the two-dimensional Euclidean norm (\ref{Enorm}) and $\|-\|_{\rm M}$ the two-dimensional hyperbolic norm (\ref{Mnorm}). For arbitrary two-dimensional vectors $\vec x,\vec y\in\R^2$ we denote as above $\vec x\cdot\vec E\vec y=E\, x^\mu\, \epsilon_{\mu\nu}\, y^\nu$.
In Appendix~\ref{AppPartialSum} we prove
\begin{lemma}\label{partialsum1}
For any $\vec x,\vec y\in\R^2$ and $a\in\C^*$, one has
\begin{eqnarray}
\sum_{n=0}^\infty\,  f_{mn}^{(E_\vt)}(\vec x)\,f_{nm}^{(E_\vt)}(\vec
y)\, a^n &=&  \frac{E\, a^m}{\pi}\, \exp\Big(-\frac{E}2\, \|\vec x-\vec
y\|_\vt^2+(a-1) \, E\, (\vec x,\vec y)_\vt-a\i\vec x\cdot\vec E\vec
y\Big) \label{partialsum1expl} \\
&& \times\, L_m^0\big(E\, \|\vec x-\vec y\|_\vt^2-a\, (1-a^{-1})^2\, E\, (\vec x,\vec y)_\vt+(a-a^{-1})\i\vec x\cdot\vec E\vec y\big)\ .\notag
\end{eqnarray}
\end{lemma}

Now we determine the propagator of the Klein-Gordon field in four
dimensions coupled to a constant electric field, where the wave
operator parallel to the electric field is given by the
two-dimensional $\vt$-regularized operator
$\left(\Pmm-\mu^2\right)_{\kappa}$. The coordinate vector is again
written as $\vec x=(\vec x_\parallel,\vec x_\bot)\in \R^4$, with $\vec
x_\bot\in\ker(\vec E)$ the components perpendicular to the electric field, and analogously for the momenta $\vec p=(\vec p_\parallel,\vec p_\bot)\in (\R^4)^*$ and derivatives $\p{\mu}=(\partial_\parallel,\partial_\bot)$.

\begin{proposition}
The propagator of the $\vt$-regularized wave operator
$\big(\Kmm-\mu^2\big)_\kappa =\left(\Pmm-\mu^2\right)_{\kappa}+(\i\p{\bot})^2$ coincides in the limit $\kappa\rightarrow0^+$ with the causal propagator \eqref{propfgs1}.
\end{proposition}
\begin{proof}
The inverse of $\big(\Kmm-\mu^2\big)_\kappa$ is given by
\begin{eqnarray}
\Delta^{(\kappa,1)}(\vec x,\vec y)=\bra{\vec x}\frac{1}{\left(\Pmm-\mu^2 \right)_{\kappa}+(\i\p{\bot})^2}\ket{\vec y}\ ,
\end{eqnarray}
where $\big(\Kmm-\mu^2\big)_\kappa=\e^{\i\kappa}\, \P{}^2(\frac\pi2-\kappa)-\e^{-\i\kappa}\, \mu^2+(\i\p{\bot})^2$ with $\kappa>0$ has the eigenvalue equation
\begin{eqnarray}
\big(\Kmm-\mu^2\big)_\kappa \big( f_{mn}^{(E_\vt)}(\vec x_\parallel)\, \e^{-\i\vec p_\bot\cdot\vec x_\bot} \big) =
\Big(4 \i E\,\big(m+\mbox{$\half$}\big)+\|\vec p_\bot\|_{\rm
  E}^2-\e^{-\i\kappa}\, \mu^2\Big) \, f_{mn}^{(E_\vt)}(\vec
x_\parallel)\, \e^{-\i\vec p_\bot\cdot\vec x_\bot} \notag \\
\label{Kmmeigeq}\end{eqnarray}
with $\vt=\frac\pi2-\kappa$. We write $\mu_\kappa^2:=\e^{-\i\kappa}\, \mu^2$ for brevity, remembering that it has a small negative imaginary part. Using the identity
\begin{eqnarray}
\frac{1}a=-\i\int_0^\infty\, \dd s\ \e^{\i s\, a} \qquad \mbox{for} \quad \Im(a)>0\ ,\label{schwpara0}
\end{eqnarray}
we obtain
\begin{eqnarray}
\Delta^{(\kappa,1)}(\vec x,\vec y)&=&-\i\int_0^\infty\, \dd s\ \int\, \frac{\dd^2\vec p_\bot}{(2\pi)^2}\ \sum_{m,n=0}^\infty\, f^{(E_\vt)}_{mn}(\vec x_\parallel)\, f^{(E_\vt)}_{nm}(\vec y_\parallel) \\
&&\qquad \qquad \times\, \e^{-\i s\, \mu_\kappa^2}\, \e^{-4s\, E\,
  (m+\half)}\, \e^{\i s\, \|\vec p_\bot\|_{\rm E}^2-\i(\vec x_\bot-\vec y_\bot)\cdot\vec p_\bot}\ . \notag
\end{eqnarray}
The sum over $n$ is given by Lemma~\ref{partialsum1} with $a=1$, and
the resulting sum over $m$ follows from the identity~\cite[eq.~(48.4.1)]{han75}
\begin{eqnarray}
\e^{-y/2}\, \sum_{m=0}^\infty\,  L^0_m(y)\, t^m = \frac{1}{1-t}\, \exp\Big(\,\frac{y}2\, \frac{t^{1/2}+t^{-1/2}}{t^{1/2}-t^{-1/2}}\,\Big) \qquad \mbox{for} \quad |t|<1
\label{msumHan}\end{eqnarray}
which yields
\begin{eqnarray}
\Delta^{(\kappa,1)}(\vec x,\vec y)&=& -\i\frac{E}{2\pi}\, \e^{-\i\vec
  x_\parallel\cdot\vec E\vec y_\parallel}\, \int_0^\infty\, \dd s\
\frac{\exp\Big(-\i s\, \mu_\kappa^2-\frac{1}2\,
E\, \|\vec x_\parallel-\vec y_\parallel\|^2_\vt\, \coth(2s\,
E)\Big)}{\sinh(2s\, E)}\notag \\ && \qquad \qquad \times \ \int\, \frac{\dd^2\vec
  p_\bot}{(2\pi)^2}\ \e^{\i s\, \|\vec p_\bot\|_{\rm E}^2-\i(\vec x_\bot-\vec y_\bot)\cdot\vec p_\bot}\ .
\end{eqnarray}
The integration over the perpendicular momenta can now be performed by using
\begin{eqnarray}
\int\, \dd p\ \e^{\i s\, p^2-\i(x-y)\, p}=\sqrt{\frac{\i\pi}s}\, \e^{\i(x-y)^2/4s} \ ,
\end{eqnarray}
to get
\begin{eqnarray}
\Delta^{(\kappa,1)}(\vec x,\vec y) &=&\frac{E}{8\pi^2}\, \e^{-\i\vec x_\parallel\cdot\vec E\vec y_\parallel}\, \int_0^\infty\, \frac{\dd s}s\, \frac{1}{\sinh(2s\, E)}\\
&&\qquad\qquad \times\, \exp\Big(\,-\i s\,\mu_\kappa^2-\frac{1}2\,
E\, \|\vec x_\parallel-\vec y_\parallel\|^2_\vt\, \coth(2s\,
E)+\i\frac{\|\vec x_\bot-\vec y_\bot\|_{\rm E}^2}{4s}\, \Big) \ . \notag
\end{eqnarray}
Taking the limit $\kappa\rightarrow0^+$, thus $\vt\rightarrow\frac\pi2$, and substituting $E\rightarrow e\, E/2$ to conform to the conventions of \cite{fgs91}, this result is identical to \eqref{propfgs1}.
\end{proof}

The eigenfunctions for the full regularized wave operator
$\big(\Kmm-\mu^2\big)_\kappa$ factorize into components
perpendicular to the electric field times the eigenfunctions of
$(\Pmm-\mu^2)_\kappa$. Since the eigenvalues of the perpendicular
momentum operators do not produce new pole singularities, we can neglect
them in this calculation and also in the calculation leading to
\eqref{propfgs1}. This result thus easily carries over to the
two-dimensional case confirming that the $\vt$-regularization imposes
causality of the critical LSZ model. We expect that the $\vt$-regularization also leads to the causal propagators for $\sigma\neq1$.

The Schwinger parameter $s>0$ introduced in \eqref{schwpara0} only allows for the regularizations $\vt>0$ and $\mu^2-\i\epsilon$ because of the requirement $\Im(a)>0$, where the latter regularization is normally associated to the Feynman boundary condition on the propagator. The other choices $\vt<0$ and $\mu^2+\i\epsilon$ can be applied by using
\begin{eqnarray}
\frac{1}a=\i\int_{-\infty}^0\, \dd s\ \e^{\i s\, a}\qquad \text{for} \quad \Im(a)<0\ .
\end{eqnarray}
The regularization $\mu^2+\i\epsilon$ is known as the Dyson boundary condition, which leads to an anti-causal propagator where anti-particles travel forward and particles backward in time. This strongly suggests that the regularization $\vt<0$ leads to the Dyson propagator. 

The regularization of the mass $\mu_\kappa^2= \e^{-\i\kappa}\, \mu^2$ is actually irrelevant for the analysis above. Its only function is to provide a continuous interpolation between the hyperbolic and Euclidean wave operators with the help of the parameter $\vt$ alone, without the need to keep track of additional minus signs in front of the mass term. This means that the interpretation in terms of Feynman/Dyson propagators for the cases $\vt\rightarrow\pm \, \frac\pi2$ still holds by regularizing only the operator~$\Pmm$. 

\section{Quantum Duality}\label{quantum1}

In this section we will treat the problem of implementing duality covariance at quantum level for our field theories on Minkowski space. The $\vt$-regularization allows us to regularize the covariant field theories such that the duality is preserved at quantum level. This is done in the same spirit as in \cite{ls02,fs08}, with the $\vt$-regularization now being the only new ingredient. In the following this will be demonstrated for the two-dimensional Grosse-Wulkenhaar model. The more general case of the LSZ model is treated in exactly the same way.

We only need to address how the $\vt$-regularization affects the behaviour of quantities under the duality transformation. The regularized propagator with $\vt=\frac\pi2-\kappa>0$ reads
\begin{eqnarray}
\Delta^{(\kappa)}(\vec x,\vec y)=\bra{\vec x}\big(\mbox{$\half$}\, \Pmm+\mbox{$\half$}\, \Ptm-\mu^2\big)^{-1}_\kappa\ket{\vec y} = \sum_{m,n}\, \frac{f_{mn}^{(E_\vt)}(\vec x)\, f_{nm}^{(E_\vt)}(\vec y)}{2\i E\, \left(m+n+1\right)-\e^{-\i\kappa}\, \mu^2}\ .
\end{eqnarray}
In Appendix~\ref{Fourier} we show that the Fourier transformation of the matrix basis functions is given by
\begin{eqnarray}
\f\big[f_{mn}^{(E_\vt)}\big](\vec k)=f_{nm}^{(1/E_\vt)}(\vec
k)=\frac{(-\i)^{m-n}}{E}\, f_{mn}^{(E_{\vt})}\big(\vec E^{-1}\vec k \big)
\end{eqnarray}
with $\vec E^{-1}\vec k=-E^{-1}\,
(k^1,k^0)$.\footnote{There is a subtle difference here between the
  Euclidean and hyperbolic cases. Contrary to the ordinary Landau
  wavefunctions in Euclidean space, the (unscaled) Fourier transforms
  of the complex Landau wavefunctions have interchanged indices $m\leftrightarrow n$ and a
  reflected regularization parameter $\vt\rightarrow-\vt$. The
  interchange is equivalent to time reversal (or parity), see
  Appendix~\ref{Fourier}. The reflection corresponds to charge
  conjugation, i.e. exchange of particles with anti-particles; this follows from the results of \S\ref{regprop}, where the regularization $\vt>0$ is identified with the Feynman boundary condition and $\vt<0$ with the Dyson boundary condition. The specific rescalings of momenta from $(\R^2)^*$ to $\R^2$ in both cases, which are formally identical but differ by the signature of the metric applied, compensates for this difference.} Since
\begin{equation}
\f\big[\big(\P{}^2(\vt)+\Pt(\vt)\big) f^{(E_\vt)}_{mn} \big](\vec k)=4E_\vt\, \left(m+n+1\right)\,\f\big[f^{(E_\vt)}_{mn} \big](\vec k)\ ,
\end{equation}
we find that Fourier transformation relates the propagator in position space to the momentum space propagator even in the regularized case as
\begin{eqnarray}
(\f\otimes \f)\big[\Delta^{(\kappa)}\big](\vec k,\vec
p)=\frac{1}{E^2}\, \Delta^{(\kappa)}\big(\vec E^{-1}\vec k\,,\, \vec
E^{-1}\vec p \big)\ .\label{lsdualprop}
\end{eqnarray}
This relation just reflects the classical duality covariance for $g=0$.

Analogously to the Euclidean case~\cite{ls02}, the UV/IR-symmetric regularization now amounts to cutting off the matrix element sums at some finite rank $N$ by modifying the regularized position space propagator to
\begin{eqnarray}
\Delta^{(\kappa)}_\Lambda(\vec x,\vec y)=\bra{\vec x}\big(\mbox{$\half$}\, \Pmm+\mbox{$\half$}\, \Ptm-\mu^2\big)^{-1}_\kappa \, L\big(\,\Lambda^{-2}\, \big|\P{}^2(\vt)+\Pt(\vt)\big|\, \big)\ket{\vec y}\ , \label{regprop1}
\end{eqnarray}
where $\Lambda\in\R_+$ is a cutoff parameter, and $L:\R_+\to[0,1]$ a smooth cutoff function which is monotonically decreasing with $L(z)=1$ for $z<1$ and $L(z)=0$ for $z>2$.  We adjust the matrix basis functions so as to diagonalize the regulated Grosse-Wulkenhaar propagator
\begin{eqnarray}
\Delta_{\Lambda|mn;kl}^{(\kappa)}&=&\int\, \dd^2\vec x\ f_{mn}^{\kappa}(\vec x)\, \big(\mbox{$\half$}\, \Pmm+\mbox{$\half$}\, \Ptm-\mu^2\big)^{-1}_\kappa\, L\big(\, \Lambda^{-2}\,\big|\P{}^2(\vt)+\Pt(\vt)\big|\, \big)\, f_{kl}^{\kappa}(\vec x)\notag\\[4pt]
&=&\frac{\delta_{ml}\,\delta_{nk}}{2\i E\, (m+n+1)-\e^{-\i\kappa}\, \mu^2}\, L\big(4 \Lambda^{-2}\, E\, (m+n+1)\big)\ . \label{regprop2}
\end{eqnarray}
The interaction vertices in the matrix space representation are now quite complicated; they are proportional to
\begin{eqnarray}
v^\kappa(m_1,n_1;\dots ;m_4,n_4)= \int\,\dd^2\vec x\ \big(f_{m_1n_1}^{\kappa}\star_\theta f_{m_2n_2}^{\kappa}\star_\theta f_{m_3n_3}^{\kappa}\star_\theta f_{m_4n_4}^{\kappa}\big)(\vec x)\label{intvertex0}
\end{eqnarray}
with $\theta\neq2/E$ in general. Since for $\kappa>0$ the complex Landau wavefunctions $f_{mn}^\kappa$ are elements of the Gel'fand-Shilov spaces $\s_\alpha^\alpha(\R^2)$ with $\alpha\geq\frac12$, which are closed under multiplication of functions with the star-product, the interaction vertex \eqref{intvertex0} is well-defined. 

Feynman diagrams can now be obtained by taking suitable combinations of derivatives of the partition function (\ref{ZJGWmatrix}) with respect to the external sources involving the regularized propagator. Denoting 
\begin{eqnarray}
\Delta_{\Lambda|mn;kl}^{(\kappa)}=:\delta_{mk}\,\delta_{nl}\ \Delta_\Lambda^{(\kappa)}(m,n)\ ,
\end{eqnarray}
they have the schematic form
\begin{eqnarray}
\sum_{n_1,m_1,\ldots,n_K,m_K}\ \prod_{k=1}^K\, \Delta_\Lambda^{(\kappa)}(m_k,n_k)\ (\cdots)\ ,
\end{eqnarray}
where $(\cdots)$ denotes the contributions from products of the noncommutative interaction vertices (\ref{intvertex0}) and combinatorial factors. Since the propagator $\Delta_\Lambda^{(\kappa)}(m,n)$ is nonzero only for $4E\, (m+n+1)<2\Lambda^2$, which at finite $\Lambda$ is only true for a finite number of distinct values of $(m,n)\in\N_0^2$, every Feynman amplitude is represented by a finite sum and thus constitutes well-defined duality covariant Green's functions in the matrix basis; this circumvents the issue of the appropriate test function space for the time being. By multiplying these expressions with $f^{\kappa}_{m_in_i}(\vec x_i)$ for $i=1,\ldots,M$, we get back the position space Green's functions with $M$ external legs by summing over all $m_i,n_i$. They are also well-defined and duality covariant, since they are built from finite sums of well-defined covariant objects. This establishes the quantum duality in Minkowski space for the case $\kappa>0$. 

To prove the duality covariance at $\kappa=0$ in the same manner as
above, one has to ensure that the interaction vertex
(\ref{intvertex0}) away from the self-dual point is well-defined. In
the absence of further analysis, the $\vt$-regularization should be
kept while the matrix cutoff is removed, and all summations and
integrations have been performed. Of course the limit $\Lambda\rightarrow\infty$ can still be ill-defined and may require renormalization; removing this regularization requires a good decay behaviour of the matrix space propagator for large values of its indices, see \S\ref{ren}. In addition, the results from \S\ref{genoscbasis0} are not able to exclude the possibility that even at finite $\kappa>0$ there might be extra divergences at $\Lambda\rightarrow\infty$ if we work in matrix space, stemming from the complex matrix basis itself. This, however, does not affect the duality covariance of the quantum field theory, which has been achieved for the Green's functions in position space through the regularization of the propagators in \eqref{regprop1}. This result is independent of the matrix basis. 

\section{Asymptotic Analysis of Propagators}\label{ren}

One of the most intriguing features of Euclidean duality covariant
field theories is their renormalizability. We will not attempt to
prove here the renormalizability of their Minkowski space
counterparts, but start this program by deriving their propagators in
position and matrix space representations, and studying their asymptotics. We begin by extending the formulas given in \cite{grvt06} to the hyperbolic setting, giving the propagators for the general LSZ models in generic $D=2n$ spacetime dimensions in the position and matrix bases.

\subsection{Position Space Representation\label{proppos}}

In the notations of \S\ref{higher2} and \S\ref{regprop}, the main result from which all causal propagators in Minkowski space and their Euclidean counterparts can be derived is
\begin{proposition}\label{posprop}
The propagator of the regularized LSZ model in {$D=2n$ spacetime dimensions} is given by 
\begin{equation}
\begin{aligned}
\Delta^{(\kappa,\sigma)}(\vec x,\vec y)=& -\i\e^{-\i\vt}\, \frac{E}{2\pi}\, \int_0^\infty\, \dd s\ \frac{\e^{-s\, \mu_\kappa^2}}{\sinh(2s\, E_{-\vt})}\, \exp\Big( -\frac{\sinh(2s\, \tilde E_{-\vt})}{\sinh(2s\, E_{-\vt})}\, \i\vec x_0\cdot\vec E\vec y_0\Big) \\
&\times\, \exp\Big(-\half\, \coth(2s\, E_{-\vt})\, E\, \big(\|\vec x_0\|^2_\vt+\|\vec y_0\|^2_\vt \big)+\frac{\cosh(2s\, \tilde E_{-\vt})}{\sinh(2s\, E_{-\vt})}\, E\, (\vec x_0,\vec y_0)_\vt\Big) \\
&\times\, \prod_{k=1}^{n-1}\, \frac{B_k}{2\pi}\, \frac{1}{\sinh(2s\, B_k)}\, \exp\Big(-\frac{\sinh(2s\, \tilde B_k)}{\sinh(2s\, B_k)}\, \i\vec x_k\cdot\vec B_k\vec y_k\Big) \\
&\times\, \exp\Big(-\half\, \coth(2s\, B_k)\, B_k\, \big(\|\vec x_k\|^2_{\rm E}+\|\vec y_k\|^2_{\rm E} \big)+\frac{\cosh(2s\, \tilde B_k)}{\sinh(2s\, B_k)}\, B_k\, (\vec x_k,\vec y_k)_{\rm E}\Big)
\end{aligned}
\label{propLSZgen}\end{equation}
with $\vt=\frac\pi2-\kappa>0$, $\mu_\kappa^2=\e^{-\i\kappa}\, \mu^2$, $\tilde E=(2\sigma-1)\, E$ and $\tilde B_k=(2\sigma-1)\, B_k$.
\end{proposition}

The proof of Proposition~\ref{posprop} is found in Appendix~\ref{APPpositionprop}. We can now read off the causal propagators for the four-dimensional LSZ and Grosse-Wulkenhaar models. Since $(-,-)_{\pi/2}=\i(-,-)_{\rm M}$ and thus $\|-\|_{\pi/2}^2=\i\|-\|_{\rm M}^2$, one finds
\begin{cor}
The causal propagator of the LSZ model for generic $\sigma\in[0,1]$ in four-dimensional Minkowski space is given by 
\begin{eqnarray}
\Delta^{(0,\sigma)}(\vec x,\vec y)&=&-\frac{\i E\, B}{(2\pi)^2}\, \int_0^\infty\, \dd s\ \frac{\e^{-s\, \mu^2- A_{\rm M}-A_{\rm E}}}{\sin(2s\, E)\, \sinh(2s\, B)} \notag\\
&&\times\, \exp\Big(-\frac{\sin(2s\, \tilde E)}{\sin(2s\, E)}\, \i\vec x_0\cdot\vec E\vec y_0-\frac{\sinh(2s\, \tilde B)}{\sinh(2s\, B)}\, \i\vec x_1\cdot\vec B\vec y_1\Big)
\end{eqnarray}
with
\begin{eqnarray}
A_{\rm M} &=& -\frac E2\, \cot(2s\, E)\, \left(\|\vec x_0\|_{\rm M}^2+\|\vec y_0\|_{\rm M}^2\right)+\frac{\cos(2s\, \tilde E)}{\sin(2s\, E)}\, E\,  (\vec x_0,\vec y_0)_{\rm M}^2 \ , \notag \\[4pt]
A_{\rm E}&=& \frac B2\, \coth(2s\, B)\, \left(\|\vec x_1\|_{\rm E}^2+\|\vec y_1\|_{\rm E}^2\right)-\frac{\cosh(2s\, \tilde B)}{\sinh(2s\, B)}\, B\, (\vec x_1,\vec y_1)_{\rm E}^2\ .
\end{eqnarray}
\end{cor}

\begin{cor}
The causal propagator of the four-dimensional critical LSZ model in Minkowski space is given by
\begin{eqnarray}
\Delta^{(0,1)}(\vec x,\vec y)&=&-\frac{\i E\, B}{(2\pi)^2}\, \e^{-\i\vec x_0 \cdot\vec E\vec y_0-\i\vec x_1 \cdot\vec B\vec y_1}\, \int_0^\infty\, \dd s\ \frac{\e^{-s\,\mu^2}}{\sin(2s\, E)\, \sinh(2s\, B)} \\
&&\times\, \exp\big(\, \mbox{$\half$}\, E\, \|\vec x_0-\vec y_0\|^2_{\rm M}\, \cot(2s\, E)-\mbox{$\half$}\, B\, \|\vec x_1-\vec y_1\|^2_{\rm E}\, \coth(2s\, B)\, \big)\ .\notag
\label{LSZpossigma1}\end{eqnarray}
\end{cor}

\begin{cor}
The causal propagator of the four-dimensional Grosse-Wulkenhaar model in hyperbolic signature is given by 
\begin{eqnarray}
\Delta^{(0)}(\vec x,\vec y)&=& -\frac{\i E\, B}{(2\pi)^2}\, \int_0^\infty\, \dd s\ \frac{\e^{- s\, \mu^2}}{\sin(2s\, E)\, \sinh(2s\, B)} \notag \\
&& \times\, \exp\Big(\, \half\, E\, \cot(2s\, E)\, \left(\|\vec x_0\|_{\rm M}^2+\|\vec y_0\|_{\rm M}^2\right)-\frac{E}{\sin(2s\, E)}\, (\vec x_0,\vec y_0)_{\rm M} \, \Big) \notag \\
&& \times\, \exp\Big(-\half\, B\, \coth(2s\, B)\, \left(\|\vec x_1\|_{\rm E}^2+\|\vec y_1\|_{\rm E}^2\right)+\frac{B}{\sinh(2s\, B)}\, (\vec x_1,\vec y_1)_{\rm E}\Big) \ .
\end{eqnarray}
\label{GWxspacecor}\end{cor}

The Euclidean parts of the propagators here coincide with those found
in~\cite{grvt06} after suitable redefinitions of parameters.

\subsection{Matrix Space Representation}\label{matrixprops}

Below we set $\theta_0=\theta_1=\dots=\theta_{n-1}=: \theta$ for simplicity.
\begin{proposition}\label{thm1}
The matrix space propagator for the $2n$-dimensional regularized LSZ model in Minkowski space is given by
\begin{eqnarray}
\Delta_{\vec m,\vec m+\vec \alpha;\vec l+\vec\alpha,\vec l}^{(\kappa,\sigma)}
&=& -\e^{\i\kappa}\, \frac{\theta}{8\Omega}\, \int_0^1\, \dd s\
s^{-\i\e^{\i\kappa}\, (\sigma\, \alpha_0+\frac12)+\sum_{i=1}^{n-1}\,
  (\sigma\, \alpha_i+ \frac12)-1+\frac{\theta\, \mu^2}{8\Omega}} \label{matrixprop0}\\
&& \qquad \qquad \qquad \qquad \times\,
\Delta_{m_0,m_0+\alpha_0;l_0+\alpha_0,l_0}^{(\kappa)}(s) \
\prod_{i=1}^{n-1}\, \Delta^{\rm E}_{m_i,m_i+\alpha_i;l_i+\alpha_i,l_i}(s) \notag
\end{eqnarray}
with hyperbolic part
\begin{eqnarray}
\Delta_{m,m+\alpha;l+\alpha,l}^{(\kappa)}(s) 
&=&\sum_{u=\max(0,-\alpha)}^{\min(m,l)}\ 
\frac{s^{-\i\e^{\i\kappa}\, u}\,
  \big(1-s^{-\i\e^{\i\kappa}}\,\big)^{m+l-2u}}{
  \Big(1-\frac{(1-\Omega)^2}{(1+\Omega)^2}\,
  s^{-\i\e^{\i\kappa}}\Big)^{\alpha+m+l+1}} \\ && \qquad
\qquad \qquad \times\,
\Big(\, \frac{4\Omega}{(1+\Omega)^2}\, \Big)^{\alpha+2u+1}\,
\Big(\, \frac{1-\Omega}{1+\Omega}\, \Big)^{m+l-2u}\, \mathcal{A}(m,l,\alpha,u) \notag 
\end{eqnarray}
and Euclidean part
\begin{eqnarray}
\Delta^{\rm E}_{m,m+\alpha;l+\alpha,l}(s) 
&=&\sum_{u=\max(0,-\alpha)}^{\min(m,l)}\ \frac{s^{u}\,
  (1-s)^{m+l-2u}}{\Big(1-\frac{(1-\Omega)^2}{
    (1+\Omega)^2}\, s\Big)^{\alpha+m+l+1}} \\ && \qquad \qquad \qquad \times\,
\Big(\, \frac{4\Omega}{(1+\Omega)^2}\, \Big)^{\alpha+2u+1}\,
\Big(\, \frac{1-\Omega}{1+\Omega}\, \Big)^{m+l-2u}\,
\mathcal{A}(m,l,\alpha,u)\ , \notag 
\end{eqnarray}
where
\begin{eqnarray}
\mathcal{A}(m,l,\alpha,u)=\sqrt{{\alpha+m\choose \alpha+u}\,
  {\alpha+l\choose \alpha+u}\, {m\choose u}\, {l\choose u}}
\label{calAfact}\end{eqnarray}
and $\vec\alpha=(\alpha_0,\alpha_1,\ldots,\alpha_{n-1})\in\Z^n$ with
$\alpha_j=n_j-m_j=k_j-l_j$.
\end{proposition}

The proof of Proposition~\ref{thm1} is found in
Appendix~\ref{APPmatrixprop}. The respective special cases, like the
four-dimensional Grosse-Wulkenhaar model, can easily be read off from
this general expression.

\subsection{Power Counting\label{power}}

The {power counting theorem} for general
{non-local} matrix models was proven by Grosse and Wulkenhaar
in~\cite{gw03a}. In the matrix basis, every Feynman diagram of the
duality covariant field theory is represented by a {ribbon
  graph}, whose topology is decisive for the question of whether or not it is
divergent. Power counting in a dynamical matrix model depends
crucially on this topological data. For a \emph{regular} matrix model,
the power counting degree of divergence for an $N$-leg ribbon graph
$G$ of genus $g$ with $V$ vertices and $B$ loops carrying external
legs is given by~\cite{gw03a}
\beq
\omega(G)=D+V\,(D-4)-\mbox{$\frac12$}\, N\,(D-2) -D\,(2g+B-1) \ .
\eeq
Here we briefly recall the role played by
the asymptotic behaviour of the propagator in the derivation of this
power counting theorem. For this, one uses multiscale analysis of the
Schwinger parametric representation of the propagator, which works in
both position and matrix space.

The slicing of the propagator is defined as
\begin{eqnarray}
\Delta=\sum_{i=0}^\infty \, \Delta^i\qquad \text{through} \quad
\int_0^1\, \dd s=\sum_{i=0}^\infty\ \int_{M^{-2i}}^{M^{-2(i-1)}}\,
\dd s
\label{propslice}\end{eqnarray}
with an arbitrary constant $M>1$. This leads to a scale decomposition
of the amplitude ${\mathcal A}_G$ of
any given Feynman graph $G$ as
\beq
{\mathcal A}_G=\sum_I\, {\mathcal A}_G^I \ ,
\eeq
where $I=\{i_\ell\}$ runs through all assignments of positive integers
$i_\ell$ to each line $\ell$ of $G$. 
One then seeks appropriate bounds on the sliced
propagators.

For the $i$-th slice, the main bounds in matrix space are given by~\cite{rvtw05,gmrvt06}
\begin{eqnarray}
\big| \Delta^i_{\vec m\vec n;\vec k\vec l}\big| &\leq& K\,M^{-2i}\, \e^{-c\,
  M^{-2i}\, \|\vec m+\vec n+\vec k+\vec l\|_1} \ , \label{GWbound1}\\[4pt]
\sum_{\vec l}\, \max_{\vec n,\vec k} \, \big| \Delta^i_{\vec m\vec n;\vec
  k\vec l}\big| &\leq& K'\,M^{-2i}\, \e^{-c'\, M^{-2i}\, \|\vec m\|_1} \label{GWbound2}
\end{eqnarray}
for some positive constants $K,K'$ and $c,c'$, where we have introduced the
$\ell^1$-norm $\|\vec m\|_1:=m_0+m_1+ \dots+ m_{n-1}$. Perturbative power
counting amounts to finding which summations cost a factor $M^{2n\, i}$ through~\eqref{GWbound1},
\begin{eqnarray}
\sum_{\vec m\in\N_0^n}\, \e^{-c\, M^{-2i}\, \|\vec
  m\|_1}=\frac{1}{\big(1-\e^{-c\, M^{-2i}}\big)^n}=\frac{M^{2n\, i}}{c^n}\,
\big(1+{O}(M^{-2i})\big)\ ,
\end{eqnarray}
and which cost ${O}(1)$ due to the bound \eqref{GWbound2}. Integrating out loops at higher scales of a graph then gives effective coupling constants in powers of $M$. The
important point is that the faster the propagator decays, the smaller
is the contribution of the integration over internal lines to
effective coupling constants. This in turn reduces the number of
divergent graphs which require renormalization.

As an immediate application of this result in the present context, we can straightforwardly establish
the power counting theorem for the 1+1-dimensional self-dual
Grosse-Wulkenhaar model. In this case the matrix model is local with
propagator $\Delta^{(\kappa)}_{mn;kl}=\delta_{mk}\, \delta_{nl}\,
\Delta^{(\kappa)}(m,n)$ given by
\bea
\Delta^{(\kappa)}(m,n)&=& -\frac\i{2E\,(m+n+1)+\i\e^{-\i\kappa}\,
  \mu^2} \\[4pt] &=& -\i\, \int_0^\infty\, \dd s\ \e^{-2E\,(m+n+1)\,
    s-\i\e^{-\i\kappa}\, \mu^2\, s} \ = \ K\, \int_0^1\, \dd s\ \e^{-2E\,(m+n+1)\,
    s-\i\e^{-\i\kappa}\, \mu^2\, s} \ , \notag
\eea
where $K=\i\big(\e^{-2E\,(m+n+1)-\i\e^{-\i\kappa}\, \mu^2}-1\big)^{-1}$. Slicing this propagator as in (\ref{propslice}), we easily find that the $i$-th slice can be bounded as
\beq
\big| \Delta^{(\kappa)\, i}(m,n)\big| \leq |K|\, M^{-2i}\, \big(M^2-1\big)\, \e^{-2E\,(m+n+1)\,M^{-2i}}\, \e^{-\sin(\kappa)\, \mu^2\, M^{-2i}} \ .
\eeq
The Minkowski space propagator thus has the
requisite exponential decays (\ref{GWbound1})--(\ref{GWbound2}), and
hence the perturbative multiscale renormalization in this case can be
treated exactly as in the Euclidean setting~\cite{rvtw05}; we expect
renormalizability to hold in this case. The case $\Omega<1$ is much
more difficult; in the Euclidean case the coupling $\Omega$ flows very
rapidly to the self-dual point $\Omega=1$, and it would be interesting
to see if this is also the case for the hyperbolic self-dual point.

\subsection{Asymptotics}

As discussed in \S\ref{power}, the asymptotics of the propagators play
an important role in perturbative renormalization. In this paper we
are also interested in determining to what extent the complex matrix
basis is applicable to the perturbative analysis of the duality
covariant field theories; here the asymptotics also give crucial
information. However, the asymptotic behaviours of the hyperbolic parts of the propagators are difficult to investigate due to the oscillatory behaviours of the integrands.

For example, consider the Grosse-Wulkenhaar model in four-dimensional
Minkowski space with propagator given by
Corollary~\ref{GWxspacecor}. Introducing \emph{short variables} $\vec
u_k=\vec x_k-\vec y_k$ and \emph{long variables} $\vec v_k=\vec x_k+\vec y_k$ for $k=0,1$, and using elementary hyperbolic and trigonometric identities, we can write this propagator in the form
\begin{eqnarray}
\Delta^{(0)}(\vec u,\vec v)&=&-\frac{\i B}{(2\pi)^2}\, \int_0^\infty\, \dd s\ \e^{- s\,\mu^2/B}\, \frac{1}{\sin(2s)}\, \frac{1}{\sinh(2s)}\\
&&\qquad \qquad \qquad \qquad \times\, \exp\Big(\,\frac{B}4\,
\cot(s)\, \|\vec u_0\|_{\rm M}^2-\frac{B}4\, \tan(s)\, \|\vec
v_0\|_{\rm M}^2\,\Big) \notag\\
&&\qquad \qquad \qquad \qquad \times\, \exp\Big(-\frac{B}4\,
\coth(s)\, \|\vec u_1\|_{\rm E}^2-\frac{B}4\, \tanh(s)\, \|\vec
v_1\|_{\rm E}^2\Big)\ , \notag
\end{eqnarray}
where we set $E=B$ for simplicity. The integral is sliced in the usual way to get
\begin{eqnarray}
\Delta^{(0)\,j}(\vec u,\vec v)&=&-\frac{\i B}{(2\pi)^2}\, \int_{M^{-2j}}^{M^{-2(j-1)}}\, \dd s\ \e^{- s\,\mu^2/B}\, \frac{1}{\sin(2s)}\, \frac{1}{\sinh(2s)} \label{4Dslice} \\
&&\qquad \qquad \qquad \qquad \qquad \times\, \exp\Big(\,\frac{B}4\,
\cot(s)\, \|\vec u_0\|_{\rm M}^2-\frac{B}4\, \tan(s)\, \|\vec
v_0\|_{\rm M}^2\,\Big) \notag\\
&&\qquad \qquad \qquad \qquad \qquad \times\, \exp\Big(-\frac{B}4\,
\coth(s)\, \|\vec u_1\|_{\rm E}^2-\frac{B}4\, \tanh(s)\, \|\vec
v_1\|_{\rm E}^2\Big) \notag
\end{eqnarray}
with $M>1$.

The Euclidean part of the modulus of the integral (\ref{4Dslice}) can be easily bounded from above by
maximizing each of the hyperbolic functions in the integrand on the interval
$[M^{-2j},M^{-2(j-1)}]$. The factor $\e^{-\frac{B}4\,\tanh(s)\,
  \|\vec v_1\|_{\rm E}^2}$ takes its maximum at $s=M^{-2j}$ where
$\tanh(s)= M^{-2j}-\frac13\, M^{-6j}+O\big((M^{-2j})^5\big) <c'\, M^{-2j}$ for some
constant $c'>0$, while $\e^{-\frac{B}4\, \coth(s)\, \|\vec u_1\|_{\rm
    E}^2}$ takes its maximum value at $s=M^{-2(j-1)}$ with $\coth(s)<
M^{2(j-1)}+M^{-2(j-1)}<c''\, M^{2j}$ and some constant $c''>0$. The
function $\sinh(2s)^{-1}$ can be bounded from above by $M^{2j}$, and
in this way one arrives at the very rough bound
\begin{eqnarray}
\big|\Delta^{(0)\,j}(\vec u,\vec v)\big| &\leq& K\, M^{2j}\, \e^{-c\,
  (M^{2j}\, \|\vec u_1\|_{\rm E}^2+M^{-2j}\, \|\vec v_1\|_{\rm
    E}^2)} \\ && \times\, \int_{M^{-2j}}^{M^{-2(j-1)}}\, \dd s\
\frac{\e^{- s\, \mu^2/B}}{\big|\sin(2s) \big|}\, \exp\Big(\,
\frac{B}4\, \cot(s)\, \|\vec u_0\|_{\rm M}^2-\frac{B}4\, \tan(s)\,
\|\vec v_0\|_{\rm M}^2\,\Big) \notag
\end{eqnarray}
for some positive constants $K$ and $c$. This reproduces the first required
bound which gives exponential decay in both short and long variables
in the Euclidean plane~\cite{gmrvt06}; in particular, integrating over
long Euclidean coordinates costs a factor $M^{2j}$ while short Euclidean
coordinates cost $M^{-2j}$. However, the asymptotic behaviour of the full propagator remains unclear; the hyperbolic part of the integrand is oscillatory, so that more sophisticated methods are needed to bound this integral.

There is a special case in which one can deduce the qualitative
behaviour. The propagator of the critical, regularized massless LSZ model in 1+1
dimensions can be written using Proposition~\ref{posprop} as
\begin{eqnarray}
\Delta_{\mu^2=0}^{(\kappa,1)}(\vec x,\vec y)=-\frac{\i E}{2\pi}\,
\int_0^\infty\, \dd s\ \frac{\e^{-\i\vec x\cdot\vec E\vec
    y}}{\sinh(2s\, E)}\, \exp\Big(-\frac{E}2\, \coth(2s\, E)\, \|\vec
x-\vec y\|^2_\vt\Big) \ ,
\end{eqnarray}
where the integration contour has been rotated as $s\rightarrow s\,
\e^{\i\vt}$. Substituting
\beq
u=\mbox{$\frac12$}\, E\, \|\vec x-\vec y\|_\vt^2\, \big(\coth(2s\, E)-1\big)
\eeq
we get
\begin{eqnarray}
\Delta_{\mu^2=0}^{(\kappa,1)}(\vec x,\vec y)&=& -\frac{\i}{4\pi}\, \e^{-\i\vec
  x\cdot\vec E\vec y}\, \int_0^\infty\, \dd u\
\frac{\e^{-u-\frac{E}2\, \|\vec x-\vec y\|_\vt^2}}{\sqrt{u^2+E\, u\,
    \|\vec x-\vec y\|_\vt^2}} \notag \\[4pt] &=& -\frac{\i}{4\pi}\, \e^{-\i\vec
  x\cdot\vec E\vec y}\, K_0\Big(\, \frac{E}2\, \|\vec x-\vec
y\|_\vt^2\, \Big)\ ,
\label{Deltamu0}\end{eqnarray}
with $K_0(z)$ the modified Bessel function of the second kind of
order~$0$.

This implies that there is still a logarithmic ultraviolet divergence at $\vec x=\vec y$ due to the singular behaviour of $K_0(z)$ at $z=0$. Since~\cite[9.7.2]{as70}
\begin{eqnarray}
K_0(z)= \sqrt{\frac{\pi}{2z}}\, \e^{-z}\,
\left(1+{O}(z^{-1})\right)\qquad \mbox{for} \quad z\to\infty \ , \label{k0}
\end{eqnarray}
this also implies that the propagator $\Delta_{\mu^2=0}^{(\kappa,1)}(\vec
x,\vec y)$ has an asymptotic exponential decay in the {short variable}
$\vec u=\vec x-\vec y$ only for 
\begin{eqnarray}
\Re\left(\|\vec u\|^2_\vt\right)>0\ ,
\end{eqnarray}
and thus only for $|\vt|<\frac\pi2$.
We believe that for $\sigma<1$ the asymptotic exponential
decay in the long variable $\vec v=\vec x+\vec y$ also persists as long as
$|\vt|<\frac\pi2$. We conclude that the propagator has a worse
behaviour in Minkowski space than in Euclidean space, but we can
control its asymptotic behaviour with the help of the parameter
$\vt$. Regarding the restriction $|\vt|<\frac\pi2$ as being part of the
regularization of the field theory, one could then try to carry out
the perturbative multiscale renormalization of the Minkowski space
duality covariant field theory. 

In the matrix space representation there is a similar problem, since
the integrand in \eqref{matrixprop0} is oscillatory. Thus bounding the
magnitude of the integral by an integral over the magnitude of the
integrand possibly produces a big error and might lead to poor
estimates of the asymptotic behaviour. One can use this approximation to show
that the Minkowski space Grosse-Wulkenhaar propagator at
$|\vt|=\frac\pi2$ has an exponential decay in each index separately,
as in (\ref{GWbound1}). To find the other bounds, however, one has to
take care of the oscillating behaviour of the integrand. The
asymptotics of the special case \eqref{Deltamu0} for $|\vt|<\frac\pi2$
raises the hope that the propagators at hand may have such an
asymptotic behaviour in position space so that the matrix basis is
applicable.\footnote{As these propagators are duality covariant, they
  also have a similar decay in momentum space.}

\subsection*{Acknowledgments}

We thank Dorothea Bahns, Lyonell Boulton, Harald Grosse, Edwin Langmann, Jochen
Zahn and Konstantin Zarembo for helpful discussions. The work of RJS is supported in part by grant ST/G000514/1 ``String Theory
Scotland'' from the UK Science and Technology Facilities Council.

\appendix

\section{Proof of Proposition~\ref{genlandautheo}}\label{matrixfunctions}

The complex Landau wavefunctions are built on tensor products of the
complex harmonic oscillator wavefunctions $f_m^{(E_\vt)}$ as
\begin{eqnarray}
f_{mn}^{(E_\vt)}(\vec x)
=\sqrt{\frac{E}{4\pi}}\,\wi{\kb{f_m^{(E_\vt)}}{f_n^{(E_{-\vt})}}}(\vec
x)\ .
\end{eqnarray}
Using \eqref{wignertrafo} and \eqref{genosc1} we get
\begin{eqnarray}
f_{mn}^{(E_\vartheta)}(t,x)&=&\sqrt{\frac{E}{4\pi}}\, \int\, \dd k\
\e^{\i E\,k\,x/2}\,
f_m^{(E_\vartheta)}(t+k/2)\,f_n^{(E_\vartheta)}(t-k/2)\notag\\[4pt]
&=&\sqrt{\frac{E}{4\pi}}\,
\sqrt{\frac{E_\vt}{2\pi}}\,\frac{1}{\sqrt{2^{m+n}\, m!\, n!}} \, \int\, \dd k\
\e^{\i E \,k\,x/2}\, \e^{-\frac14 \,
  E_\vt\,\left((t+k/2)^2+(t-k/2)^2\right)} \notag\\
&&\times\, H_m\big(\sqrt{E_\vt/2}\,(t+k/2)\big)\,
H_n\big(\sqrt{E_\vt/2}\,(t-k/2)\big)\ .
\end{eqnarray}
The generating function for the Hermite polynomials
\begin{eqnarray}
\e^{-a^2\, (\xi^2-2\xi\, q)}=\sum_{m=0}^\infty\, \frac{1}{m!}\, (a\,
\xi)^m\, H_m(a\, q)
\end{eqnarray}
is used to obtain the generating function for the complex matrix
basis functions as
\begin{eqnarray}
K^{(E_\vartheta)}(\xi,\eta;t,x)&:=&\sqrt{\frac{4\pi}{E}}\,
\sum_{m,n=0}^\infty\, \sqrt{\frac{2^{m+n}}{m!\, n!}}\,
\big(\sqrt{E_\vt/2}\, \xi\big)^m\, \big(\sqrt{E_\vt/2}\, \eta\big)^n\,
f_{mn}^{(E_\vartheta)}(t,x)\notag\\[4pt]
&=&\sqrt{\frac{E_\vt}{2\pi}}\, \int\, \dd k\ \e^{\i E\,k\,x/2}\,
\e^{-\frac14\, E_\vt\,\left((t+k/2)^2+(t-k/2)^2\right)} \notag\\
&&\qquad \qquad \qquad \times\, \e^{-\half\,
  E_\vt\, \left(\xi^2-2\xi\, (t+k/2)+\eta^2-2\eta\, (t-k/2)\right)} \notag\\[4pt] &=&
2\e^{\half\, E_\vt\, (-x^{(\vt)}_+\, x^{(\vt)}_-+2\xi \,
    x^{(\vt)}_-+2\eta \,x^{(\vt)}_+-2\eta\, \xi)} \ ,
\end{eqnarray}
where we used the complex light-cone coordinates
(\ref{genLCcoords}). 
The complex matrix basis functions can now be obtained by taking
suitable derivatives with respect to the variables $\xi$ and $\eta$ to get
\begin{eqnarray}
f_{mn}^{(E_\vartheta)}(t,x)&=&\sqrt{\frac{E_\vt}{4\pi}}\,
\frac{1}{\sqrt{m!\, n!}}\, \left(\frac{1}{E_\vt }\right)^{\frac{m+n}2}\,
\left.\pdm{m}{\xi}\,
  \pdm{n}{\eta}K^{(E_\vartheta)}(\xi,\eta;t,x)\right|_{\xi=\eta=0}
\notag\\[4pt]
&=&\sqrt{\frac{E}{\pi}}\, \sqrt{m!\, n!}\, (E_\vt )^{\frac{m-n}2}\,
\e^{-E_\vt \, x^{(\vt)}_+\, x^{(\vt)}_-/2}\, \big(x^{(\vt)}_-\big)^{m-n}\notag\\
&&\times\, \sum_{p=0}^n\, \big(E_\vt\, x^{(\vt)}_+\, x^{(\vt)}_-\big)^{n-p}\,
\frac{(-1)^{p}}{(m-p)!\, (n-p)!\, p!}\ ,
\end{eqnarray}
where we assumed $m\geq n$. This last sum can be identified with an associated Laguerre polynomial 
\begin{eqnarray}
L_n^k(z)=\sum_{q=0}^n\, \frac{(n+k)!}{(n-q)!\, (k+q)!\, q!}\ (-z)^q
\end{eqnarray}
by shifting $p\rightarrow q=n-p$. We finally arrive at
\begin{eqnarray}
f_{mn}^{(E_\vartheta)}(t,x)= (-1)^n\, \sqrt{\frac{E}{\pi}}\,
\sqrt{\frac{n!}{m!}}\, \left(E_\vt \right)^{\frac{m-n}2}\, \e^{-E_\vt
  \, x^{(\vt)}_+\, x^{(\vt)}_-/2}\, \big(x^{(\vt)}_-\big)^{m-n}\, L_n^{m-n}\big(E_\vt\, x^{(\vt)}_+\,
  x^{(\vt)}_-\big) \ .
\end{eqnarray}
An identical calculation for $n\geq m$ leads to the same result with
$+\leftrightarrow -$ and $m\leftrightarrow n$, yielding~(\ref{genlandau0}).

\section{One-Loop Effective Action in Matrix Space}\label{effaction}

We will now reconstruct a classic result in quantum electrodynamics
using the $\vt$-regularization and the complex matrix basis. In his seminal paper~\cite{Schw51} Schwinger
calculated the effective action for both a Dirac field and a Klein-Gordon field of charge $e$ in a uniform external electromagnetic background in four spacetime dimensions. In a pure electric field $E$ the one-loop correction for the Klein-Gordon theory (before charge renormalization) is given by the Lagrangian
\begin{eqnarray}
\mathcal{L}^{(1)}=\frac{1}{16\pi^2 }\, \int_0^\infty\, \dd s\ s^{-3}\,
\e^{-\mu^2\, s}\, \Big(\,\frac{e\, E\, s}{\sin(e\, E\, s)}-1\, \Big) \
.\label{efflagkg}
\end{eqnarray}
By deforming the contour of integration above the real axis one picks
up the poles at $s=s_n=n\, \pi/e\, E$ for $n\in\N$ by the residue theorem; this leads
to the famous formula for the probability per unit time and unit
volume $2\, \Im\big(\mathcal{L}^{(1)}\big)$ to create a particle-antiparticle pair in the scalar field
theory. We will now show that the regularized matrix basis approach leads to the same result quite effortlessly. We work throughout in the notations of \S\ref{cont}.

The generating functional for connected graphs $W[J,J^*]$ is defined
via the vacuum-to-vacuum amplitude (\ref{vactovacampl}) in the
presence of the external sources $J$ and $J^*$ as
\begin{eqnarray}
W[J,J^*]=-\i \log Z_0[J,J^*] \ .
\end{eqnarray}
It can be
expressed in terms of the causal propagator (\ref{propfeyn1}) as~\cite{Schw51}
\begin{eqnarray}
W[J,J^*]=\int \, \dd^4\vec x \ \int\, \dd^4\vec y\ J^*(\vec x)\, \Delta_c(\vec x,\vec y)\, J(\vec y)-\i\log\det\left( \Delta_F^{-1}\, \Delta_c \right)\ ,
\end{eqnarray}
with $\Delta_F=\Delta_c\big|_{E=0}$ the usual Feynman propagator. By using the $\vt$-regularization we can write
\begin{eqnarray}
W[J,J^*]=\int \, \dd^4\vec x \ \int\, \dd^4\vec y\ J^*(\vec x)\, \Delta^{(\kappa,1)}(\vec x,\vec y)\, J(\vec y)-\i\log \det\Big(\, \frac{-\p{\mu}^2-\mu^2}{\Kmm-\mu^2}\, \Big)_{\kappa}\ ,
\end{eqnarray}
which is understood in the limit $\kappa\rightarrow0^+$ with 
\begin{eqnarray}
\Big(\, \frac{-\p{\mu}^2-\mu^2}{\Kmm-\mu^2}\, \Big)_{\kappa}= \frac{- \p{\mu}^2-\e^{-\i\kappa}\, \mu^2}{\e^{\i\kappa}\, \P{}^2\big(\mbox{$\frac\pi2$}-\kappa \big)-\e^{-\i\kappa}\, \mu^2+(\i\p{\bot})^2} \ .
\end{eqnarray}

The effective action is now defined as the Legendre transformation
\begin{eqnarray}
\Gamma[\phi_{\rm cl},\phi^*_{\rm cl}]= W[J,J^*]-\int\, \dd^4\vec x\ J(\vec x)\, \phi_{\rm cl}^*(\vec x)-\int\,\dd^4\vec x\ J^*(\vec x)\, \phi_{\rm cl}(\vec x) \label{eff1}
\end{eqnarray}
of $W[J,J^*]$ with respect to the ``classical'' fields $\phi_{\rm cl}(\vec x)$ and $\phi^*_{\rm cl}(\vec x)$ defined by
\begin{equation}
\begin{aligned}
\phi_{\rm cl}(\vec x)&=\frac{\delta W[J,J^*]}{\delta J^*(\vec x)}=\int\, \dd^4\vec y\ \Delta^{(\kappa,1)}(\vec x,\vec y)\, J(\vec y) \ , \\[4pt]
\phi_{\rm cl}^*(\vec x)&=\frac{\delta W[J,J^*]}{\delta J(\vec x)}=\int\, \dd^4\vec y\ J^*(\vec y)\, \Delta^{(\kappa,1)}(\vec y,\vec x) \ .
\end{aligned}
\end{equation}
These equations may be inverted to give
\begin{eqnarray}
J(\vec x)=(\Kmm-\mu^2)_{\kappa}\phi_{\rm cl}(\vec x)\qquad\text{and}\qquad J^*(\vec x)=-(\Kmm-\mu^2)_{\kappa}\phi^*_{\rm cl}(\vec x)\ ,
\end{eqnarray}
and inserting this into \eqref{eff1} yields
\begin{eqnarray}
\Gamma[\phi_{\rm cl},\phi_{\rm cl}^*]&=&- \int\,\dd^4\vec x\ \int\,
\dd^4\vec y\ (\Kmm-\mu^2)_{\kappa}\phi^*_{\rm cl}(\vec x)\,
\Delta^{(\kappa,1)}(\vec x,\vec y)\,
(\Kmm-\mu^2)_{\kappa}\phi_{\rm cl}(\vec y) \notag\\ && -\, \i\log \det\Big(\, \frac{-\p{\mu}^2-\mu^2}{\Kmm-\mu^2}\, \Big)_{\kappa}\notag\\
&&-\, \int\, \dd^4\vec x\ \big( (\Kmm-\mu^2)_{\kappa}\phi_{\rm cl}(\vec x)\big)\, \phi_{\rm cl}^*(\vec x)+\int\, \dd^4\vec x\ \big((\Kmm-\mu^2)_{\kappa}\phi_{\rm cl}^*(\vec x)\big)\, \phi_{\rm cl}(\vec x) \notag\\[4pt]
&=&\so[\phi_{\rm cl},\phi_{\rm cl}^*]+\i\log\det\Big(\,
\frac{\Kmm-\mu^2}{-\p{\mu}^2-\mu^2}\, \Big)_{\kappa} \ .
\end{eqnarray}
This is the \emph{full} effective action of the quantum field theory;
the quantum mechanical content is completely captured by the one-loop correction
\begin{eqnarray}
\i\log\det\Big(\, \frac{\Kmm-\mu^2}{-\p{\mu}^2-\mu^2}\,
\Big)_{\kappa}=W[0,0]\ .
\end{eqnarray}
We define $W[0,0] =:\int\, \dd^4\vec x\ \mathcal{L}^{(1)}(\vec x)$,
with $\mathcal L^{(1)}$ the one-loop effective Lagrangian. The
probability that no pair gets produced out of the vacuum is given by
$\big|\bk{0,{\rm out}}{0,{\rm in}}\big|^2=\e^{-2\, \Im (W[0,0])}$.

The effective action is given by
\begin{eqnarray}
W[0,0]= \i\, \text{Tr}\, \log\Big(\, \frac{\Kmm-\mu^2}{-\p{\mu}^2-\mu^2}\,
\Big)_{\kappa}\ ,
\end{eqnarray}
and the eigenvalue equation for the operator
$\big(\Kmm-\mu^2\big)_\kappa$ is given by (\ref{Kmmeigeq}). We
adhere to Schwinger's convention by substituting $E\rightarrow e\, E/2$. Using the identity
\begin{eqnarray}
\log\left(\frac{a}b\right)=\int_0^\infty\, \frac{\dd s}s\ \big(\e^{\i
  s\, a}-\e^{\i s\, b}\, \big)
\end{eqnarray}
which is valid for $\Im(a)>0$ and $\Im(b)>0$, the effective Lagrangian can be obtained through
\begin{eqnarray}
\mathcal{L}^{(1)}(\vec x)&=&\i\bra{\vec x}\log\Big(\,
\frac{\Kmm-\mu^2}{-\p{\mu}^2-\mu^2}\, \Big)_\kappa\ket{\vec x}
\\[4pt]
&=&\i\int_0^\infty \, \frac{\dd s}s\ \int\frac{\dd^2\vec
  p_\bot}{(2\pi)^2}\ \e^{-\i s\, \mu_\kappa^2}\, \e^{\i s\, \|\vec p_\bot\|_{\rm
    E}^2}
\notag\\ && \qquad \qquad \times\, \Big(\, \sum_{m,n=0}^\infty\,
f^{(E_\vt)}_{nm}(\vec x_\parallel)\, f^{(E_\vt)}_{mn}(\vec x_\parallel)\, \e^{-2s\, e\,
  E\, (m+\half)}-\int\, \frac{\dd^2\vec p_\parallel}{(2\pi)^2}\ \e^{\i
  s\, \|\vec p_\parallel\|_{\rm M}^2}\, \Big) \ .\notag
\end{eqnarray}
The integration over parallel momenta gives $\frac1{4\pi\, s}$.

We can now use Lemma~\ref{partialsum1} with $\vec x=\vec y$ and $a=1$ to obtain
\begin{eqnarray}
\mathcal{L}^{(1)}(\vec x)&=&\i\int_0^\infty\, \frac{\dd s}s\ \int\,
\frac{\dd^2\vec p_\bot}{(2\pi)^2}\ \e^{-\i s\, \mu_\kappa^2}\,
\Big(\, \frac{e\, E}{2\pi}\, \e^{-s\,e\, E}\ \sum_{m=0}^\infty\,
\e^{-2s\, e\, E\, m}-\frac{1}{4\pi\, s}\, \Big)\, \e^{\i s\, \|\vec
  p_\bot\|_{\rm E}^2}\notag\\[4pt]
&=&\frac{1}{16\pi^2}\, \int_0^\infty\, \frac{\dd s}{s^2}\ \e^{-\i s\,
  \mu_\kappa^2}\, \Big(\, \frac{e\, E}{\sinh(e\, E\,
  s)}-\frac{1}{s}\, \Big)\ ,
\end{eqnarray}
which is independent of $\vec x$. The integral converges near
infinity since $\mu_\kappa^2$ has a small imaginary part, and near $0$
due to the $\frac1s$ subtraction of the free scalar propagator. By
rotating the integration contour as $s\to-\i s$, and taking the
limit $\kappa\to0^+$, this Lagrangian coincides with Schwinger's
result \eqref{efflagkg}. The case of four-dimensional Dirac fields can
be treated in the same way, by transforming the spinor propagator to the
scalar propagator; see~\cite[App.~F]{thesis} for details. This analysis again
exemplifies the fact that the matrix basis provides an easy way of
doing otherwise cumbersome calculations in quantum electrodynamics. 

\section{Proof of Lemma~\ref{partialsum1}}\label{AppPartialSum}

For $m\geq n$ the explicit expression for the first eigenfunctions
on the left-hand side of (\ref{partialsum1expl}) is
\begin{eqnarray}
f_{mn}^{(E_\vartheta)}(\vec x)= (-1)^{n}\, \sqrt{\frac{E}{\pi}}\,
\sqrt{\frac{n!}{m!}}\, \e^{-E_\vt \, x^{(\vt)}_+\, x^{(\vt)}_-/2}\,
\big(\sqrt{E_\vt}\,x^{(\vt)}_-\big)^{m-n}\,L_{n}^{m-n}\big(E_\vt\,x^{(\vt)}_+\,
  x^{(\vt)}_-\big) \ ,
\end{eqnarray}
while the second eigenfunctions have a similar representation
\begin{eqnarray}
f_{nm}^{(E_\vartheta)}(\vec y)= (-1)^{n}\, \sqrt{\frac{E}{\pi}}\,
\sqrt{\frac{n!}{m!}}\, \e^{-E_\vt \, y^{(\vt)}_+\, y^{(\vt)}_-/2}\,
\big(\sqrt{E_\vt}\, y^{(\vt)}_+\big)^{m-n}\,L_{n}^{m-n}\big(E_\vt\,y^{(\vt)}_+\,
  y^{(\vt)}_-\big)\ ,
\end{eqnarray}
with the notations (\ref{Evartheta}) and (\ref{genLCcoords}). These representations can also be used for $n>m$ due to the identity
\begin{eqnarray}
(-1)^n\, r^{m-n}\, L_n^{m-n}(r^2)=(-1)^m\, r^{n-m}\, \frac{m!}{n!}\,
L_m^{n-m}(r^2)\ .
\end{eqnarray}
The sum over $n$ thus has the form
\begin{eqnarray}
\sum_{n=0}^\infty\, f_{mn}^{(E_\vt)}(\vec x)\,f_{nm}^{(E_\vt)}(\vec
y)\, a^n& =& \frac{E}{\pi}\, \frac{\big(E_\vt \, x^{(\vt)}_-\,
    x^{(\vt)}_+\big)^m}{m!}\, \e^{-E_\vt\,( x^{(\vt)}_+\, x^{(\vt)}_-+y^{(\vt)}_+\, y^{(\vt)}_-)/2} \\
&& \times\, \sum_{n=0}^\infty\,  n!\, \Big(\, \frac{a}{E_\vt \, x^{(\vt)}_-\,
  y^{(\vt)}_+}\, \Big)^{n}\, L_n^{m-n}\big(E_\vt\,x^{(\vt)}_+\,
x^{(\vt)}_- \big)\,
L_n^{m-n}\big(E_\vt\,y^{(\vt)}_+\, y^{(\vt)}_- \big)\ . \notag
\end{eqnarray}
It can be done explicitly by using the identity~\cite[48.23.11]{han75}
\begin{eqnarray}
\sum_{n=0}^\infty\, n!\, c^n\, L_n^{m-n}(\xi)\, L_n^{k-n}(\eta)=k!\,
\e^{c\,\xi\,\eta}\, (1-\eta\,c)^{m-k}\, c^m\,
L_k^{m-k}\big((1-\xi\,c)\, (\eta\,c-1)/c\big) 
\end{eqnarray}
with $k=m$, $\xi=E_\vt\,x^{(\vt)}_+\, x^{(\vt)}_-$, $\eta=E_\vt\,y^{(\vt)}_+\, y^{(\vt)}_-$, and
$c=a/E_\vt\,x^{(\vt)}_-\, y^{(\vt)}_+$. This yields
\begin{eqnarray}
\sum_{n=0}^\infty \, f_{mn}^{(E_\vt)}(\vec x)\,f_{nm}^{(E_\vt)}(\vec
y)\, a^n=\frac{E}{\pi}\, \e^{-(\xi+\eta)/2}\, \e^{c\, \xi\, \eta}\,
a^m\, L^0_m\left(\eta+\xi-c\, \xi\, \eta-c^{-1}\right)\ ,
\end{eqnarray}
which after some elementary algebra gives (\ref{partialsum1expl}).

\section{Matrix Basis in Momentum Space}\label{Fourier}

The complex Landau wavefunctions have special symmetries which will be
useful in analysing their Fourier transforms.
\begin{proposition}
The complex Landau wavefunctions satisfy the relations
\begin{eqnarray}
f_{mn}^{(E_\vt)}\big(E^{-1}\, t\,,\,E^{-1}\, x\big)&=&E\,
f_{mn}^{(1/E_{-\vt})}(t,x) \ , \label{sym1}\\[4pt]
f_{mn}^{(E_\vt)}(-t,x)&=&(-1)^{m-n}\, f_{nm}^{(E_\vt)}(t,x) \ , \label{sym2}\\[4pt]
f_{mn}^{(E_\vt)}(t,-x)&=&f_{nm}^{(E_\vt)}(t,x) \ , \label{sym3}\\[4pt]
f_{mn}^{(E_\vt)}(x,t)&=&(-\i)^{m-n}\, f_{nm}^{(E_{-\vt})}(t,x) \ . \label{sym4}
\end{eqnarray}
\end{proposition}
\begin{proof}
The relation \eqref{sym1} follows directly from the explicit
expression \eqref{genlandau0} by noting that $E$ and $x_\pm^{(\vt)}$
occur only in the combinations $\sqrt{E}\, x_\pm^{(\vt)}$ and $E\,
x_+^{(\vt)}\, x_-^{(\vt)}$. Time reversal $t\rightarrow-t$ only
affects the term involving $x_{-{\rm sgn}(m-n)}^{(\vt)}\rightarrow -x_{{\rm
    sgn}(m-n)}^{(\vt)}=-x_{-{\rm sgn}(n-m)}^{(\vt)}$, which gives
\eqref{sym2}. Parity $x\rightarrow-x$ sends
$x_{-{\rm sgn}(m-n)}^{(\vt)}\rightarrow x_{{\rm
    sgn}(m-n)}^{(\vt)}=x_{-{\rm sgn}(n-m)}^{(\vt)}$, which shows
\eqref{sym3}. Under interchange of $t$ and $x$, we find
$x_\pm^{(\vt)}\rightarrow \pm\i\e^{-\i\vt}\, x_\mp^{(-\vt)}$, and thus
$\sqrt{E_{\vt}}\, x_\pm^{(\vt)}\rightarrow\sqrt{E_{-\vt}}\, \big(\pm\i
x_\mp^{(-\vt)} \big)$ and 
$E_\vt \, x_+^{(\vt)}\, x_-^{(\vt)}\rightarrow E_{-\vt}\,
x_+^{(-\vt)}\, x_-^{(-\vt)}$. Putting these transformations into \eqref{genlandau0} proves \eqref{sym4}.
\end{proof}

\begin{proposition}
The Fourier transformation of the complex Landau wavefunction $f_{mn}^{(E_\vt)}(\vec x)$ is given by
\begin{eqnarray}
\f\big[f_{mn}^{(E_\vt)}\big](\vec k)=f_{nm}^{(1/E_\vt)}(\vec
k)=\frac{(-\i)^{m-n}}E\, f_{mn}^{(E_{\vt})}\big(\vec E^{-1}\vec k \big)
\label{Fourierfmn}\end{eqnarray}
with $\vec E^{-1}\vec k=-E^{-1}\, (k^1,k^0)$.
\end{proposition}
\begin{proof}
Denote momentum space derivatives as
$\hat{\partial}_\mu:=\frac\partial{\partial k^\mu}$. Using the
explicit forms of the hyperbolic and Euclidean space wave operators
given in \S\ref{Formulation}, in Fourier space we find that these operators have the form
\begin{eqnarray}
\frac1{2\pi}\, \int\,\dd^2\vec x\ (\Pmm\phi)(\vec x)\, \e^{-\i\vec k\cdot\vec x}
&=& \frac1{2\pi}\, \int\, \dd^2\vec x\ \phi(\vec x)\, \Ptm\e^{-\i\vec k\cdot\vec x}\\[4pt]
&=&\big((k^2_0-k^2_1)+2\i E\, (k^0\, \hat\partial^1-k^1\,
\hat\partial^0)+E^2\, (\hat\partial_{0}^2-\hat\partial_{1}^2)\big) \f[\phi](\vec k) \notag
\end{eqnarray}
and
\begin{eqnarray}
\frac1{2\pi}\, \int\,\dd^2\vec x\ (\Pii\phi)(\vec x)\, \e^{-\i\vec k\cdot\vec x}
&=& \frac1{2\pi}\, \int\, \dd^2\vec x\ \phi(\vec x)\,\Pti\e^{-\i\vec k\cdot\vec x}\\[4pt]
&=&\big((k^2_0+k^2_1)+2\i E\, (k^0\, \hat\partial^1+k^1\,
\hat\partial^0)-E^2\,
(\hat\partial_{0}^2+\hat\partial_{1}^2)\big)\f[\phi](\vec k) \ . \notag
\end{eqnarray}
From the explicit forms of the regularized wave operators
(\ref{ptheta3}), this gives
\begin{eqnarray}
\f\big[\P{}^2(\vt)\phi\big](\vec k)&=&\e^{\i\vt}\,E^2\,
\big(-(\e^{\i\vt}\, \hat\partial_{0}^2+\e^{-\i\vt}\,
\hat\partial_{1}^2)+2\i E^{-1}\, (\e^{\i\vt}\,
k^1\, \hat\partial^0+\e^{-\i\vt}\, k^0\, \hat\partial^1)\notag\\ && \qquad
\qquad +\, (\e^{-\i\vt}\, k_0^2+\e^{\i\vt}\, k_1^2)\big)\f[\phi](\vec
k) \notag\\[4pt]
&=&\e^{2\i\vt}\, E^2\, \tilde{\mathcal P}{}^2(-\vt)\f[\phi](\vec k)\ ,
\end{eqnarray}
where the differential operator $\tilde{\mathcal P}{}^2(-\vt)$ has the same form as
$\Pt(-\vt)$ with the substitutions $\p{\mu}\rightarrow\hat\partial_\mu$, $x^\mu\rightarrow k^\mu$ and $E\rightarrow E^{-1}$. On the other hand, by substituting $\phi=f_{mn}^{(E_\vt)}$ we find
\begin{eqnarray}
\f\big[\P{}^2(\vt)f_{mn}^{(E_\vt)}\big](\vec k)=4E_\vt\,
\big(m+\mbox{$\half$}\big)\, \f\big[f_{mn}^{(E_\vt)}\big](\vec k)
\end{eqnarray}
and thus
\begin{eqnarray}
\tilde{\mathcal P}{}^2(-\vt)\f\big[f_{mn}^{(E_\vt)}\big](\vec
k)=4E_\vt^{-1}\, \big(m+\mbox{$\half$}\big)\, \f\big[f_{mn}^{(E_\vt)}
\big](\vec k)\ .
\end{eqnarray}
By Parseval's theorem the Fourier transforms of the matrix basis
functions have the same normalization as the position space wavefunctions, from which
we conclude the first equality of (\ref{Fourierfmn}). The second
equality of (\ref{Fourierfmn}) follows from the symmetry relations \eqref{sym1}--\eqref{sym4}.
\end{proof}

As a simple application of this result, we can establish that the
Feynman propagator for the free Klein-Gordon theory in the complex matrix
basis possesses the same mass-shell singularities as in momentum space.
\begin{lemma}
The Feynman propagator in the complex matrix basis is given
by
\beq
\big(\Delta_{F}^{(\kappa)}\big)_{mn;kl}:= \big(G^{(\kappa)\,-1
}\big)_{mn;kl} =\int\, \dd^2\vec k\ \frac{f_{mn}^{(1/E_\vt)}(\vec k)\, f_{kl}^{(1/E_\vt)}(\vec k)}{-\|\vec k\|_{\rm M}^2+\mu_\kappa^2} \ ,
\eeq
with
\beq
G^{(\kappa)}_{mn;kl}=\big\langle
f_{nm}^{(E_{-\vt})}\big|\big(\partial_\mu^2+\mu^2\big)_\kappa \big|f_{kl}^{(E_\vt)}
\big\rangle
\eeq
and $\mu_\kappa^2:=\e^{-\i\kappa}\,\mu^2$ for
$\vt=\frac\pi2-\kappa >0$.
\label{massshelllem}\end{lemma}
\begin{proof}
We simply relate the Klein-Gordon operator in the different basis sets. One has
\beq
\big(\partial_\mu^2+\mu^2\big)_\kappa \delta(\vec x-\vec
y)=\langle\vec x|\big( \partial_\mu^2+\mu^2\big)_\kappa |\vec
y\rangle = \sum_{n,m,k,l}\, f_{mn}^{(E_\vt)}(\vec x)\,
G^{(\kappa)}_{mn;kl}\, f_{lk}^{(E_\vt)}(\vec y) \ ,
\eeq
and thus
\bea
G^{(\kappa)}_{mn;kl}&=& \int\, \dd^2\vec x\ \int\, \dd^2\vec y\ f_{mn}^{(E_\vt)}(\vec x)\, \langle \vec x|\big(\partial_\mu^2+\mu^2\big)_\kappa|\vec y\rangle\, f_{kl}^{(E_\vt)}(\vec y) \notag \\[4pt] &=& \int\, \frac{\dd^2\vec k}{(2\pi)^2}\ \int\, \dd^2\vec x \ \int\, \dd^2\vec y\ f_{mn}^{(E_\vt)}(\vec x)\, \e^{\i \vec k\cdot\vec x}\, \big(-\|\vec k\|_{\rm M}^2+\mu^2_\kappa\big)\, \e^{-\i \vec k\cdot \vec y}\, f_{kl}^{(E_\vt)}(\vec y) \notag\\[4pt] &=& \int\, \dd^2\vec k\ \f\big[f_{nm}^{(E_{-\vt})}\big](\vec k)^*\, \big(-\|\vec k\|_{\rm M}^2+\mu^2_\kappa\big)\, \f\big[f_{kl}^{(E_\vt)}\big](\vec k) \ .
\eea
It follows that the Fourier transforms of the functions $f_{mn}^{(E_\vt)}$ diagonalize $G^{(\kappa)}_{mn;kl}$, and the result now follows from (\ref{Fourierfmn}).
\end{proof}

\section{Proof of Proposition~\ref{posprop}}\label{APPpositionprop}

The propagator is given by
\begin{eqnarray}
\Delta^{(\kappa,\sigma)}(\vec x,\vec y)&=&\bra{\vec x}\big(\sigma\,
\e^{\i\kappa}\, \K^2(\vt)+\tilde\sigma\,
\e^{\i\kappa}\,\Kt(\vt)-\e^{-\i\kappa}\, \mu^2\big)^{-1}\ket{\vec
  y} \\[4pt]
&=&\e^{-\i\kappa}\, \bra{\vec x}\Big( \sigma\,
\P{}^2(\vt)+\tilde\sigma\, \Pt(\vt)+\e^{2\i\vt}\,
\sum_{k=2}^{n}\, \big(\sigma\, (\P{}^2_{i})_{k}+\tilde\sigma\,
(\Pt_{i})_{k} \big)+\e^{2\i\vt}\, \mu^2\Big)^{-1}\ket{\vec y} \notag
\end{eqnarray}
where $\vt=\frac\pi2-\kappa>0$ and we have set $\tilde\sigma=1-\sigma$. The (regularized) wave operators have the eigenvalue equations
\begin{equation}
\begin{aligned}
\big(\sigma\, \P{}^2(\vt)+\tilde\sigma\,
\Pt(\vt)\big)f_{m_0n_0}^{(E_\vt)}(\vec x_0)&=4 E_\vt\, \big(\sigma\,
m_0+\tilde\sigma\, n_0+\mbox{$\half$}\big)\, f_{m_0n_0}^{(E_\vt)}(\vec
x_0) \ , \\[4pt]
\big(\sigma\, (\Pii)_{k+1} +\tilde\sigma\, (\Pti)_{k+1}
\big)f^{(B_k)}_{m_kn_k}(\vec x_k)&=4B_k\, \big(\sigma \,
m_k+\tilde\sigma\, n_k+\mbox{$\half$} \big)\, f_{m_kn_k}^{(B_k)}(\vec
x_k) \ ,
\end{aligned}
\end{equation}
with $f_{m_kn_k}^{(B_k)}(\vec x_k)$ the usual Landau wavefunctions and
$B_k\in\R_+$ for $k=1,\dots,n-1$. Using the identity
\begin{eqnarray}
a^{-1}=\int_0^\infty\, \dd s\ \e^{-s\, a}
\end{eqnarray}
which is valid for $\Re(a)>0$, we find
\begin{eqnarray}
\Delta^{(\kappa,\sigma)}(\vec x,\vec y)
&=&-\i\e^{-\i\vt}\, \int_0^\infty\, \dd s\ \e^{- s\, \mu_\kappa^2}\ \sum_{m_0,n_0=0}^\infty \,
f_{m_0n_0}^{(E_\vt)}(\vec x_0)\, f_{n_0m_0}^{(E_\vt)}(\vec y_0)\,
\e^{-4s\, E_{-\vt}\, (\sigma\, m_0+\tilde\sigma\, n_0+\frac12)}\notag\\
&&\qquad \qquad \times\, \prod_{k=1}^{n-1}\, \Big(\,
\sum_{m_k,n_k=0}^\infty\, f_{m_kn_k}^{(B_k)}(\vec x_k)\,
f_{n_km_k}^{(B_k)}(\vec y_k)\, \e^{-4s\,B_k\, (\sigma \,
  m_k+\tilde\sigma \, n_k+\frac12)}\, \Big) \ .
\end{eqnarray}
By Lemma~\ref{partialsum1} the sum over $n_0$ gives
\begin{eqnarray}
&&\frac{E}{\pi}\, \sum_{m_0=0}^\infty\,  \e^{-4s\, E_{-\vt}\, (m_0+\frac12)}\\
&&\qquad \times\, \exp\Big(-\frac{E}2\, \|\vec x_0-\vec
y_0\|_\vt^2+\big(\e^{-4s\, E_{-\vt}\, \tilde\sigma }-1 \big)\, E\,
(\vec x_0,\vec y_0)_\vt-\e^{-4s\, E_{-\vt}\, \tilde\sigma }\,\i\vec
x_0\cdot\vec E\vec y_0\Big) \notag\\
&&\qquad \times\, L^0_{m_0}\left(E\, \|\vec x_0-\vec
  y_0\|_\vt^2-4\sinh^2(2s\, E_{-\vt}\, \tilde\sigma )\, E\, (\vec
  x_0,\vec y_0)_\vt-2\sinh(4s\, E_{-\vt}\, \tilde\sigma )\,\i\vec
  x_0\cdot\vec E\vec y_0\right)\ . \notag
\end{eqnarray}
The sum over $m_0$ can be performed by using the identity
(\ref{msumHan}) with $t=\e^{-4 s\, E_{-\vt}}$ to get
\begin{eqnarray}
&&\frac{E}{2\pi\, \sinh(2s\, E_{-\vt})}\, \exp\left(-\frac{\cosh(2s\,
    E_{-\vt})}{2\sinh(2s\, E_{-\vt})}\, E\, \|\vec x_0-\vec y_0\|_\vt^2\right.\\
&&\qquad \qquad \qquad \qquad \qquad \qquad +\, \Big(\, \e^{-4s\,
  E_{-\vt}\, \tilde\sigma}-1 +2\e^{-2s\, E_{-\vt}}\,
\frac{\sinh^2(2s\, E_{-\vt}\, \tilde\sigma)}{\sinh(2s\, E_{-\vt})}\,
\Big)\, E\, (\vec x_0,\vec y_0)_\vt\notag\\
&&\qquad \qquad \qquad \qquad \qquad \qquad \left.+\,\Big(\,
  -\e^{-4s\, E_{-\vt}\, \tilde\sigma}+\e^{-2s\, E_{-\vt}}\,
  \frac{\sinh(4s\, E_{-\vt}\, \tilde\sigma)}{\sinh(2s\, E_{-\vt})}\,
  \Big) \,\i\vec x_0\cdot\vec E\vec y_0\right) \ . \notag
\end{eqnarray}
By using elementary hyperbolic identities, the terms proportional to $(\vec x_0,\vec y_0)_\vt$ can be simplified to
\begin{eqnarray}
\e^{-4s\, E_{-\vt}\, \tilde\sigma}-1+2\e^{-2s\, E_{-\vt}}\,
\frac{\sinh^2(2s\, E_{-\vt}\, \tilde\sigma)}{\sinh(2s\, E_{-\vt})} =
\frac{\cosh(2s\, \tilde E_{-\vt})}{\sinh(2s\,
  E_{-\vt})}-\frac{\cosh(2s\, E_{-\vt})}{\sinh(2s\,E_{-\vt})}
\end{eqnarray}
where we defined $\tilde
E_{-\vt}:=(1-2\tilde\sigma)\, E_{-\vt}=(2\sigma-1)\,
E_{-\vt}$. Likewise, 
the terms proportional to $\i\vec x_0\cdot\vec E\vec y_0$ can be rearranged to
\begin{eqnarray}
-\e^{-4s\, E_{-\vt}\, \tilde\sigma}+\e^{-2s\, E_{-\vt}}\,
\frac{\sinh(4s\, E_{-\vt}\, \tilde\sigma)}{\sinh(2s\,
  E_{-\vt})}=-\frac{\sinh(2s\, \tilde E_{-\vt})}{\sinh(2s\,
  E_{-\vt})}\ .
\end{eqnarray}
The triangle relation $\|\vec x_0-\vec
y_0\|^2_\vt=\|\vec x_0\|^2_\vt+\|\vec y_0\|^2_\vt-2(\vec x_0,\vec y_0)_\vt$
allows us to combine further terms. The sums over $n_k$ and $m_k$ for
$k=1,\dots,n-1$ are
treated in exactly the same way, and putting everything together we
finally get (\ref{propLSZgen}).

\section{Proof of Proposition~\ref{thm1}}\label{APPmatrixprop}

The $2n$-dimensional regularized LSZ wave operator in the matrix basis
is given by \eqref{2nlsz}--\eqref{G0} with $\theta_j=\theta$
and $\mathcal D^{(\sigma)}_{mn;kl}:=\mathcal D^{j\, (\sigma)}_{mn;kl}$
for $j=0,1,\dots,n-1$. Each of these operators have non-vanishing
matrix elements only for
\begin{eqnarray}
n_j-m_j=k_j-l_j=:\alpha_j \qquad \mbox{for} \quad j=0,1,\ldots,n-1 \ .
\end{eqnarray}
This is due to the $SO(1,1)\times SO(2)^{n-1}$ symmetry of the action. We can thus eliminate $n$ components and write instead
\begin{eqnarray}
D_{\vec m,\vec m+\vec \alpha;\vec l+\vec\alpha,\vec
  l}^{(\kappa,\sigma)}=\i\mathcal
D^{(\sigma)}_{m_0,m_0+\alpha_0;l_0+\alpha_0,l_0}-\e^{-\i\kappa}\,
\sum_{i=1}^{n-1}\, \mathcal
D^{(\sigma)}_{m_i,m_i+\alpha_i;l_i+\alpha_i,l_i}-\e^{-\i\kappa}\,
\mu^2\ \delta_{\vec m\vec l}
\label{DSOsym}\end{eqnarray}
with $\vec\alpha\in\Z^n$.

The $n$ components of the wave operator
(\ref{DSOsym}) are independent and its eigenvectors are therefore products
of the eigenvectors of the individual matrices. The mass term is
already diagonal and so are the terms proportional to
$\tilde\Omega$. Thus for every $\alpha\in\Z$ we seek solutions of the
eigenvalue equations
\begin{eqnarray}
\sum_{l=0}^\infty\, \mathcal D^{(1/2)}_{m,m+\alpha;l+\alpha, l} U_{ l
  v}^{(\alpha)}= v\,U_{m v}^{(\alpha)} \label{GWeig2}\ .
\end{eqnarray}
This equation has been solved in~\cite{gw05}. The eigenvectors are given
by
\begin{eqnarray}
U_{mv}^{(\alpha)}=\sqrt{ {\alpha+m\choose m}\, {\alpha+y\choose y} }
\, \Big(\, \frac{2\, \sqrt{\Omega}}{1+\Omega}\, \Big)^{\alpha+1}\,
\Big(\, \frac{1-\Omega}{1+\Omega}\, \Big)^{m+y}\,
{}_2F_1\left(\begin{array}{c}-m,-y\\ 1+\alpha\end{array}\left | \ -\frac{4\Omega}{(1-\Omega)^2}\right)\right. \label{U}
\end{eqnarray}
and the eigenvalues are
\begin{eqnarray}
v=\frac{4\Omega}{\theta}\, \big(2y+\alpha+1 \big)
\label{vyrel}\end{eqnarray}
for $y\in \N_0$. As expected, this is the usual harmonic oscillator
spectrum. The hypergeometric function ${}_2F_1$ appearing in
(\ref{U}) with negative integer values in its first two arguments is
an orthogonal Meixner polynomial. In particular, $U_{mv}^{(\alpha)}$
is symmetric in its lower indices. 

For the full wave matrix the
addition of the $\tilde\Omega$-term modifies the eigenvalues $v\to v'$ with
\begin{eqnarray}
v'= \frac{4\Omega}{\theta}\, \big(2y+2\sigma\, \alpha+1\big)\ .
\end{eqnarray}
The complete matrix space wave operator in $D=2n$ dimensions has the representation
\begin{eqnarray}
D_{\vec m,\vec m+\vec \alpha;\vec l+\vec \alpha,\vec
  l}^{(\kappa,\sigma)}=\sum_{\vec v}\, U_{\vec m\vec
  v}^{(\vec\alpha)}\, \Big(\i v_0'-\e^{-\i\kappa}\,
\sum_{i=1}^{n-1}\, v'_i-\e^{-\i\kappa}\, \mu^2\, \Big)\,
\big(U^{(\vec\alpha)\,-1}\big)_{\vec l\vec v} \ ,
\end{eqnarray}
where 
\begin{eqnarray}
U_{\vec m\vec v}^{(\vec\alpha)}=\prod_{j=0}^{n-1}\, U_{m_jv_j}^{(\alpha_j)}
\end{eqnarray}
and
\begin{eqnarray}
\i v'_0-\e^{-\i\kappa}\, \sum_{i=1}^{n-1}\,
v'_i-\e^{-\i\kappa}\, \mu^2 &=&\frac{8\Omega}{\theta}\, \Big(\, \i
y_0+\i \big(\sigma\, \alpha_0+\mbox{$\frac12$}\big)-\e^{-\i\kappa}\,
\mu^2\, \frac{\theta}{8\Omega} \notag\\ && \qquad \qquad -\,
\e^{-\i\kappa}\, \sum_{i=1}^{n-1}\, y_i-\e^{-\i\kappa}\,
\sum_{i=1}^{n-1}\, \big(\sigma\, \alpha_i+\mbox{$\frac12$}\big)\, \Big) \label{GWeig3}
\end{eqnarray}
with $y_j\in\N_0$. From the orthogonality relations for the Meixner
polynomials it follows that 
\begin{eqnarray}
\big(U^{(\vec\alpha)\,-1}\big)_{\vec m\vec v} = U_{\vec m\vec
  v}^{(\vec\alpha)} \ .
\end{eqnarray}

In the following we will use the notation
$U_{mv}^{(\alpha)}=U_{m}^{(\alpha)}(y)$ where $v$ and $y$ are related by \eqref{vyrel}. Using the Schwinger parametrization this yields the propagator
\begin{eqnarray}
\Delta_{\vec m,\vec m+\vec \alpha;\vec l+\vec\alpha,\vec
  l}^{(\kappa,\sigma)} &=& \sum_{y_0,y_1,\dots,y_{n-1}=0}^\infty\, \Big(\i v_0'-\e^{-\i\kappa}\,
\sum_{i=1}^{n-1}\, v'_i-\e^{-\i\kappa}\, \mu^2\, \Big)^{-1}\
\prod_{j=0}^{n-1}\, \left(U_{m_j}^{(\alpha_j)}(y_j)\,
  U_{l_j}^{(\alpha_{j})}(y_j)\right)\notag\\[4pt] &=&
-\e^{\i\kappa}\, \frac{\theta}{8\Omega}\,
\int_0^\infty\, \dd t\ \e^{\i t\, \e^{\i\kappa}\, (\sigma\,
  \alpha_0+\frac12) -t\, \sum_{i=1}^{n-1}\, (\sigma\,
  \alpha_i+\frac12)-t\, \frac{\theta\, \mu^2}{8\Omega}}\\
&&\times\, \Big(\, \sum_{y_0=0}^\infty\, \e^{\i t\, \e^{\i\kappa}\,
  y_0}\, U_{n_0}^{(\alpha_0)}(y_0)\, U_{ l_0}^{(\alpha_0)}(y_0)\,
\Big)\ \prod_{i=1}^{n-1}\, \Big(\, \sum_{y_i=0}^\infty\, \e^{-t\,
  y_i}\, U_{m_i}^{(\alpha_i)}(y_i)\, U_{ l_i}^{(\alpha_{i})}(y_i)\,
\Big) \ . \notag
\end{eqnarray}
The sum over $y_0$ can be performed by using the explicit formula for
the eigenvectors \eqref{U}, and the hypergeometric identity~\cite{gw05}
\begin{eqnarray}
&&\sum_{y=0}^\infty\, {\alpha+y\choose y} \,
{}_2F_1\left.\left(\begin{array}{c}-m,-y\\1+\alpha\end{array}\right|w\right)\,
{}_2F_1\left.\left(\begin{array}{c}-l,-y\\1+\alpha\end{array}\right|w\right)
\ z^y\label{ff} \\
&& \qquad \ =\ \frac{\big(1-(1-w)\,
  z\big)^{m+l}}{(1-z)^{\alpha+m+l+1}}\,
{}_2F_1\left(\begin{array}{c}-m,-l\\1+\alpha\end{array}\left| \ \frac{z\,
    w^2}{\big(1-(1-w)\, z\big)^2}\right)\right. \qquad \mbox{for}
\quad|z|<1 \notag 
\end{eqnarray}
with $z=\e^{\i t\, \e^{\i\kappa}}\,(1-\Omega)^2\, (1+\Omega)^{-2}$
and $w=-4\Omega\, (1-\Omega)^{-2}$.

After some algebra this leads to
\begin{eqnarray}
&& \sum_{y_0=0}^\infty\, \e^{\i t\, \e^{\i\kappa}\, y_0}\,
U_{m_0}^{(\alpha_0)}(y_0)\, U_{ l_0}^{(\alpha_{0})}(y_0) \\ && \qquad
\qquad \ =\ \frac{\big(1-\e^{\i t\, \e^{\i\kappa}}\big)^{m_0+
    l_0}}{\left(1-\frac{\e^{\i t\, \e^{\i\kappa}}\,
      (1-\Omega)^2}{(1+\Omega)^2}\right)^{\alpha_0+m_0+ l_0+1}}\,
\sqrt{{\alpha_0+m_0\choose m_0}\, {\alpha_0+ l_0\choose l_0}}\notag\\
&&\qquad \qquad \qquad \qquad \qquad \times\,{}_2F_1\left(\begin{array}{c}-m_0,-
      l_0\\1+\alpha_0\end{array}\left| \ \Big(\,
  \frac{4\Omega}{(1+\Omega)^2}\, \Big)^2\, \Big(\,
  \frac{1+\Omega}{1-\Omega}\, \Big)^2\, \frac{\e^{\i t\,
      \e^{\i\kappa}}}{\big(1-\e^{\i t\, \e^{\i\kappa}}
    \big)^2}\right)\right. \ .\notag
\end{eqnarray}
Now we substitute $s=\e^{-t}$ (with Jacobian $s^{-1}$) and use the expansion of the hypergeo\-metric functions
\begin{eqnarray}
&&{}_2F_1\left.\left(\begin{array}{c}-m,-l\\
      1+\alpha\end{array}\right|z\right)= 
\sum_{u=\max(0,-\alpha)}^{\min(m,
  l)}\, \frac{m!\, l!\,\alpha!}{(m-u)!\,( l-u)!\,(\alpha+u)!\,
  u!}\ z^u \ .
\end{eqnarray}
After a bit of algebra the various factorial terms can be recombined into the quantity
(\ref{calAfact}), and we find
\begin{eqnarray}
&& \sum_{y_0=0}^\infty\, \e^{\i t\, \e^{\i\kappa}\, y_0}\,
U_{m_0}^{(\alpha_0)}(y_0)\, U_{ l_0}^{(\alpha_{0})}(y_0)
\label{minkpart} \\ && \qquad
\qquad \ =\ 
\sum_{u_0=\max(0,-\alpha_0)}^{\min(m_0, l_0)}\,
\frac{s^{-\i\e^{\i\kappa}\, u_0}\,
  \big(1-s^{-\i\e^{\i\kappa}}\big)^{m_0+
    l_0-2u_0}}{\left(1-\frac{(1-\Omega)^2}{(1+\Omega)^2}\,
    z^{-\i\e^{\i\kappa}}\right)^{\alpha_0+m_0+ l_0+1}} \notag\\ &&
\qquad \qquad \qquad \qquad \qquad \times\, \Big(\,
\frac{4\Omega}{(1+\Omega)^2}\, \Big)^{\alpha_0+2u_0+1} \,
\Big(\,
\frac{1-\Omega}{1+\Omega}\, \Big)^{m_0+ l_0-2u_0}\, \mathcal{A}(m_0,
l_0,\alpha_0,u_0) \ . \notag
\end{eqnarray}

The sums over $y_i$ for $i=1,\ldots,n-1$ are performed in a completely
analogous way. The only difference between the Euclidean and
hyperbolic parts of the propagator is the additional factor $-\i\e^{\i\kappa}$ in the
exponential of $y_0$. This simply changes $\i t\, \e^{\i\kappa}\to
-t$ and $s^{-\i\e^{\i\kappa}}\to s$ everywhere in the above
derivation, and we arrive finally at the expression
(\ref{matrixprop0}).

\bibliographystyle{alpha}
\bibliography{references}

\end{document}